\documentclass[12pt,draftcls,onecolumn]{IEEEtran}
\usepackage{amsmath,amssymb,bbm,color,psfrag,amsfonts}
\usepackage{dsfont}
\usepackage{graphicx}                       % \includegraphics
\usepackage[normalem]{ulem}    
\usepackage{hyperref}
\usepackage{wrapfig}
\usepackage{pdfpages}
\usepackage{accents}
\usepackage{comment}
\usepackage{url}

\usepackage[shortlabels]{enumitem}

\usepackage{booktabs}  
\usepackage{multirow}
\usepackage[normalem]{ulem}
\usepackage[flushleft]{threeparttable}
\usepackage{diagbox}
\usepackage{tabularx} % for 'tabularx' environment

\usepackage[font={small, sf}, labelfont=bf]{caption}
\usepackage{subcaption}

\usepackage{cite}

\usepackage[ruled,lined,linesnumbered]{algorithm2e} 
   \SetKwInOut{Input}{input}\SetKwInOut{Output}{output}
   \DontPrintSemicolon

\newtheorem{theorem}{Theorem}
\newtheorem{definition}{Definition}

\newtheorem{lemma}{Lemma}
\newtheorem{proposition}{Proposition}
\newtheorem{remark}{Remark}
\newtheorem{example}{Example}
\newtheorem{corollary}{Corollary}

\newenvironment{proof}{\begin{IEEEproof}}{\end{IEEEproof}}

\usepackage{color}

\newcommand{\real}{{\mbox{\bf R}}}
\newcommand{\preal}{{\mbox{\bf R}}_{\ge 0}}    % positive real
\newcommand{\spreal}{{\mbox{\bf R}}_{>0}}    % strict positive real
\newcommand{\unit}{{\mbox{\bf e}}}    % strict negative real

\newcommand{\zerobf}{\mathbf{0}}
\newcommand{\onebf}{\mathbf{1}}

\newcommand{\until}[1]{[#1]}

\newcommand{\trans}[1]{#1^{\mathrm{T}}}
\newcommand{\mc}{\mathcal}

\newcommand{\Drelone}{\tilde{\mc D}}
\newcommand{\Dreltwo}{\hat{\mc D}}

\newcommand{\starobj}{g} 
\newcommand{\starobjin}{g^{\mathrm{in}}}
\newcommand{\starobjout}{g^{\mathrm{out}}}
\newcommand{\roofun}{\chi}

\newcommand{\comset}{\Theta}
\newcommand{\sdsign}{s}
\newcommand{\pnetval}{R}  % supportable value for parallel networks
\newcommand{\clink}{o} % critical link index for parallel networks
\newcommand{\cmax}{\mathrm{end}} % critical link index for parallel networks

\newcommand{\range}{{\mathcal R}}

\DeclareMathOperator{\diag}{diag}

\DeclareMathOperator{\conv}{conv} % convex hull

\DeclareMathOperator{\aff}{aff}
\DeclareMathOperator{\cube}{cube}
\DeclareMathOperator{\proj}{proj}    		% projection 
\DeclareMathOperator{\swp}{sweep}    	% sweep
\newcommand{\argmax}{\mathop{\rm argmax}}

\newcommand{\sign}{\mathop{\bf sign}}
\newcommand{\cl}[1]{\mathop{\bf cl}(#1)}

\DeclareMathOperator{\starop}{\Omega}

\newcommand{\setdef}[2]{\left\{#1 \; | \; #2\right\}}
\newcommand{\map}[3]{#1: #2 \rightarrow #3}
\newcommand{\longversion}[1]{}

\usepackage{etoolbox}
\newtoggle{long}
\toggletrue{long}
\usepackage{tikz}
\usepackage{pgfplots}

\usetikzlibrary{positioning}
\usetikzlibrary{shapes, arrows.meta}

\tikzset{ dot/.style ={circle, draw, inner sep=2pt},  
	     supply dot/.style ={circle, draw, inner sep=2pt, fill = blue!40}, 
	     demand dot/.style={circle, draw, inner sep=2pt, fill = green!60},
              main node/.style={circle,draw,font=\sffamily\bfseries\scriptsize, thick, minimum width=16pt}, 
              supply node/.style={circle,draw,font=\sffamily\bfseries\scriptsize, fill=blue!60, thick}, 
              demand node/.style={circle,draw,font=\sffamily\bfseries\scriptsize, fill=green!60, thick},
              edge label/.style={font=\sffamily\small},
              main edge/.style={-Stealth, thick, auto}, 
              infeasible node/.style={circle,minimum size=0.6cm, inner sep=0pt, draw, fill= red!50}, 
              feasible node/.style={circle,minimum size=0.6cm, inner sep=0pt, draw, fill= blue!50},
              face/.style={circle,draw,inner sep=1pt, fill=black}}

\usepackage[mode=buildnew]{standalone}

\graphicspath{{./fig/fig/}{./fig/tikz/}}

\title{Computing Optimal Control of Cascading Failure in DC Networks}
\author{Qin Ba\thanks{The authors are with the Sonny Astani Department of Civil and Environmental Engineering at the University of Southern California, Los Angeles, CA. \texttt{\{qba,ksavla\}@usc.edu}. They were supported in part by NSF CAREER ECCS Project No. 1454729.} \quad Ketan Savla}

\date{March 20, 2018}							% Activate to display a given date or no date

\begin{document}
\maketitle

\begin{abstract}
We consider discrete-time dynamics, for cascading failure in DC networks, whose map is composition of failure rule with control actions. Supply-demand at the nodes is monotonically non-increasing under admissible control. Under the failure rule, a link is removed permanently if its flow exceeds capacity constraints. We consider finite horizon optimal control to steer the network from an arbitrary initial state, defined in terms of active link set and supply-demand at the nodes, to a feasible state, i.e., a state which is invariant under the failure rule. 
There is no running cost and the reward associated with a feasible terminal state is the associated cumulative supply-demand. We propose two approaches for computing optimal control\longversion{, and provide time complexity analysis for these approaches}. The first approach, geared towards tree reducible networks, decomposes the global problem into a system of coupled local problems, which can be solved to optimality in two iterations. 
\longversion{In the first iteration, optimal solutions to the local problems are computed, from leaf nodes to the root node, in terms of the coupling variables. In the second iteration, in the reverse order, the local optimal solutions are instantiated with specific values of the coupling variables.} 
When restricted to the class of one-shot control actions, the optimal solutions to the local problems possess a piecewise affine property, which facilitates analytical solution. 
% develops a network decomposition which can be implemented in two iterations for tree reducible networks. In the first iteration, starting from leaves, every node computes, via local optimization, local control action parameterized by the flow on link to its parent node. The first iteration ends with root node computing a specific local control action which is broadcast to its children nodes, and thereafter in the second iteration, all the nodes sequentially compute their respect local control actions.
%%In the special case of a unit time horizon, the entire procedure is shown to be equivalent to a system of equations which can be solved sequentially. 
%In some special cases, the time complexity of the algorithm is shown to be linear with respect to the number of nodes in the network.  
The second approach computes optimal control by searching over the reachable set, which is shown to admit an equivalent finite representation by aggregation of control actions leading to the same reachable active link set. An algorithmic procedure to construct this representation is provided by leveraging and extending tools for arrangement of hyperplanes and polytopes. Illustrative simulations, including showing the effectiveness of a projection-based approximation algorithm, are also presented.     

%Simulations to illustrate monotonicity of the optimal cost with the time horizon on an IEEE benchmark network are also included. 

%\kscomment{
%We consider the problem of optimal online load shedding control over multiple rounds to prevent cascading failure in DC power networks. The state space consists of discrete and continuous variables, corresponding to link active status and demand-supply at the nodes, respectively. In order to facilitate solution by value iteration, we formulate an equivalent state aggregation approach. The key idea in aggregation is partitioning of feasible control actions in terms of possible network topologies. 
%
%We propose a generic branch and bound algorithm to solve the problem, and present a set of tools to reduce computational complexity of the algorithm. In particular, we adapt results on rank one perturbations to pseudo-inverse of Laplacian matrices to compute flow redistribution under link failure, and introduce monotonicity properties under which the set of control policies requiring consideration reduces to a considerably small set. These tools are illustrated in the context of proportional load shedding control and simple network topologies. }
\end{abstract}

\section{Introduction}
\label{sec:introduction}

Cascading failure in physical networks can be modeled via discrete-time dynamics, where the time epochs correspond to component failures. The map of the dynamical system is described in terms of composition of a failure rule with a control policy. A common failure rule is permanent removal of a link from the network if its physical flow exceeds capacity. Analysis of such dynamics under a given control policy has attracted considerable attention, primarily through simulations, e.g., see \cite{Cohen.Erez.ea:00,Watts:01,Motter:2002fk,Crucitti.Latora.ea:04,Barrat.Barthelemy.ea:08}. However, control \emph{design} is relatively less well understood, e.g., see \cite{Bienstock:11} and our previous work in \cite{Savla.Como.ea.TNSse14} for few such examples. In this paper, we consider such an optimal control problem for power networks. 

The network state is described in terms of active links, i.e., links which have not been removed so far, and the external power injection/withdrawal, also referred to as \emph{supply-demand}, at the nodes. Under the failure rule, at a given network state, links are permanently removed if their power flow exceeds thermal capacity constraint. The control actions correspond to changing supply-demand at the nodes. A network state is called \emph{feasible} if it is invariant under the failure rule, and is called \emph{infeasible} otherwise. We are interested in designing control actions to steer the network from an arbitrary initial state to a terminal feasible state within a given finite time horizon. In this paper, admissible control actions are those under which the magnitude of supply-demand at the nodes is non-increasing, and we consider the setting in which there is no running cost and the cost associated with a terminal feasible state is equal to the negative of cumulative supply-demand associated with that state. We use DC approximation for power flow for tractability, in line with standard practice when multiple power flow computations are involved, e.g., see \cite{Bienstock:11, bienstock2010nk, bernstein2014power}. 

%DC power flow model to approximate AC model, in order to account for the capacity constraint and address the fundamental combinatorial complexity underlying the multi-stage optimal control problem. This is a standard practice when multiple power flow computations are required, e.g., see \cite{salmeron2004analysis, Bienstock:11, bienstock2010nk, bernstein2014power}. 
%Generalizations of failure rules, control actions, and cost are also discussed. 

The optimal control problem studied in this paper was formulated in \cite{Bienstock:11,Bienstock:16}, where the focus was on low-complexity control policies. To the best of our knowledge, a formal framework for computing optimal control beyond these low-complexity policies is lacking in the literature. The objective of this paper is to develop rigorous approaches to address this shortcoming. This is however a challenging task. The hybrid state space prevents straightforward application of standard optimal control and dynamic programming tools. Furthermore, it is not possible to find a natural ordering in the state space due to non-monotonicity of power flow. Non-monotonicity here refers to counterintuitive behaviors, reminiscent of Braess's paradox \cite{braess1968paradoxon}, under which removal of links can make an infeasible power network feasible and arbitrary load shedding can make a feasible network infeasible, e.g., see \cite{Ba.Savla.CDC16, bienstock2010nk, Lai.Low:Allerton13, guo2017monotonicity}.

We distinguish our work with network interdiction problems, e.g., see \cite{salmeron2004analysis, pinar2010optimization, bienstock2010nk}. The latter is a static problem to find the smallest set of links whose removal causes a severe blackout. The solutions are based on well-known mixed-integer programming techniques, such as Benders' decomposition and bilevel programming. On the other hand, we do not allow control of links, and consider a multistage framework induced by the cascading dynamics. As already mentioned, standard control synthesis methods do not apply straightforwardly to our setting.

We provide two approaches for computing optimal control.
%\ksmargin{the directionality issue is a bit confusing}
The first approach is geared towards tree reducible networks, i.e., networks which can be reduced to a tree by recursively replacing subnetworks between supply-demand nodes with links. 
%\qbcomment{subnetworks between two nodes, and containing no supply or demand nodes in the interior, with single links}. 
For such networks, we decompose the (global) optimal control problem into coupled local problems associated with nodes in the tree corresponding to the reduced network, which can be solved to optimality in two iterations. In the first iteration, from leaves to the root, every node solves the local problem as a function of the local coupling variable, which corresponds to outflow from that node. In the second iteration, in the reverse order from the root to the leaves, the local optimal solutions are instantiated with specific values of the coupling variables. When restricted to control actions which shed load only at $t=0$, the local problems, in spite of non-convexity, possess a piecewise affine property, which facilitates analytical solution. 

The second approach computes optimal control by \emph{searching} for an optimal feasible terminal state among the states reachable from the initial condition. This search is made possible by an \emph{equivalent} partition of the one-step reachable set from a network state into a \emph{finite} number of aggregated states, with each corresponding to the same reachable active link set. 
%showing that the reachable set admits an equivalent finite representation. The key is that 
%the one-step reachable set from any network state can be partitioned into a finite number of aggregated states, with each corresponding to the same reachable active link set. 
These partitions are determined by admissibility constraints for control actions (to maintain monotonicity of supply-demand at the nodes), and the link failure rules. Linearity of these constraints allows us to leverage and extend tools from the domain of \emph{arrangement of hyperplanes} e.g., see \cite{edelsbrunner1987algorithms} \cite[Chapter 24]{toth2004handbook}, and \emph{polytopes}, e.g., see \cite{ziegler2012lectures} \cite{grunbaum1967convex}, to construct these partitions.

In summary, the paper makes several contributions towards computing optimal control of cascading failure in power networks. First, we cast the problem as multistage optimization involving continuous and discrete variables. While one can use sampling approaches for sub-optimal solution, we provide an exact finite representation through an equivalent finite partition of the one-time reachable set. 
Second, we provide an algorithmic procedure to construct these partitions by making connections to the problem of arrangement of hyperplanes.  
This well-studied problem in computational geometry is finding increasing application in engineering domains such as robotics~\cite{halperin1995arrangements}, fiber-optic networks~\cite{agarwal2013resilience} and even power networks~\cite{bernstein2014power}. Constructing partitions in our case requires a \emph{sweep} operation on polytopes in arbitrary dimensions. A formal approach for this operation, as we provide, is not present in the literature to the best of our knowledge. 
Third, we provide a decomposition approach to compute optimal control for tree reducible networks in two iterations. The analytical solution when the control actions are restricted to shedding load only at $t=0$ relies on establishing invariance of piece-wise linear property of local optimization problems (when viewed as an operator), which could be of independent interest. In our simulation studies, we also consider optimal solution in certain subspaces, a special case of which is the the scaling-based, or proportional, control in \cite[Section 6.1.1]{Bienstock:16}.

% The input-output invariance property could be of independent interest. Third, the sweep operation could be of independent interest. 
% \qbcomment{
%% Computational geometry tools related to arrangement of hyperplanes have traditionally been used in path planning and related topics in robotics, e.g., see \cite{halperin1995arrangements}, and recently come into play in geography related network problems, see \cite{agarwal2013resilience} for fiber-optic networks and for \cite{ bernstein2014power} power networks. 
% The novelty in this paper is not to directly apply the techniques, but to extend them to be incorporated in the optimal search process. The latter is achieved via a novel operation, introduced in this paper, on polytopes of arbitrary dimensions.} We also develop a projection-based approximation algorithm. An extreme case is projection on to a one dimensional space, to which belongs the scaling-based, or proportional, control policies in \cite[Section 6.1.1]{Bienstock:16}. Note that these are the only control policies reported in the literature, to the best of our knowledge, for the problem considered in this paper. 

We conclude this section by defining a few notations. $\real$, $\preal$ and $\spreal$ respectively denote the set of real, non-negative real, and positive real numbers. $\zerobf$ and $\onebf$ denote vectors of all zeros and ones  of proper sizes, respectively. For an integer \( n \), \( [n]:= \{1, 2, \ldots, n\} \). \( |  S | \) denotes the cardinality of \( S \). For a vector \( x \in \real^{d}\), \( \diag{(x)} \in \real^{d\times d} \) denotes the diagonal matrix whose (diagonal) entries are those of $x$. For two vectors $ x $ and $ y $ with the same size, $ x \le y $ means $ x_{i} \le y_{i} $ for all \( i \). The same convention is adopted for $ \ge $, $ < $ and $ > $. Given sets $S_1 \subset \real^n$ and $S_2 \subset \real^n$, $S_1+S_2$ denotes the Minkowski sum of $S_1$ and $S_2$. Several technical proofs are postponed until the Appendix. 

\section{Problem Setup}
\label{sec:setup}
%\qbcomment{
%The problem is formulated within the same simplifications as in \cite{Bienstock:16} on modeling of power system failures, that is, link failures and slow-moving cascade processes are considered; and fast scale dynamics is not explicitly modeled and the standard linearized (DC) approximation to power flow is used. 
%}
We start by recalling the DC power flow approximation.
\subsection{DC Power Flow Approximation}
%DC power flow models are increasingly used in congestion-constrained problems in power system analysis. 
%\ksmargin{should we say 'directed simple graph here'?}
In this model, it is assumed that the transmission lines are lossless and the voltage magnitudes are constant at 1.0 unit. The graph topology of the power network is described by an \emph{undirected multigraph} $ \mathcal{G} = (\mathcal{V}, \mathcal{E})$, that is, multiple parallel links can connect the same two nodes in \( \mathcal{G} \). For convenience, every link in $\mathcal{E}$ is arbitrarily assigned a direction -- the results in the paper do not depend on the direction convention. Let $ \mathcal{V}_{+} \subset \mathcal{V}$ and $ \mathcal{V}_{-} \subset \mathcal{V}$ be the set of supply and demand nodes respectively\footnote{In all the figures, except Figure~\ref{fig:def-tree-reducible-net}, we color supply nodes with blue, and demand nodes with green. For example, see Figure~\ref{fig:triangular-net}.}.  A node is called a \emph{transmission node} if it is neither a supply nor a demand node. Let \( \mathcal{V}_{l} := \mathcal{V}_{+} \cup \mathcal{V}_{-} \) denote the set of non-transmission nodes. Since the network can loose connectivity under cascading dynamics, we let \( (\mathcal{V}, \mathcal{E}) = (\mathcal{V}^{(1)}, \mathcal{E}^{(1)} ) \cup \ldots \cup  (\mathcal{V}^{(r)}, \mathcal{E}^{(r)} ) \) denote the partition of the original graph (i.e., the graph at $t=0$) into its \( r \) connected components. The partition will evolve with the dynamics. 
%Accordingly, let \( \mathcal{V}_{ k }^{(i)} := \mathcal{V}_{k}\cap \mathcal{V}^{(i)} \), \( k\in \{l, +, -\} \), be the quantities associated with the weakly connected components. \ksmargin{do not understand what is this saying} \kscomment{These sets \( \mathcal{V}_{l}, \mathcal{V}_{+}, \mathcal{V}_{-} \) are implicitly associated with a network $ \mathcal{G} = (\mathcal{V}, \mathcal{E})$ and we use them without explicitly mentioning in the paper.} 

The graph \( \mathcal{G} \) is associated with a node-link incidence matrix \( A \in \real^{ \mathcal{V}\times \mathcal{E}} \), where the \( i\)th column \( A_{i} \in \real^{ \mathcal{V}} \) corresponds to link $ i\in \mathcal{E} $ and has $ +1 $ and $ -1 $ respectively on the tail and head node of link $ i $, and $0$ on other nodes. 
%$ A = [a_{1}, a_{2}, \ldots, a_{| \mathcal{E}|}] $ where $ a_{i} \in \real^{ \mathcal{V}} $, being the $ i $th column of $ A $, corresponds to link $ i $. $ a_{i} $ has $ +1 $ and $ -1 $ respectively on the tail and head node of link $ i $, and $0$ on other nodes. 
The links are associated with a flow vector $f \in \real^{\mathcal{E}}$; The signs of elements of $f$ are to be interpreted as being consistent with the directional convention chosen for links in $\mathcal{E}$. We also associate $\mc G$ with a diagonal matrix $W \in \real^{\mathcal{E} \times \mathcal{E}}$ whose diagonal elements give the negative of susceptances, or weights, of the corresponding links. For brevity, $w_i$ shall denote the $i$-th diagonal element of $W$. The nodes are associated with phase angles $\phi \in \real^{\mc V}$, and the supply and demand nodes are associated with a supply-demand vector $p \in \real^{ \mathcal{V}}$; \( p_{i}>0 \) for \( i\in \mathcal{V}_{+} \) and \( p_{i} < 0 \) for \( i\in \mathcal{V}_{-} \). 
% While it is convention to set the entries of \( p \) corresponding to supply (respectively, demand) nodes to be positive (respectively, negative), we choose all the entries of \( p \) to be nonnegative. The sign of entries of \( p \) is captured by a diagonal matrix \( D = \{1, -1, 0\}^{ \mathcal{V}\times \mathcal{V}_{l}} \), where \( D_{ij} \), \ie the entry of \( D \) at \( i \)th row and \( j \)th column, is \( 1 \) (respectively, \( -1 \)) for \( i \) and \( j \) corresponding to the same supply (respectively, demand) node, and \( 0 \) otherwise. 

The quantities defined above are related by Kirchhoff's law and Ohm's law in DC approximation as follows:
\begin{equation}
 \label{eq:power-network-laws}
  A f =  p \quad f = W \trans{A} \phi
\end{equation}

In order for \eqref{eq:power-network-laws} to be feasible, the supply-demand vector $ p $ needs to be balanced over \( \mathcal{G}^{(i)} \) for all \( i\in [r] \), that is, 
\begin{equation}
\label{eq:load-balance}
p\in \mathcal{B}_{ \mathcal{E}} : = \setdef{u\in \real^{ \mathcal{V}}}{ \sum_{v\in \mathcal{V}^{(i)}} u_{v} = 0, \, i\in [r] }
\end{equation}
%\kscomment{\sout{\eqref{eq:load-balance} can be obtained by summing up all the rows of the first equation in (\ref{eq:power-network-laws}) associated with nodes in \( \mathcal{G}^{(i)} \).}} 

For a given network \( \mathcal{G} = ( \mathcal{V}, \mathcal{E}) \) with balanced supply and demand \( p \), there exists a unique flow \( f \) satisfying \eqref{eq:power-network-laws}, and it is given by \cite{Ba.Savla:TCNS16}: 
\begin{equation}
\label{eq:flow-susceptance-relationship}
f = W \trans{A} L^\dagger ( \mathcal{E})  p =:  f( \mathcal{E}, p)
\end{equation} 
where $ L( \mathcal{E}) :=A W \trans{A} \in \real^{ \mathcal{V} \times \mathcal{V}}$ is the weighted Laplacian matrix of \(  \mathcal{G} \) and $ L^\dagger( \mathcal{E}) $ is its pseudo-inverse. 
%As indicated in the function \( f( \mathcal{E}, p) \), flow \( f \) depends on the link set \( \mathcal{E} \). This is because \( A \), \( W \) and \( L \) depends on \( \mathcal{E} \). Nevertheless, 
\eqref{eq:flow-susceptance-relationship} implies that, for a given \( \mathcal{E} \), \( f( \mathcal{E}, p) \) is linear in \( p \). 
%In addition, we have the following remarks about \eqref{eq:flow-susceptance-relationship}. 
%\begin{remark}
%\label{rem:laplacian}
%\kscomment{\eqref{eq:flow-susceptance-relationship} holds true even when \( \mathcal{G} \) is a directed multigraph, that is, multiple parallel links can connect the same two nodes in \( \mathcal{G} \) \cite{Ba.Savla:TCNS16}.} 
%%\leavevmode
%%\begin{enumerate}[(1)]
%%\item \kscomment{\eqref{eq:flow-susceptance-relationship} holds true even when \( \mathcal{G} \) is a directed multigraph, that is, multiple parallel links can connect the same two nodes in \( \mathcal{G} \) \cite{Ba.Savla:TCNS16}.} 
%%\item \( L = \diag\{ L_{1}, L_{2}, \ldots, L_{r} \} \) is a block diagonal matrix with \( L_{i} \) being the laplacian matrix of \( \mathcal{G}^{(i)} \) for all \( i\in \until{r} \). It then follows that $ L^{\dagger} = \diag\{ L_{1}^{\dagger}, L_{2}^{\dagger}, \ldots, L_{r}^{\dagger} \} $. 
%%\end{enumerate}
%\end{remark}

\subsection{Cascading Failure Dynamics}
\label{sec:cascade-dym}
%\ksmargin{move the comment on extensions of failure rule to a later part?}
Let \( \mathcal{E}^{0} \) be the initial link set and let \( p^{0} \) be the initial supply-demand vector satisfying the balance condition in \eqref{eq:load-balance}. The corresponding link flow \( f \) is uniquely determined by \eqref{eq:power-network-laws} or \eqref{eq:flow-susceptance-relationship}. We associate with each link \( i\in \mathcal{E}^{0} \) a thermal capacity \( c_{i} > 0 \). If the magnitude of flow on a link \( i \in \mathcal{E}^{0}\) exceeds its thermal capacity, i.e., \( |f_{i}| >c_{i} \), then link \( i \) fails and is removed from the network irreversibly. This changes the topology of the network, causing flow redistribution, which might lead to more link failures, and so on. Such continuing link failures constitute the \emph{uncontrolled} cascading failure dynamics. Note that we consider a link failure rule which is deterministic and which depends solely on the instantaneous flow. This is to be contrasted with other deterministic outage rules based on moving average of successive flows, or stochastic line outage rules, e.g., see \cite{Bienstock:11, bernstein2014power}. 
%\sout{\kscomment{While these other variations can be easily incorporated in our formulation, the technical results in this paper are specific to the deterministic version.}}

Our objective in this paper is to stop cascading failure through appropriate control actions. While shedding all load at $t=0$ achieves this objective trivially, we desire to take control actions that are optimal in a certain sense. Consider the following description of controlled cascading failure dynamics in discrete-time. Each time epoch corresponds to failure of some links (see Remark~\ref{rem:time-epochs}). The node set remains the same. Let $ ( \mathcal{E}^{t}, p^{t}) $ be the state of the network at time $ t $, with $ \mathcal{E}^{t} \subset 2^{ \mathcal{E}^{0}} $ and $ p^{t} \in \real^{ \mathcal{V}} $ denoting the \emph{active link set} and supply-demand vector at time $t$, respectively. We consider load shedding as the control and, for convenience, employ control variable \( u\in \real^{ \mathcal{V}} \) to be supply-demand vector after load shedding. The controlled cascading failure dynamics, for $t=0, 1, \ldots$, and starting from the initial state \( (\mathcal{E}^{0}, p^{0}) \), is given by:
\begin{equation}
\label{eq:cascade-dynamics}
\left(\mathcal{E}^{t+1}, p^{t+1}\right) = \mathcal{F} \left(\mathcal{E}^{t}, p^{t}, u^{t} \right), \qquad u^{t} \in U(\mathcal{E}^{t},p^{t})
\end{equation}
where the component functions are \( \mathcal{F}_{\mathcal{E}}( \mathcal{E}, p, u) \equiv \mathcal{F}_{\mathcal{E}}(\mathcal{E},u) := \setdef{i\in \mathcal{E}}{ |f_{i}( \mathcal{E}, u)|  \le c_{i}}  \) and \( \mathcal{F}_p(\mathcal{E}, p, u) \equiv \mathcal{F}_p(u) := u \).
%where
%\begin{equation}
%\label{def:feasible-link-set}
%\begin{split}
%\mathcal{F}_{\mathcal{E}}( \mathcal{E}, p, u) \equiv \mathcal{F}_{\mathcal{E}}(\mathcal{E},u) & := \setdef{i\in \mathcal{E}}{ -c_{i} \le f_{i}( \mathcal{E}, u)  \le c_{i}} \\
%\mathcal{F}_p(\mathcal{E}, p, u) \equiv \mathcal{F}_p(u) & := u
%\end{split}
%\end{equation} 
%where the functions \( \mathcal{F}_{ \mathcal{E}} \) and \( \mathcal{F}_{p} \) are for the maps for \( \mathcal{E}^{t+1} \) and \( p^{t+1} \), respectively. As defined in \eqref{def:feasible-link-set}, 
\( \mathcal{F}_{ \mathcal{E}} \) is the set of feasible links in \( \mathcal{E} \) under supply-demand vector \( u \); and the control input \( u^{t} \) at time \( t \) becomes the next state supply-demand vector \( p^{t+1} \). In order for \( \mathcal{F}_{ \mathcal{E}}( \mathcal{E}, u) \) to be well-defined, \( u \) must be balanced with respect to the active link set \( \mathcal{E} \). This is ensured by the following definition of state-dependent control space  \( U( \mathcal{E}, p) \):
 \begin{equation}
\label{def:control-space}
\begin{aligned}
U( \mathcal{E}, p) = \cube(p) \cap \mathcal{B}_{ \mathcal{E}}
\end{aligned}
\end{equation}
where \( \cube{p} := \setdef{u \in \real^{ \mathcal{V}}}{ 0\le \sign(p_{v}) u \le |p_{v}|\, \forall v\in \mathcal{V}}\), with \( \sign(x) \) being 1 for \( x\ge 0 \) and -1 for \( x<0 \), characterizes the load shedding requirement. 
%where \( \cube(p) \) characterizes the load shedding property, and is defined as: 
%\begin{equation}
%\label{def:cube}
%\cube(p) := \setdef{u \in \real^{ \mathcal{V}}}{0 \le u_{v} \le p_{v} \text{ for } p_{v}\ge 0; \, p_{v} \le u_{v}\le 0 \text{ for } p_{v} < 0} 
%\end{equation}
\( U( \mathcal{E}, p) \) includes all \emph{admissible} load shedding controls at state \( (\mathcal{E}, p) \). 
In particular, if all the supply and demand nodes are disconnected from each other at state \( ( \mathcal{E}, p) \), then $\mc B_{\mc E}=\{\zerobf\}$, and in this case \( U( \mathcal{E}, p) = \{\zerobf\} \). 

\begin{remark}
\label{rem:time-epochs}
\leavevmode
\begin{enumerate}
\item \eqref{eq:cascade-dynamics} is of interest only until the time epoch when the link failures stop, e.g., when the network state becomes \emph{feasible} (defined in Section~\ref{sec:problem-formulation}). For the sake of completeness, one can define subsequent time epochs arbitrarily, e.g., at fixed intervals.
\item The number of connected components may increase under \eqref{eq:cascade-dynamics}. When this happens at state \( ( \mathcal{E}, p) \), it is possible that \( p\not \in \mathcal{B}_{ \mathcal{E}} \). However, \eqref{def:control-space} ensures that the control action $u \in U(\mc E,P)$, which is the controlled value of $p$, is balanced with respect to $\mc E$.  Therefore, the balance condition in \eqref{eq:load-balance} is operationally satisfied. 
\item By adopting the steady state DC model in \eqref{eq:cascade-dynamics}, we implicitly neglect the transient power flow dynamics. This is justified by the fact that the transient dynamics evolve at a considerably faster time scale in comparison to the initial slow-evolving cascade process observed in practice \cite{Bienstock:16}. Along the same reasoning, the cascading dynamics in \eqref{eq:cascade-dynamics} can be generalized to other systems whose dynamics evolve fast enough so that assuming the system reaching equilibrium between failure epochs is a reasonable approximation. The control action $u$ is to be then interpreted as adjusting system parameters to choose equilibrium, e.g., (DC) power flow in our setting. 
%\item The cascading dynamics in \eqref{eq:cascade-dynamics} can be generalized, beyond the specific DC network setting of this paper, to settings where the physical dynamics evolves fast enough to reach equilibrium between failure epochs. $u$ can then be interpreted as control action to choose equilibrium at the next time epoch under the map \( \mathcal{F}_{p} \).
\end{enumerate}
\end{remark}

\subsection{Problem Formulation}
\label{sec:problem-formulation}
Let 
\begin{equation}
\label{eq:S-def}
 %\mathcal{S} := \setdef{( \mathcal{E}, p)}{ \mathcal{F}_{ \mathcal{E}}( \mathcal{E}, p) = \mathcal{E}, p \in \mathcal{B}_{\mathcal{E}}} 
  \mathcal{S} := \setdef{( \mathcal{E}, p)}{ p \in \mathcal{B}_{\mathcal{E}}, |f_{i}( \mathcal{E}, p)| \le c_{i},\forall\, i\in \mathcal{E}} 
 \end{equation}
  denote the set of \emph{feasible states}. Set \(\mathcal{S}\) is invariant under the uncontrolled cascading dynamics. Note that \( ( \mathcal{E}, \zerobf) \in \mathcal{S} \) for every \( \mathcal{E} \in 2^{ \mathcal{E}^{0}} \). Since \( \mathcal{E}^{t}, t \geq 0 \), is non-increasing sequence, and \( \mathcal{E}^{0} \) is finite, the dynamics converges to a feasible state within $2^{\mc E^0}$ time epochs. 

Our objective is to choose control actions to steer the network from an arbitrary given initial state $(\mc E^0,p^0)$ to a feasible state $(\mc E^N, p^N) \in \mc S$ within a given finite horizon $N$, while optimizing a certain performance criterion. The control horizon $N$ is typically much smaller than $2^{\mc E^0}$. 
%At the same time, one may wish to terminate cascading failure in limited number of epochs. Therefore, without loss of generality, we consider \( N \) (finite) horizon for the load shedding problem. 
Let a generic sequence of control actions over the control horizon be denoted by \( u := (u^{0}, \ldots, u^{N-1}) \). 
%
% which leads the system from the initial state $(\mathcal{E}^{0}, p^{0})$ to a steady-state \( ( \mathcal{E}^{N}, p^{N})  \in \mathcal{S}\) containing maximal amount of remaining load. Formally, the problem is:
In this paper, we wish to solve the following optimal control problem:
\begin{equation}
\label{eq:control-formulation-new}
{\sup_{u \in \mc D(\mc E^0,p^0, N)}}  \trans{\sdsign} p^{N}
%\begin{array}{rl}
%\displaystyle{\sup_{u^{\until{N}} \in \mc D}} & \trans{\sdsign} p^{N} \\
% \textrm{s.t.} &  (\mathcal{E}^{N}, p^{N}) \in \mc S   \\
%  & \eqref{eq:cascade-dynamics}
%\end{array}
\end{equation}
%\begin{equation}
%\label{eq:control-formulation}
%\begin{array}{rl}
%\displaystyle{\sup_{u^{\until{N}}} } & \trans{\sdsign} p^{N} \\
% \textrm{s.t.} &  (\mathcal{E}^{N}, p^{N}) \in \mc S   \\
%  & \eqref{eq:cascade-dynamics}
%\end{array}
%\end{equation} 
where \( \sdsign \in \{1, 0, -1\}^{ \mathcal{V}} \) is a constant defined as: \( \sdsign_{v} := 1 \) for \( v \in \mathcal{V}_{+} \), \( \sdsign_{v} := -1 \) for \( v\in \mathcal{V}_{-} \), and \( \sdsign_{v} := 0 \) otherwise, and the set of feasible control actions is defined as:
\begin{equation}
\label{eq:feasible-set-def}
\begin{aligned}
\mc D(\mc E^0,p^0, N):= \{(u^{0}, \ldots, u^{N-1}) \, & |\, u^t \in U(\mc E^t, p^t) \text{ for } t  = 0, \ldots, N-1; \\ 
&(\mc E^{N-1}, u^{N-1}) \in \mc S; (\mc E^t, p^t)_{t \in [N]} \text{ satisfies } \eqref{eq:cascade-dynamics} \}
\end{aligned}
\end{equation}

For brevity, we shall not show the dependence of $\mc D$ on $\mc E^0$, $p^0$ and $N$ when clear from the context.
\begin{remark}
\label{rem:additional-check}
\begin{enumerate}
%\item \eqref{eq:feasible-set-def} implies that, when checking feasibility of a given $u$, one has to check an additional condition for $u^{N-1}$ in comparison to that for $u^t$ for $t =0, \ldots, N-2$. In addition to checking $u^{N-1} \in U(\mc E^{N-1},p^{N-1})$, one also needs to check if the resulting $(\mc E^{N},p^N)$ belongs to $\mc S$.
\item An arbitrary sequence of \emph{admissible} control actions (cf. \eqref{def:control-space}) is not necessarily \emph{feasible}. \eqref{eq:feasible-set-def} implies that, when checking feasibility of a given $u$, in addition to checking admissibility, i.e., $u^{t} \in U(\mc E^{t},p^{t})$ for $t =0, \ldots, N-1$, one also has to check that the failure stops at \( t= N-1 \). 
\item In \eqref{eq:control-formulation-new}, we use supremum rather than maximum because $\mc D$ is not closed in general, as illustrated in Example \ref{eg:policy-not-closed} below. This matter is addressed in Section~\ref{sec:retrieve-optimal-control} and it could be ignored before that. The computational complexity of characterizing \( \mathcal{D} \), and hence of solving \eqref{eq:control-formulation-new} is attributed to the cascading dynamics in \eqref{eq:cascade-dynamics}.
\end{enumerate}
\end{remark}
%That is to say, at \kscomment{stage} \( N-1 \) and state \( ( \mathcal{E}^{N-1}, p^{N-1}) \), we assume the controller must take an action \( u^{T-1} \) such that \( ( \mathcal{E}^{N-1}, u^{T-1} ) \) is feasible. Correspondingly, the control space at \( t \) is redefined as \( U_{t}( \mathcal{E}, p): = U( \mathcal{E}, p) \) for \( 0\le t \le N-2 \) and \( U_{N-1}( \mathcal{E}, p) := \setdef{u \in U( \mathcal{E}, p)}{ \mathcal{F}_{ \mathcal{E}}(\mathcal{E}, u) = \mathcal{E} } \). 
%To show this, we define the set of admissible control actions sequence as follows. 
%\begin{equation}
%\label{eq:feasible-policy-space}
%\begin{split}
%\mathcal{D} :=  \{(u^{0}, \ldots, u^{t-1}) \, &|\, u^{t} \in \real^{\mathcal{V}_{l} }, u^{t} \in \mathcal{U}(\mathcal{E}(t-1), u^{t-1}), \\ 
%&\mathcal{E}^{t+1} = \mathcal{F}_{ \mathcal{E}}( \mathcal{E}^{t}, u^{t}), \forall\, t = 0, 1, \ldots, T-1; \mathcal{E}^{t} = \mathcal{E}(T-1)\} 
%\end{split}
%\end{equation}
%where  \( u(-1):= p^{0} \) and \( \mathcal{E}^{0}:= \mathcal{E}^{0} \).   

\begin{example}
\label{eg:policy-not-closed}
Consider the optimal control problem in \eqref{eq:control-formulation-new} for the network shown in Fig. \ref{fig:triangular-net-0}, for $N=3$. Node 1 is the supply node and nodes 2 and 3 are the demand nodes. The initial supply-demand vector is \( p^{0} =\trans{[30, -10, -20]} \). The link weights are \( w = \trans{[2, 1, 1, 1]} \) and the link capacities are \( c = \trans{[6, 7, 14, 5]} \). Consider \( u^{0} = \alpha_{k} = \trans{[21+2/k, -7-1/k, -14-1/k]} \) for some \( k\ge 1 \). The resulting flow is \(f(\mc E^0, u^{0}) = \trans{[8+6/(7k),  4+3/(7k), 9+5/(7k), 5+ 2/(7k)]} \). Consequently, links \( e_{1} \) and \( e_{4} \) fail due to flow exceeding capacity, and the resulting $\mc E^1$ is shown in Fig. \ref{fig:triangular-net-a}. For \( u^{1} = \alpha_{\infty} = \trans{[21, -7, -14]} \),  \( f_{2}( \mathcal{E}^{1}, u^{1}) = 7 \le c_{2} \), \( f_{3}( \mathcal{E}^{1}, u') = 14 \le c_{3} \), and therefore there are no more link failures. This implies that \( u(k) := (\alpha_{k}, \alpha_{\infty}, \alpha_{\infty}) \in \mathcal{D} \) for every \( k \ge 1 \). However, \( \hat{u}= \lim_{k\to \infty} u(k) =  (\alpha_{\infty}, \alpha_{\infty}, \alpha_{\infty}) \not \in \mathcal{D} \). This is because $f(\mc E^0, \alpha_{\infty})$ is such that only link $e_1$ fails. The resulting $\hat{\mc E}^1$ is shown in Fig. \ref{fig:triangular-net-b} and $f(\hat{\mc E}^1, \alpha_{\infty})= \trans{[\text{null}, 28/3, 35/3, 7/3]}$, which implies that \( e_{2} \) fails. Thereafter, $\hat{\mc E}^2=\{e_3,e_4\}$, \( f_{3}(\hat{\mathcal{E}}^{2}, \alpha_{\infty}) = 21 > c_{3} \) and  \( f_{4}(\hat{\mathcal{E}}^{2}, \alpha_{\infty}) = 7 > c_{4} \). All links fail under \( \hat{u} \) and thus \( \hat{u} \not\in \mathcal{D} \).   
% (\trans{[21+2/k, -7-1/k, -14-1/k]},  \trans{[21, -7, -14]}, \trans{[21, -7, -14]})
%and is following by the failure of \( e_{3} \) and \( e_{4} \) if no more load shedding is taken. 
This demonstrates that \( \mathcal{D} \) is not closed for the given choice of network parameters.
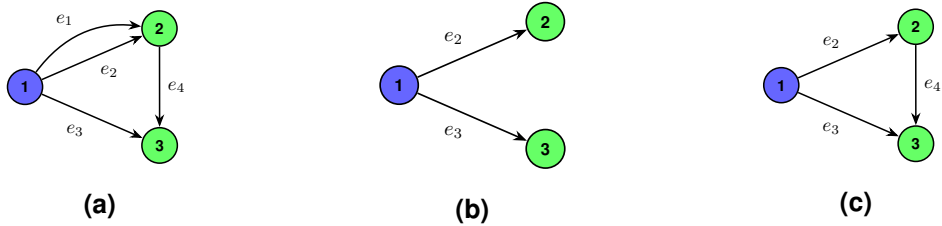
\begin{figure}[htbp]
   \centering
   \begin{subfigure}{.3\linewidth}
  \centering
  \begin{tikzpicture}[node distance=0.6cm and 2cm]
 \node[supply node] (1) {1};
 \node[demand node] (2) [above right =of 1] {2};
 \node[demand node] (3) [below right =of 1] {3};

 \path[main edge, edge label]
 (1) edge[bend left] node{\( e_{1} \)} (2)
   (1) edge node[below right]{\( e_{2} \)} (2)
   (1) edge node[below left] {\( e_{3} \)} (3)
   (2) edge node{\( e_{4} \)} (3);
   
\end{tikzpicture}
  \caption{}
  \label{fig:triangular-net-0}
\end{subfigure}%
\begin{subfigure}{.3\linewidth}
  \centering
  \begin{tikzpicture}[node distance=0.6cm and 2cm]
 \node[supply node] (1) {1};
 \node[demand node] (2) [above right =of 1] {2};
 \node[demand node] (3) [below right =of 1] {3};

 \path[main edge, edge label]
   (1) edge node{\( e_{2} \)} (2)
   (1) edge node[below left] {\( e_{3} \)} (3);
   
\end{tikzpicture}
  \caption{}
  \label{fig:triangular-net-a}
\end{subfigure}
\begin{subfigure}{.3\linewidth}
  \centering
  \begin{tikzpicture}[node distance=0.6cm and 2cm]
 \node[supply node] (1) {1};
 \node[demand node] (2) [above right =of 1] {2};
 \node[demand node] (3) [below right =of 1] {3};

 \path[main edge, edge label]
   (1) edge node{\( e_{2} \)} (2)
   (1) edge node[below left] {\( e_{3} \)} (3)
   (2) edge node{\( e_{4} \)} (3);
   
\end{tikzpicture}
  \caption{}
  \label{fig:triangular-net-b}
\end{subfigure}
\caption{\sf The graph topology for the network used in Example~\ref{eg:policy-not-closed} to illustrate that the feasible control action set $\mc D$ is not necessarily closed.}
\label{fig:triangular-net}
\end{figure}
\end{example}

\begin{remark}
In writing the optimal control problem in \eqref{eq:control-formulation-new}, we only consider the terminal cost \( \trans{\sdsign} p^{N} \), in addition to imposing feasibility condition on the terminal state. \( \trans{\sdsign} p^{N} \) is the remaining cumulative supply and demand once the cascading failure stops, and hence is a natural choice for the objective function in \eqref{eq:control-formulation-new}. Extension to including running cost is discussed in Remark~\ref{rem:in-orthant}.
%Remark~\ref{rem:aggregation-is-decomposition}.
\end{remark}

\subsection{Solution via Search}
\label{sec:problem-statement} 
A generic approach to solving \eqref{eq:control-formulation-new} is by performing a \emph{search}, e.g., see \cite[Chap 3]{russell2009artificial}, on a directed tree composed of the states reachable from the initial state $(\mc E^0,p^0)$ in most $N$ time steps. In other words, the tree is rooted at $(\mc E^0,p^0)$, and has depth $N$. Each node of the tree corresponds to a state $(\mc E,p)$ which is reachable in one time step from its parent node under a control action which is associated with the incoming arc to that node. 
%; see Fig. \ref{fig:seach-tree} for an illustration. 
When considering one time step reachable set from a given node $(\mc E,p)$, one only considers control actions belonging to $U(\mc E,p)$. The set of goal states for the search is \( \mathcal{S} \) and we associate every feasible state \( ( \mathcal{E}, p)\in \mathcal{S} \) with a reward \( \trans{\sdsign}p \) and every infeasible state with a reward \( -\infty \). The objective is to search for a state in \( \mathcal{S} \) with maximal reward.

Let \( J_{t}( \mathcal{E}, p) \) be the maximum among utilities of all the states that can be reached in at most $t$ time steps starting from $(\mc E, p)$. Solving \eqref{eq:control-formulation-new} is equivalent to computing \( J_{N}( \mathcal{E}^{0}, p^{0}) \). This computation can be done as follows:
\begin{subequations}
\label{eq:value-iteration}
\begin{align}
%\label{eq:value-iteration-a}
%J_{0}( \mathcal{E}, p) & = r( \mathcal{E}, p) \\
\label{opt:lp-redispatch} 
J_{1}( \mathcal{E}, p ) & =  \max_{u \in U(\mc E,p)} \trans{\sdsign} u \quad \text{ s.t.} \quad |f( \mathcal{E}, u)|  \leq c_{ \mathcal{E}} \\
 \label{eq:value-iteration-b} 
J_{t}(\mathcal{E}, p) & = \sup_{u \in U(\mathcal{E}, p)} J_{t-1} \left(\mathcal{F}_{ \mathcal{E}}(\mathcal{E}, u), u \right), \,\,
 t = 2, \ldots, N
 \end{align}
\end{subequations} 
where \eqref{opt:lp-redispatch} uses the flow capacity constraint to account for the additional constraint to be satisfied by $u^{N-1}$, as commented on in Remark~\ref{rem:additional-check}. \eqref{opt:lp-redispatch} is a linear program with nonempty feasible set (recall \( ( \mathcal{E}, \zerobf) \in \mathcal{S} \) for all \( \mathcal{E} \)) and commonly referred to as \emph{LP power redispatch}, e.g., see \cite{chen2005cascading}.  \eqref{eq:value-iteration-b} uses supremum because \( \mathcal{F}_{ \mathcal{E}}(\mathcal{E}, u) \), and hence \( J_{t-1} \left(\mathcal{F}_{ \mathcal{E}}(\mathcal{E}, u), u \right) \), is not continuous w.r.t. \( u \). 
%\qbcomment
%\begin{equation}
%\label{opt:lp-redispatch} 
%\begin{array}{rcl}
%J_{1}( \mathcal{E}, p ) =& {\displaystyle \max_{u \in U(\mathcal{E}, p)} } &  \trans{\sdsign} u \\
%& \text{s.t.}  & -c_{ \mathcal{E}} \leq f( \mathcal{E}, u) \leq c_{ \mathcal{E}}  \\
%\end{array}
%\end{equation}}
%whereas maximization is used in \eqref{opt:lp-redispatch} \kscomment{because \( \mathcal{F}_{ \mathcal{E}}( \mathcal{E}, u) = \mathcal{E} \) for \( u\in U_{N-1}( \mathcal{E}, p) \)}.  
It is straightforward to see that \( J_{1}( \mathcal{E}, p) \le J_{t}(\mathcal{E}, p) \le \trans{\sdsign} p \) for all \( (\mathcal{E},p) \) and \( t\in [N] \). 
%\kscomment{These bounds can be used to \emph{trim} the tree in the search process; details are provided in Section~\ref{sec:tree-search}.} 
%\qbcomment{
%\begin{remark}
%Though \eqref{eq:value-iteration} is similar to the value iteration in dynamic programming, we note that in this paper, a search algorithm in forward direction over the state tree is preferred over the value iteration in backward direction, because the reachable set can be considerably smaller than the state space. 
%%Comparison of time complexity of these two algorithms can be found in Section~\ref{sec:running-time}.  
%\end{remark}
%}

Executing tree search, or equivalently implementing \eqref{eq:value-iteration}, is not directly amenable to a computational procedure, since the number of one-step reachable states from $(\mc E,p)$, or equivalently the set of admissible control actions $U(\mc E,p)$, is a continuum in general. A natural strategy is to discretize $U(\mc E,p)$, at the expense of getting less scalable algorithms and approximate solutions.
%Computing a solution to \eqref{eq:control-formulation} by tree searching is prohibitive. One has to nearly search through the entire (reachable) state space\footnote{Since reachable set is the relevant state space for searching, by \emph{state space}, we shall be implicitly referring to reachable set.} in order to obtain an optimal solution. However, the action set \( U( \mathcal{E}, p) \) is a continuum for most of the interesting state \( ( \mathcal{E}, p) \) (e.g., \( (\mathcal{E}, 0) \) excluded ) and hence the search tree has infinite branching factor. One has to resort discretization of action \( U( \mathcal{E}, p) \) in order to proceed with numerical implementation of tree search.  The dimensions of affine hull of \( U( \mathcal{E}, p) \)  is \( |\mathcal{V}_{l}| - r \), where \( \mathcal{V}_{l} \) and \( r \) are the set of nontransmission nodes and the number connected network components, respectively, for network being at the state \( ( \mathcal{E}, p) \). Plus the complexity induced by the depth of the search tree, it amounts to search within a \( N(|\mathcal{V}_{l}| - r) \) dimensional space, which is informally referred as the \emph{dimension of problem} \eqref{eq:control-formulation}.
In this paper, we propose the following two approaches for better computational efficiency: 
%
%To address the challenge above, we proposed the following two methods. 
\begin{enumerate}[(I)]
\item (semi-)analytic solution for a certain class of networks, or for optimal solution within a certain class of control policies (Section~\ref{sec:tree-dp}); and
%
%By exploiting sparsity properties in network structures, we obtain explicit solution for parallel networks and algorithmic solution for tree reducible networks (see Section~\ref{sec:tree-dp}).
% is decomposed into problems depending solely on local network components and consequently the problem dimension is greatly decreased 
\item an algorithmic procedure to construct an \emph{equivalent} finite abstraction of the set of admissible control actions, such that computing optimal solution over this finite abstraction gives solution to \eqref{eq:control-formulation-new} (Section~\ref{sec:state-aggregation-main}).
%
%By aggregating the state in a special way, we obtain finite representation of the infinite search tree (see Section~\ref{sec:state-aggregation}). Classical tree search algorithms are hence applied. 
\end{enumerate}

\section{Analytical Solution}
\label{sec:tree-dp}

%Network topology has considerable effect on the complexity of optimal control synthesis. 
In this section, we present (semi-)analytical solution to \eqref{eq:control-formulation-new} in some special cases.
 
% the class of \emph{tree reducible} networks and manage to decompose \eqref{eq:control-formulation} recursively to a class of sub-problems that depends only on local star subnetworks. The dimension of the searching space is hence greatly decreased. Furthermore, the decomposition scheme provides an efficient algorithm to find the optimal action sequence within the class of constant control. 

\subsection{Parallel Networks}
\label{sec:para-net}

\begin{figure}[htbp]
\centering
 \begin{subfigure}[t]{.44\linewidth}
  \centering
   \begin{tikzpicture}[node distance=3cm]
\node[supply node] (1) {1}; 
\node[demand node] (2) [right of=1] {2}; 
%\node[coordinate] (in) [left = 0.5cm of 1] {}; 
% \node[coordinate] (out) [right = 0.5cm of 2] {}; 
  
  \path[main edge, edge label] 
(1) edge[bend left] node{$ e_{1} $} (2)
(1) edge[bend right] node{$e_{2}$} (2);
% (in) edge node[above left]{\(\mu \)} (1)
% (2) edge  node[above right]{\(\mu \)} (out);
\end{tikzpicture}
  \caption{}
  \label{fig:2node-example}
\end{subfigure}
 \begin{subfigure}[t]{.54\linewidth}
  \centering
  \begin{tikzpicture}[node distance=0.4cm and 1.5cm]
 \node[supply node] (1) {1};
 \node[main node] (2) [right =of 1] {2};
 \node[demand node] (3) [right =of 2] {3};

 \path[main edge, edge label]
   (1) edge node[below]{$e_{1}$} (2)
   (1) edge[bend left] node {$e_{2}$} (2)
   (2) edge node[below]{\( e_{3} \)} (3)
   (2) edge[bend left] node {\( e_{4} \)} (3)
   (1) edge[bend right] node[below] {\( e_{5} \)} (3);

\end{tikzpicture}
  \caption{}
\label{fig:diamond-net}
\end{subfigure}
\caption{(a) a parallel network with two links. (b) the network used in Example~\ref{eg:one-shot-control} to illustrate that the set of feasible one-shot control actions can be neither connected nor closed.}
\end{figure}
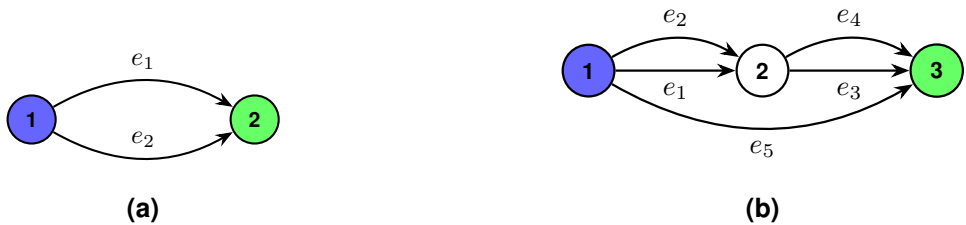

A parallel network \( ( \{1, 2\},  \mathcal{E}) \) consists of two nodes that are connected by multiple parallel links, e.g., see Figure~\ref{fig:2node-example}. We set the convention that the links are directed from node 1 (supply) to node 2 (demand). Since $p_2=-p_1$, following (\ref{eq:power-network-laws}), the link flows are given by \( f_{i} = p_{1} w_{i}/\sum_{j \in \mathcal{E}} w_{j} \) for all \( i\in \mathcal{E} \). 
%\begin{equation}
%\label{eq:flow-weight-relationship}
%f_i = \frac{w_{i}}{ \sum_{j \in \mathcal{E}} w_{j} } p_{1} \, , \qquad  i \in \mathcal{E}
%\end{equation} 
The following monotonicity result is straightforward. 
\begin{lemma}
\label{lem:monotonicity}
Consider two arbitrary parallel networks \( ( \{1, 2\},  \mathcal{E} ) \) and \( ( \{1, 2\},  \mathcal{E}') \) such that \( \mathcal{E} \subseteq \mathcal{E}' \) and two arbitrary supply-demand vectors \( p = p_{1}\trans{[1\; -1]} \)  and \( p' = p'_{1}\trans{[1\; -1]} \) such that \( 0 < p_{1} \le p'_{1}  \). Then the following are true: (i) \(  \mathcal{F}_{ \mathcal{E}}( \mathcal{E}, p') \subseteq \mathcal{F}_{ \mathcal{E}}( \mathcal{E}, p) \); and (ii) \(   \mathcal{F}_{ \mathcal{E}}( \mathcal{E}, p) \subseteq \mathcal{F}_{ \mathcal{E}}( \mathcal{E}', p) \). 
\end{lemma}
\begin{remark}
\label{rem:monotonicity}
For a parallel network \( ( \{1, 2\},  \mathcal{E}) \) and a natural number \( N \ge 1 \), let \( (u^{0}, \ldots, u^{N-1}) \) and \( (\tilde{u}^{0}, \ldots, \tilde{u}^{N-1}) \) be two sequences of control actions and \( (\mathcal{E}^{1}, \ldots, \mathcal{E}^{N} ) \) and \( (\tilde{ \mathcal{E}}^{1}, \ldots, \tilde{ \mathcal{E}}^{N}) \) be the topology sequences, respectively, under two controls. Lemma~\ref{lem:monotonicity} implies that if \( u_{1}^{t} \ge \tilde{u}_{1}^{t} \) for all \(0\le  t \le N-1 \), then \( \mathcal{E}^{t} \subseteq \tilde{\mathcal{E}}^{t} \) for all \( 0\le t \le N-1 \). 
\end{remark}

For all $i \in \mc E$, ${f_i}/{c_i} = \left({p_1}/{\sum_{k\in \mathcal{E}} w_k} \right) \cdot \left({w_i}/{c_i}\right)$. Noting the common factor ${p_1}/{\sum_{k\in \mathcal{E}} w_k}$ among all links, we label links in the increasing order of $w_{i}/c_{i} $, i.e., ${w_{i}}/{c_{i}} \le {w_{j}}/{c_{j}}$ for \( i\le j \), \( i, j \in \mathcal{E} \). The chronological order of link failures according to \eqref{eq:cascade-dynamics} is expected to be aligned with the reverse labeling of the links, and is not affected by different load shedding actions, as implied by the following result. %, whose proof is omitted. 
%
% this order remains fixed throughout the cascading dynamics. This is shown in Lemma \ref{lemma:failure-order}. 

\begin{lemma}
\label{lemma:failure-order}
Consider the cascading dynamics \eqref{eq:cascade-dynamics} for a parallel network \( ( \{1, 2\},  \mathcal{E}) \). If $f_j(t) \leq c_j$ for some $j \in \mc E$ and $t \in [N]$, then $f_i(t) \leq c_i$ for all $ i \leq j $. 
\end{lemma}
\begin{proof}
Since $f_{j}(t) =  p_{1}(t) \frac{w_{j}}{ \sum_{k\in \mathcal{E}} w_{k} } \le c_{j}$, then, for 
%\begin{equation*}
%f_{j} = \frac{w_{j}}{ \sum_{k =1}^m w_{k} } p_{1} \le c_{j}
%\end{equation*}
%Then for link 
$ i<j $, we have \(f_i(t)= p_{1}(t) \frac{w_{i}}{ \sum_{k\in \mathcal{E}} w_{k} }  = w_{i} {f_{j}(t)}/{w_{j}}  \le w_{i} {c_{j}}/{w_{j}} \le c_{i}\). 
%\begin{equation*}
%f_{i} =  \frac{w_{i}}{ \sum_{k =1}^m w_{k} } p_{1} = \frac{f_{j}}{w_{j}} w_{i} \le  \frac{c_{j}}{w_{j}} w_{i} \le c_{i}
%\end{equation*}
The last inequality is due to ${w_{i}}/{c_{i}} \le {w_{j}}/{c_{j}} $ for all $ i \le j $. 
\end{proof}
\begin{remark}
\label{rem:parallel-network}
Lemma~\ref{lemma:failure-order} implies that for all \( t\in [N] \), there exists \( j\in \mathcal{E} \) such that \( \mathcal{E}^{t} = [j] \). Because \( \mathcal{E}^{t} \) is non-increasing, at most \( | \mathcal{E}|+1 \) number of distinct network topologies can occur in the cascading dynamics. 
\end{remark}

The monotonicity properties shown in Lemma~\ref{lem:monotonicity} and the tight characterization of the reachable set of graph topologies, as implied by Remark~\ref{rem:parallel-network}, allows optimal control synthesis relatively easily. Specifically, we show that a \emph{one-shot control} defined next is optimal within all control policies for parallel networks. 
\begin{definition}
\label{def:one-shot-and-constant-control}
For an initial supply-demand vector \( p^{0} \in \real^{ \mathcal{V}} \), a \( N \) stage control sequence \( (u^{0}, \ldots, u^{N-1}) \) is called \emph{one-shot control} if, for some \( 0\le t_{1} \le N-1 \), \( u^{t} = p^{0} \) for all \( t < t_{1} \),  \( u^{t_{1}} \in \cube{p^{0}} \), and \( u^{t} = u^{t_{1}} \) for all \( t \ge t_{1} \). Moreover, if \( t_{1} = 0 \), then it is also called \emph{constant control}. 
\end{definition}

In order to describe the analytical expression of an optimal one-shot control for a parallel network \( ( \{1, 2\},  \mathcal{E}) \), we first introduce a few notations. Let \( \pnetval_{i} :=  ({c_{i}}/{w_{i}}) \sum_{j=1}^{i} w_{j}\) for all \( i\in \mathcal{E} \). In general, \( \pnetval_{i} \) is neither decreasing nor increasing with respect to \( i \). The following remark is straightforward. 
%\begin{equation}
%\label{eq:2bus-net-value-fun}
%\pnetval_{i} := \frac{c_{i}}{w_{i}} \sum_{j=1}^{i} w_{j}, \quad i\in \mathcal{E} 
%\end{equation}
\begin{remark}
\label{rem:value-V}
$ \pnetval_{i} $ is the maximum supply or demand the network can support when only the first $ i $ links are active:  for all \( i\in \mathcal{E} \), \( ( [i], p^{0}) \in \mathcal{S} \) if and only if \( p^{0}_{1} \le \pnetval_{i} \). 
\end{remark}

Let \( \clink_{1} := \max \argmax_{i\in \mathcal{E}} \pnetval_{i} \), and let \( \clink_{j+1} := \max \argmax_{i> \clink_{j} } \pnetval_{i} \) if \( \clink_{j} < |\mathcal{E}| \). Let \( \cmax \) be the maximum number such that \( \clink_{\cmax} \) is defined. It is straightforward to see that \( \clink_{1} < \clink_{2} < \ldots < \clink_{\cmax} = | \mathcal{E}| \) and \( \pnetval_{\clink_{1}} > \pnetval_{\clink_{2}} > \ldots > \pnetval_{\clink_{\cmax}} = \pnetval_{| \mathcal{E} |} \). For a given initial balanced supply-demand vector \( p^{0} \in \real^{2} \), an optimal control depends on the value of \( N \). A big \( N \) provides more flexibility for control design. A small \( N \) forces to shed big portion of loads at small time instants to ensure network feasibility. For example, for \( N=1 \), sufficiently large amount of load needs to be shed at \( t=0 \) to ensure that all links become feasible. Next we define a quantity \( N_{j}(p^{0}) \) for every balanced \( p^{0} \in \real^{2} \) and \( j\in [\cmax] \), which will be used in the specification of optimal control.
%such that for \( N= N_{j}(p^{0}) \), \( \mathcal{E}^{N} = [\clink_{j}] \) under an optimal control. 
Let \( ( \mathcal{E}_{\mathrm{un}}^{0}, \ldots, \mathcal{E}_{\mathrm{un}}^{N}) \) be the non-increasing topology sequence of the uncontrolled cascading dynamics \eqref{eq:cascade-dynamics} (that is, \( u^{t} = p^{0} \) for all \( t \)). Let \( \pnetval_{\clink_{0}} := \infty \), \( \pnetval_{\clink_{\cmax+1}} = 0 \) and \( \mathcal{E}_{\mathrm{un}}^{-1} \supset \mathcal{E}_{\mathrm{un}}^{0} \) for convenience. Let \( j \in [\cmax]\cup \{0\}\) be such that \( \pnetval_{\clink_{j+1}} < p^{0}_{1} \le \pnetval_{\clink_{j}} \) and let \( N_{k}(p^{0}) := 1+ \min\setdef{t\in \{0, \ldots, N \}}{ ( \mathcal{E}_{\mathrm{un}}^{t}, p^{0}) \in \mathcal{S}} \) for \( 1\le k \le j \) and \( N_{k} \)  be such that \( \mathcal{E}_{\mathrm{un}}^{N_{k}(p^{0}) -1 } \subseteq [\clink_{k}] \subset \mathcal{E}_{\mathrm{un}}^{N_{k}(p^{0})-2} \) for all \( j+1 \le k \le \cmax \). The above definition implies that \( | \mathcal{E}| \ge N_{1}(p^{0}) \ge \ldots \ge N_{\cmax}(p^{0}) = 1 \) for all \( p^{0} \). Finally, let \( N_{0}(p^{0}) := \infty \) for all \( p^{0} \) for convenience. 

\begin{proposition}
\label{prop:para-net-opt-control}
Consider a parallel network \( ( \{1, 2\},  \mathcal{E}) \) with link weights $ w \in \spreal^{ \mathcal{E}} $, flow capacities $ c  \in \spreal^{ \mathcal{E}} $ and initial supply demand vector \( p^{0} \). 
%and let \( R_{i} \), \( i\in \mathcal{E} \), \( \clink_{j} \), \( j\in [\cmax] \), \( N_{k}(p^{0}) \), \( k\in [\cmax] \cup \{0\} \) be \kscomment{as defined above}. 
If \( N_{j}(p^{0}) \le N < N_{j-1}(p^{0}) \), then an optimal control action is as follows: \( u^{t, *} = p^{0} \) for all \( 0 \le t < N_{j}(p^{0}) -2 \) and \( u^{t, *} =  \min\{\pnetval_{\clink_{j}}, p_{1}^{0}\} [ 1 \; - 1] \) for all \(  \max\{N_{j}(p^{0}) -2, 0\}  \le t \le N-1 \).
\end{proposition}
\begin{remark}
While Proposition~\ref{prop:para-net-opt-control} gives the explicit expression for a one-shot control that is optimal for parallel networks, the optimality of one-shot control in more general settings is proven in \cite{Ba.Savla.CDC16}.
\end{remark}

\begin{corollary}
\label{coro:para-net-constant-control}
For a parallel network \( ( \{1, 2\},  \mathcal{E}) \) with link weights $ w \in \spreal^{ \mathcal{E}} $, flow capacities $ c  \in \spreal^{ \mathcal{E}} $ and initial supply demand vector \( p^{0} \), for \( N \ge | \mathcal{E}| - \clink_{1} \), the following constant control \( u^{*} \) is an optimal control: \( u^{t, *} = [ 1 \; - 1] \min \{ p^{0}_{1}, \pnetval_{ \clink_{1}}\} \) for all $0\le t\le N-1$. 
\end{corollary}

On one hand, Proposition~\ref{prop:para-net-opt-control} and Corollary~\ref{coro:para-net-constant-control} justify the study of optimal control within a special class of control policies. On the other hand, while a one-shot control action can be optimal for non-parallel networks, it is not true in general. Furthermore, the sets of feasible constant and one-shot control actions are not necessarily closed nor connected. These are illustrated in the following example. 
%Example~\ref{eg:one-shot-control}. 

\begin{example}
\label{eg:one-shot-control}
Consider the network illustrated in Fig.~\ref{fig:diamond-net}, containing a single supply node \( 1 \) and a single demand node \( 3 \), having link weights \( w =\onebf \), and with initial supply-demand vector \( p^{0} = \trans{[3,\, 0,\, -3]} \). Consider two scenarios corresponding to link capacities \( c^{1} =  \trans{[0.8,\, 1.5,\, 0.6,\, 0.5,\, 0.25]}  \) and \( c^{2} = \trans{[0.8,\, 1.5,\, 0.7,\, 0.5,\, 0.25]} \), where note that the two scenarios differ only in the capacity of link \( e_{3} \). We consider the optimal control problem for \( N =2 \). Let  \( u^{t}=z_{t} \times \trans{[1, 0, -1]} \), \( t \in \{0, 1\} \), \( 0\le z_{1} \le z_{0} \) be the control actions. The flow under \( u^{t} \) for relevant network topologies are: \( f([5], u^{t}) = \frac{1}{4} z_{t} \times \trans{[1, 2, 1, 1, 1]} \), \( f([4], u^{t}) = \frac{1}{5} z_{t} \times \trans{[1, 3, 2, 1]} \), \( f([3], u^{t}) = \frac{1}{3} z_{t} \times \trans{[1, 2, 1]}\) and \( f(\{2\}, u^{t}) = z^{t} \). The maximal value of \( z_{t} \) can be supported by these networks are, respectively, \( 1, 1.5, 1.8, 1.5 \) in the first scenario and \( 1, 1.75, 2.1, 1.5\) in the second scenario. It is straightforward to see that the network would get disconnected in both scenarios if no load shedding is implemented. By considering all possible topology sequences that can occur under a control policy, we obtain the following: 
\begin{enumerate}[(i)]
\item The best one-shot control that sheds loads at \( t=1 \) is \( z_{0} = 3 \), \( z_{1} = 1.5 \) in both scenarios.
\item The best constant controls are: \( z_{0} = z_{1} = 1.5 \) in the first scenario and \( z_{0} = z_{1} = 2.1 \) in the second scenario.
\item An optimal control is \( z^{*}_{0} = 2.1 \) and \( z^{*}_{2} = 1.8 \) in the first scenario and is \( z_{0}^{*} = z^{*}_{1} = 2.1  \) in the second scenario. 
\end{enumerate}
We can see that in both cases, the best constant controls perform no worse than the best one-shot controls, and while the  best constant control is not optimal over all controls in the first scenario, it is optimal in the second scenario. Furthermore, in the second scenario, the set of feasible constant controls is \(\{(x, x) \,|\, x \in [0, 1]\cap (2, 2.1] \} \) and is neither connected nor closed.
%\begin{figure}[htbp]
%\centering
% \includestandalone[width=0.3\linewidth]{single-sd-example}
%\caption{The graph topology for the network used in Example~\ref{eg:one-shot-control} to illustrate that the set of feasible one-shot control actions is neither connected nor closed.}
%\label{fig:diamond-net}
%\end{figure}
\end{example}

\subsection{A Decomposition Approach for Tree Reducible Networks}
\label{sec:tree-reducible-net}
In this section, we develop a decomposition approach to compute optimal control for \emph{tree reducible networks}. 

\begin{definition}[Tree Reducible Network]
\label{def:tree-reducible-net}
A network $ \mathcal{G} = ( \mathcal{V}, \mathcal{E} )$ with supply-demand vector \( p \) is called tree reducible if it is a tree, or it can be reduced to a tree\footnote{A tree is an undirected graph in which any two nodes are connected by at most one path.} \( \mathcal{T} = ( \mathcal{V}_{T}, \mathcal{E}_{T}) \) by recursively replacing subnetworks, which is between two nodes and contains no supply or demand nodes in the interior, with single links. In this case, the subnetworks and \( \mathcal{T} \) are called, respectively, the reducible components and reduced tree of \( \mathcal{G} \). 
\end{definition}

\begin{figure}[htbp]
\centering
 \begin{subfigure}[t]{.5\linewidth}
  \centering
  \includegraphics[width=0.8\linewidth]{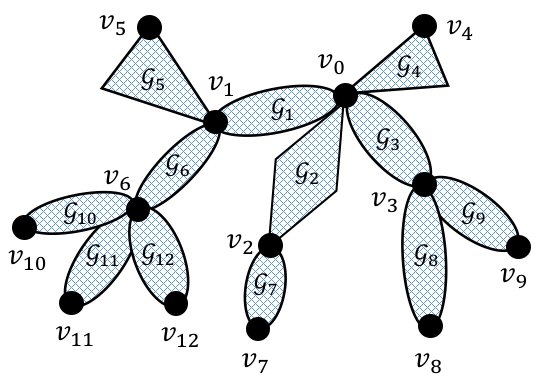}
  \caption{}
  \label{fig:tree-reducible-net}
\end{subfigure}
 \begin{subfigure}[t]{.4\linewidth}
  \centering
  \begin{tikzpicture} 
 \node[main node] (0) at (0, 0) { 0 };
 
\node[main node] (1) at (-1.5, -1.5) {1};
\node[main node] (2) at (-0.5, -1.5) {2};
\node[main node] (3) at (0.5, -1.5) {3};
\node[main node] (4) at (1.5, -1.5) {4};

\node[main node] (5) at (-2.5, -3) {5};
\node[main node] (6) at (-1.5, -3) {6};
\node[main node] (7) at (-0.5, -3) {7};
\node[main node] (8) at (0.5, -3) {8};
\node[main node] (9) at (1.5, -3) {9};

\node[main node] (10) at (-2.5, -4.5) {10};
\node[main node] (11) at (-1.5, -4.5) {11};
\node[main node] (12) at (-0.5, -4.5) {12};

 \path[main edge, edge label]
   (1)edge(0)  (2)edge(0) (3)edge(0) (4)edge(0) 
   (5)edge(1)  (6)edge(1) (7)edge(2) (8)edge(3) (9)edge(3) 
   (10)edge(6) (11)edge(6) (12)edge(6)
   ;      
\end{tikzpicture}
  \caption{}
  \label{fig:reduced-tree}
\end{subfigure}
\caption{(a) a tree reducible network \( \mathcal{G} \), which is the residual IEEE 39 bus network in Fig.~\ref{fig:ieee39-pic} resulting from removal of links \( e_{6} \), \( e_{16} \) and \( e_{40} \), e.g., caused by initial failure. It is assumed that all the nontransmission nodes in the IEEE 39 network are included in relabeled nodes \( v_{0}, \ldots, v_{12} \), e.g., \( v_{0} \) corresponds to node 16 in Fig.~\ref{fig:ieee39-pic}; (b) the reduced tree \( \mathcal{T} \) of \( \mathcal{G} \).}
\label{fig:def-tree-reducible-net}
\end{figure}
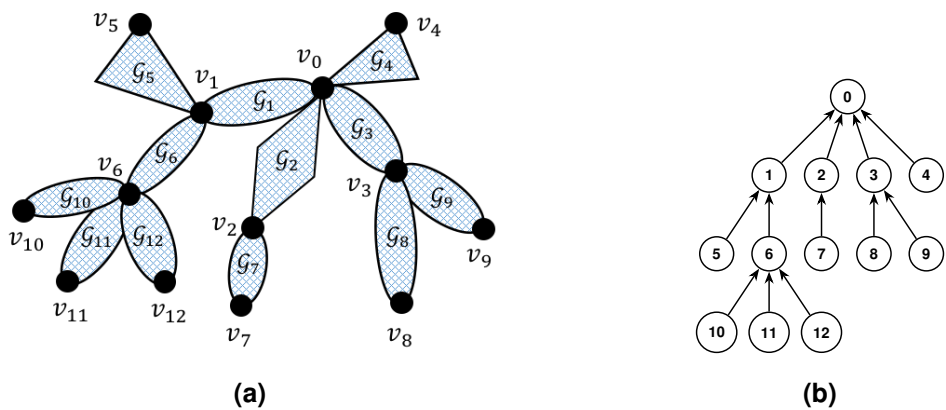

Fig.~\ref{fig:def-tree-reducible-net} provides an example of a tree reducible network, which is obtained from IEEE 39 bus network in Fig.~\ref{fig:ieee39-pic}. Each sub-network (denoted by $\mc G_{1}, \ldots, \mc G_{12}$) in Fig. \ref{fig:tree-reducible-net} represents a \emph{reducible component} and corresponds to a link in the \emph{reduced tree} \(  \mathcal{T} = (\mathcal{V}_{T}, \mathcal{E}_{T})\) shown in Fig.~\ref{fig:reduced-tree}. We assign directions for the links in $\mc E_T$ as follows\footnote{The results presented in the current Section~\ref{sec:tree-dp} do not depend on the particular choice of directions for links in $\mc E_T$, as selected here (see also Remark~\ref{rem:low-complexity}). %If some of the virtual links in $\mc E_T$ coincide with real links in $\mc E$, e.g., if $\mc E_T=\mc E$, then this is consistent with the remark in Section~\ref{sec:setup} that the directions for the links in $\mc E$ are chosen arbitrarily, and that the results are invariant with respect to the choice of directions.
}. Pick an arbitrary node in $\mc V_T$, and call it the root node. The directions for all the links incident to the root node are chosen to be incoming to the root node. The directions for the remaining links are similarly chosen to be directed towards the root node; see Figure~\ref{fig:reduced-tree} for an example. For the resulting directed tree, we fix a reverse topological ordering\footnote{That is, for every directed link $(i,j)$, we have $i>j$.} of the nodes $(0, 1, \ldots, |\mc V_T|-1)$, with $0$ being the root node. Figure~\ref{fig:reduced-tree} illustrates such an ordering. In order to minimize notations, we use the same label for a link and its tail node. For example, the link $(5,1)$ in  Figure~\ref{fig:reduced-tree} is labeled as link $5$. An in-neighbor (resp., out-neighbor) of a given node is called its child (resp., parent) node.
%We adopt a few standard terminologies for tree (e.g., see \cite{Godsil.Royle.01}).
%The \emph{root} is the top node and indexed as \( v_{0} \). 
%A \emph{child} node of a (\emph{parent}) node is a node directly connected to the node \kscomment{when moving away from the root}. 
For node \( i\in \mathcal{V}_T \), let \( \mathcal{C}_{i} \) denote the set of its children nodes, and let \( \bar{\mathcal{C}}_{i} \) 
%\( \bar{\mathcal{C}}_{i} := \setdef{ j \in \mathcal{V}}{j = i \text{ or } j \in \mathcal{C}_{k} \text{ for some } k \in \bar{\mathcal{C}}_{i}} \) 
denote the set of nodes consisting of the descendants of $i$ and the node $i$ itself. For example, in Fig.~\ref{fig:reduced-tree}, \( \mathcal{C}_{1} = \{5, 6\} \) and \( \bar{\mathcal{C}}_{1} = \{1, 5, 6, 10, 11, 12\} \). With this definition, \( \mathcal{V}_{T} \equiv \bar{\mathcal{C}}_{0} \).  Node $i$ is called a \emph{leaf} if it has no child node, i.e., if \( \mathcal{C}_{i}= \emptyset \).  
%Since every link in a tree network connects a child node and a parent node, we set the convention that the direction of the link is from the child to its parent, and label a link with the index of the child node associated to it. For example, in Fig.~\ref{fig:reduced-tree}, the link from \( v_{7} \) to \( v_{5} \) is labeled as \( e_{7} \). 

%We now describe how to translate the constraints on control sequence for the original network $\mc G$, as encapsulated by \eqref{eq:feasible-set-def}, into equivalent constraints for $\mc T$. 
Let us start with the simple case when the entire network consists of a single reducible component, say $\mc G_i$, so that the reduced tree is $\mc T = (\{v_{0}, v_{1}\},1)$, where $v_0$ and $v_1$ are the only supply-demand nodes in $\mc G_i$. Let \( a_{i} \in \{1, 0, -1\}^{ \mathcal{V}_{i}} \) be such that $a_{i,v}$ is equal to $1$ if $v=v_0$, is equal to $-1$ if $v=v_1$, and is equal to zero otherwise. We do not fix the individual identities of $v_0$ and $v_1$ as supply or demand nodes, and the choice of the signs of entries of $a_i$ is merely to set some convention. We let $z_i^t a_i$ with $z_i^t \in \real$ denote the supply-demand vector in $\mc G_i$ for $t \in [N]$, or equivalently, the control sequence for $t \in \{0, \ldots, N-1\}$. If $\tilde{p}_i a_i$, for $\tilde{p}_i \in \real$, is the initial supply demand vector, then the set of feasible control sequences as per \eqref{eq:feasible-set-def} is $\mc D(\mc E_i^{0}, \tilde{p}_ia_i, N)$, where $\mc E_i^0$ is the initial active link set in $\mc G_i$ at $t=0$. Recall that $\mc D(\mc E_i^{0}, \tilde{p}_i a_{i}, N)$ captures the constraint that the terminal state at $t=N$ is feasible, as well as the monotonicity constraint implied by \eqref{def:control-space}. We split these two constraints as $\mc D(\mc E_i^{0}, \tilde{p}_i a_i, N) = \Drelone_i(N) \cap \Dreltwo_{i}(N) $, where $\Drelone_i(N)$ captures feasibility of terminal state, while relaxing monotonicity, and $\Dreltwo_{i}(N)$ captures monotonicity while relaxing feasibility of the terminal state. These two sets are formally defined as:
\begin{align}
\Drelone_i(N):=\{(z^{0}_{i}, \ldots, z^{N-1}_{i}) \in \real^N\, & |\,  (\mc E_i^{N-1},z_i^{N-1} a_{i}) \in \mc S; \nonumber \\
 &\mc E_i^t = \mathcal{F}_{ \mathcal{E}} ( \mathcal{E}_{i}^{t-1}, z_{i}^{t-1} a_{i}), \forall t \in [N-1] \} \label{eq:Drel-def} \\
\Dreltwo_i(N)  := \{(z_i^{0}, \ldots, z_i^{N-1}) \in \real^{N} \, &|\, z_i^{0} \in \cube \tilde{p}_i, z_i^{t} \in \cube{z_i^{t-1}},\, \, \forall t \in [N-1] \}  \label{eq:Si-def}
\end{align}

When clear from the context, we shall not show the explicit dependence of $\Drelone_i$ and $\Dreltwo_i$ on $N$.

\begin{remark}
\label{rem:d-rel-set}
\leavevmode
\begin{enumerate}
\item Note that $\Drelone_i$ includes control actions that cause loss of connectivity in $\mc G_{i}$. In this case, since we have only one supply and demand, the constraint that the terminal state $(\mc E^{N-1}_i, z_i^{N-1} a_{i})$ is feasible, implies that $z_i^{N-1}=0$. In addition, since all links have symmetrical capacities, \( \Drelone_i = -\Drelone_i \). 
\item 
%Comparing \eqref{eq:feasible-set-def} with \eqref{eq:Drel-def} and \eqref{eq:Si-def}, we get that $\mc D(N,\mc E_i^{0}, p_i^0) = \Drelone_i(N) \cap \Dreltwo_{i}(N) $ for all \( i\in \mathcal{E}_{T} \). 
$\Dreltwo_i$ is a polytope. However, \( \Drelone_{i} \) is non-convex in general, as indicated by the disconnected feasible set of constant control actions in Example~\ref{eg:one-shot-control}. 
%\kscomment{(mention what points are excluded from \( \Drelone_{i} \))}. Nevertheless, $\Drelone_i$ is bounded as the link capacities are assumed to be finite. 
The explicit computation of $\Drelone_i$ follows from the discussion in Section~\ref{sec:tree-search} (cf. Remark~\ref{rem:arrangement-construction}).
\end{enumerate}
\end{remark}

The flexibility afforded by splitting the control constraints into \eqref{eq:Drel-def} and \eqref{eq:Si-def} for an isolated reducible component allows to translate capacity constraints from individual links in a general tree reducible network $\mc G$ into equivalent constraints for the equivalent links in $\mc E_T$ as follows. The control $u_i^t$ \footnote{With a slight abuse of notation, we use $u$ to denote control inputs for the original network as well as the reduced network.} at node $i \in \mc V_T$ at time $t \in \{0, \ldots, N-1\}$ is split as $u_i^t=z_i^t-\sum_{j \in \mc C_i}z_j^t$. For example, referring to Figure~\ref{fig:def-tree-reducible-net}, $u^t_{3}=z^t_{3}-z^t_{8}-z^t_{9}$. In this case, the flow on link $i$ in $\mc E_T$, $z_i^t$, is interpreted as $\mc G_i$'s \emph{share} of control input $u_i^t$. 
$u_i:=(u_i^0,\ldots,u_i^{N-1})$, $i \in \mc V_T$, is constrained to satisfy \eqref{eq:Si-def} for $\tilde{p}_i=p_i^0$, and $z_i:=(z_i^0,\ldots,z_i^{N-1})$, $i \in \mc E_T$, is constrained to satisfy \eqref{eq:Drel-def}. It is clear that, for root node, \( z_{0} = \zerobf \). Hence, let \( \Drelone_0= \{ \zerobf \} \) for convenience.

Consider the following optimization problem that will inform the decomposition approach. For $i \in \mc E_T$, given \( z_{i} \in \Drelone_i \), let: \( \starobj_{i}(z_{i}) :=  \)
\begin{equation}
\label{opt:tree-dp-problem}
\begin{array}{rcl} 
%& \displaystyle{ \max_{z_{k}\in \real^{N}\, \forall k\in \bar{\mathcal{C}}_{i}\setminus i, y_{k}\in \real^{N} \, \forall k\in \bar{\mathcal{C}}_{i} }}  & \sum_{k\in \bar{\mathcal{C}}_{i}} \sdsign_{k} y_{kN}  \\ 
& \displaystyle{ \sup_{\substack{z_{k} \in \Drelone_k \, \forall k\in \bar{\mathcal{C}}_{i}\setminus i \\ u_{k} \in \Dreltwo_k \, \forall k\in \bar{\mathcal{C}}_{i} }}}  & \sum_{k\in \bar{\mathcal{C}}_{i}} \sdsign_{k} u^{N-1}_{k}  \\ 
& \mathrm{s.t.} & z_{k} = u_{k} + \sum_{j \in \mathcal{C}_{k}} z_{j} , \quad \forall k\in \bar{\mathcal{C}}_{i} \\ 
%&& z_{k} \in \Drelone_k \, , \quad \forall k\in \bar{\mathcal{C}}_{i}\setminus i  \\
%&& u_{k} \in \Dreltwo_k \, ,\quad \forall k\in \bar{\mathcal{C}}_{i} \\
\end{array} 
\end{equation}
\eqref{opt:tree-dp-problem} can be interpreted as maximizing a certain \emph{utility function} over the subtree rooted at node \( i \in \mc V_T \), given that the outflow sequence from node \( i \) is \( z_{i}\in \real^{N} \). \eqref{opt:tree-dp-problem} is a generalization of \eqref{eq:control-formulation-new}, in the sense that \( \starobj_{0}(\zerobf) \) is equal to the optimal value of \eqref{eq:control-formulation-new}. Since the objective function of \eqref{opt:tree-dp-problem} is linear and separable and the decision variables are coupled with only equality constraints of a simple form, standard distributed algorithms, for example, ADMM \cite{Boyd.Parikh.ea:FTML11}, can be used to solve \eqref{opt:tree-dp-problem} if the sets \(\Drelone_k  \) are convex. However,  \(\Drelone_k  \) are non-convex in general (cf. Remark~\ref{rem:d-rel-set}). In order to reduce the complexity due to this non-convexity, we decompose \eqref{opt:tree-dp-problem} into the following \emph{nested form}:  
\begin{equation}
\label{opt:tree-dp-value-iteration}
\begin{array}{rcl} 
\starobj_{i}(z_{i}) = & \displaystyle{ \sup_{\substack{z_{j} \in  \mathcal{Z}_{j}  \, \forall j\in \mathcal{C}_{i} \\ u_{i}\in \Dreltwo_{i} }}}  & \sdsign_{i} u^{N-1}_{i} + \sum_{j\in \mathcal{C}_{i}} \starobj_{j}(z_{j})  \\ 
& \mathrm{s.t.} & z_{i} = u_{i} + \sum_{j \in \mathcal{C}_{i}} z_{j} \\ 
\end{array} 
\end{equation}
where \( \mathcal{Z}_{j}: = \Drelone_{j} \cap \left( \Dreltwo_{j} + \sum_{k\in \mathcal{C}_{j}} \mc Z_{k} \right) \) combines both the constraints \( z_{j} \in \Drelone_{j} \) and \( z_{j} = u_{j} + \sum_{k\in \mathcal{C}_{j}} z_{k} \), for all \( j\in \mathcal{E}_{T} \). The equivalence between \eqref{opt:tree-dp-problem} and \eqref{opt:tree-dp-value-iteration} can be seen via induction. In particular, if \( i\) is a leaf node, then \( \mathcal{C}_{i} = \emptyset \). Both \eqref{opt:tree-dp-problem} and \eqref{opt:tree-dp-value-iteration} reduce to \( \starobj_{i}(z_{i}) = \max_{z_{i} = u_{i} \in \Dreltwo_{i} } \sdsign_{i} u^{N-1}_{i} \). The optimal value function is \( \starobj_{i}(z_{i}) = \sdsign_{i} z_{i}^{N-1} \)  if \( z_{i} \in \Drelone_{i} \cap \Dreltwo_{i} = \mathcal{Z}_{i}  \), and is \( -\infty \) otherwise due to infeasibility. If $i$ is not a leaf node, then \eqref{opt:tree-dp-value-iteration} can be interpreted as being associated with a local \emph{star} subnetwork in $\mc T$. 

The solution to \eqref{opt:tree-dp-problem} is obtained by solving sub-problems in \eqref{opt:tree-dp-value-iteration} over two iterations:
\begin{enumerate}[(I)]
\item  Compute $\map{\starobj_i}{\mc Z_i}{\real}$ via \eqref{opt:tree-dp-value-iteration} for every $i \in \mc V_T$ in the reverse topological order;
%
% given \( \starobj_{j}(z_{j}) \), \( z_{j}\in \mc Z_{j} \), computed already from previous \kscomment{iterations} for all \( j\in \mathcal{C}_{i} \). 
\item Set $z_0^*=\zerobf$. Following the topological order, for every $i \in \mc V_T$, compute an optimal solution $(u_i^*,\{z_j^*, j \in \mc C_i\})$ to \eqref{opt:tree-dp-value-iteration} corresponding to $\starobj_i(z_i^*)$.
%
% The first iteration ends with root node computing a solution of \( y^{*}_{0} \) and \( z^{*}_{j} \), \( j\in \mathcal{C}_{0} \), for \( z_{0} = \zerobf \).
% The values of \( z^{*}_{j} \) for the last sub-problem are then broadcast to its children nodes \( j\in \mathcal{C}_{0} \), and thereafter in the second iteration, from the root to leaves, all the nodes \( i \) sequentially compute the solution \( y^{*}_{i} \) and \( z^{*}_{j} \), \( j\in \mathcal{C}_{i} \) using \eqref{opt:tree-dp-value-iteration}, given the value of \( z^{*}_{i} \) received from its parent node. 
\end{enumerate}
It is easy to check that, for all \( i\in \mathcal{V}_{T} \), we have \( \zerobf \in \Dreltwo_{i} \) and \( \zerobf \in \Drelone_{i} \), and hence \( \zerobf \in \mc Z_{i} \). Therefore, iteration (II) is well-posed.

%With this recursion, \eqref{opt:tree-dp-problem} is solved in two iterations over the reduced tree. In the first iteration, 
% Take the network in Fig.~\ref{fig:def-tree-reducible-net} for example, firstly \( J_{5} \) and \( J_{1} \) are to be computed from the subnetworks in Fig.~\ref{fig:tree-decom-sub5} and Fig.~\ref{fig:tree-decom-sub1}, respectively. Then \( J_{2} \) is to be computed from the subnetwork in Fig.~\ref{fig:tree-decom-sub2}, given function \( J_{2} \) is known.  In the example in Fig.~\ref{fig:def-tree-reducible-net}, the last sub-problem is for the subnetwork in Fig.~\ref{fig:tree-decom-sub0}. In the example in Fig.~\ref{fig:def-tree-reducible-net}, the value broadcasting process goes from subnetwork in Fig.~\ref{fig:tree-decom-sub0}, to Fig.~\ref{fig:tree-decom-sub2}, to Fig.~\ref{fig:tree-decom-sub1} and to Fig.~\ref{fig:tree-decom-sub5}. 

% The procedure happen in a distributed fashion and can be seen as a leaf elimination process, as demonstrated in Fig.~\ref{fig:tree-dp-illustration}. 
%\ksmargin{move Remark~\ref{rem:low-complexity}(i) to some other place?}
\begin{remark}
\label{rem:low-complexity}
The optimal solution to \eqref{eq:control-formulation-new} as computed by the decomposition approach is invariant with respect to the choice of root node, directions of the links in $\mc E_T$, and the topological ordering used for labeling nodes in $\mc V_T$.

\end{remark}

While the decomposition approach reduces complexity of solving \eqref{eq:control-formulation-new} for general tree reducible networks, the bottleneck is still non-convexity of the local problems in \eqref{opt:tree-dp-value-iteration}. In Section~\ref{sec:invariance}, we show that, for one-shot controls, the two iterations involving solutions to the local problems in \eqref{opt:tree-dp-value-iteration} admit closed-form expressions. 
\subsection{Optimal One-shot Control for Tree Reducible Networks}
\label{sec:invariance}
The next result follows straightforwardly from Definition~\ref{def:one-shot-and-constant-control}. For a general network,  it shows how to reduce the problem of finding an optimal one-shot control to that of finding an optimal constant control. \begin{lemma}
\label{lem:one-shot-constant-control}
Consider the optimal control problem in \eqref{eq:control-formulation-new}, restricted to one-shot control actions, over $N$ stages with initial state $(\mathcal{E}_{\mathrm{un}}^{0}, p^{0})$. 
%\( N \) stage optimal control problem for a network \( (\mathcal{V}, \mathcal{E}) \) with initial supply-demand vector \( p^{0} \). 
Let \( ( \mathcal{E}_{\mathrm{un}}^{0}, \ldots, \mathcal{E}_{\mathrm{un}}^{n}) \), \( n\le N-1\), be the longest topology sequence of the uncontrolled cascading dynamics \eqref{eq:cascade-dynamics} such that \( p^{0} \in \mathcal{B}_{\mathcal{E}_{\mathrm{un}}^{t}} \) for \( 0\le t \le n \) \footnote{A sufficient, but not necessary, condition for \( p^{0} \notin \mathcal{B}_{\mathcal{E}_{\mathrm{un}}^{t}} \) at some $t$ is that some non-transmission node becomes isolated at $t$.}. Let  \( \bar{u}( \mathcal{E}_{\mathrm{un}}^{t}, N-t) \) be an optimal solution to \eqref{eq:control-formulation-new}, restricted to constant control actions, over \( N-t \) stages with initial state $(\mathcal{E}_{\mathrm{un}}^{t}, p^0)$. Then, the one-shot control which sheds load to \( \bar{u}^{0}( \mathcal{E}_{\mathrm{un}}^{t^{*}}, N-t^{*}) \) at stage \( t^{*} \in \argmax_{0\le t \le n} \trans{\sdsign} \bar{u}^{0}( \mathcal{E}_{\mathrm{un}}^{t}, N-t) \) is optimal.
%Let \( t^{*} \in \argmax_{0\le t \le n} \trans{\sdsign} \bar{u}^{0}( \mathcal{E}_{\mathrm{un}}^{t}, N-t) \) be a stage in $\{0, \ldots, n\}$ at which one can shed the smallest load to ensure network feasibility at $t=N-1$. In other words, a one-shot control which sheds load to \( \bar{u}^{0}( \mathcal{E}_{\mathrm{un}}^{t^{*}}, N-t^{*}) \) at stage \( t^{*} \) is optimal.
\end{lemma}

Following Lemma~\ref{lem:one-shot-constant-control}, it is sufficient to focus on constant controls. In this case, the original problem \eqref{opt:tree-dp-problem} and the sub-problem \eqref{opt:tree-dp-value-iteration} are simplified to a great extent: for all \( i\in \mathcal{V}_{T} \), the decision variables \( z_{i} \) and \( u_{i} \) reduce to one dimensional real numbers; the set \( \Dreltwo_{i} \) reduces to \( \cube{p^{0}_{i}} \) and the set \( \Drelone_{i} \) reduces to, in general, the union of multiple disjoint intervals. While Remark~\ref{rem:arrangement-construction} explains how to compute \( \Drelone_{i} \) in this case, we make the following remarks. 
\begin{remark}
\label{rem:feasible-equiv-flow}
\leavevmode
\begin{enumerate}
\item For \( N=1 \), it is straightforward to see that \( \Drelone_{i} \) is always a single piece of closed interval. For \( 2 \le N \le 6\), we computed \( \Drelone_{i} \) for constant control corresponding to all the node pairs in the IEEE 39 bus network, and found that they were disconnected in less than 5\% of the cases. Based on this, we believe the same pattern to hold true also for larger \( N \) as well as other benchmark networks. In fact, we had to choose the parameters very carefully in Example~\ref{eg:one-shot-control} to show disconnectedness of  \( \Drelone_{i} \) under constant control.
\item Even in the case of constant controls, \( \Drelone_{i} \) can be half open (see Example~\ref{eg:one-shot-control}) and hence \eqref{opt:tree-dp-problem} may not admit a solution. Nevertheless, because the objective function in \eqref{opt:tree-dp-problem} is linear and hence continuous, we use the closure of \( \Drelone_{i} \) in \eqref{opt:tree-dp-problem} and hence closure of \( \mathcal{Z}_{i} \) in \eqref{opt:tree-dp-value-iteration} for simplicity. Proposition~\ref{prop:close-solution} shows, for a more general setting, how to obtain an interior feasible point that is arbitrarily close to the solution of the problem over the closure. 
\end{enumerate}
\end{remark}

We consider the following generalized problem of \eqref{opt:tree-dp-value-iteration} for the case of constant controls. In particular, we will provide results on the solution structure of \eqref{opt:lp-meta-star} below which enable closed form computation for every recursive step in the two iteration algorithm proposed in Section~\ref{sec:tree-reducible-net}, as far as constant controls are considered. 
\begin{equation}
\label{opt:lp-meta-star}
\begin{array}{r} 
\starobjout(z) = \starop \{(\starobjin_{j}, X_{j})\}_{j \in [n]} :=  \displaystyle{ \max_{x \in \real^{n}}}   \sum_{j=1}^{n} \starobjin_{j}(x_{j})  \\ 
 \mathrm{s.t.} \quad \trans{\onebf} x = z; \quad x_{j} \in X_{j},\; \forall\, j\in [n] \\
\end{array} 
\end{equation}
%\begin{equation}
%\label{opt:lp-meta-star}
%\begin{array}{rcl} 
%\starobjout(z) = \starop \{(\starobjin_{j}, X_{j})\}_{j \in [n]} := & \displaystyle{ \max_{x \in \real^{n}}}  & \sum_{j=1}^{n} \starobjin_{j}(x_{j})  \\ 
%& \mathrm{s.t.} & \trans{\onebf} x = z; \quad x_{j} \in X_{j},\; \forall\, j\in [n] \\
%\end{array} 
%\end{equation}
where \( X_{j} \subset \real \) is the union of a finite number of disjoint closed intervals for all \( j\in [n] \). Operator \( \starop \) maps  from \( n \) input functions \( \starobjin_{j} \) with restricted domain \( X_{j} \) to a single output function \( \starobjout \) with domain \( \sum_{j=1}^{n} X_{j} \), where the domain of \( \starobjout \) is not explicitly written in \eqref{opt:lp-meta-star}. Due to the possible disconnected feasible set \( X_{j} \), \eqref{opt:lp-meta-star} is in general non-convex. 
\begin{remark}
\label{rem:correspondence}
\eqref{opt:lp-meta-star} becomes \eqref{opt:tree-dp-value-iteration} with the following substitutions: \( i = 1 \), \( \mathcal{C}_{i} = [n]\setminus \{1\} \), \( \starobjin_{1}(x_{1}) = \sdsign_{i} x_{1}\), \( X_{1} =  \cube{p^{0}_{i}}  \), \( \starobjin_{j}(x_{j}) = \starobj_{j}(x_{j}) \) and \( X_{j} = \mathcal{Z}_{j}   \) for all \( j \in [n]\setminus \{1\} \). 
\end{remark}

The class of input functions considered depends on the following construct. For a given point \( \tau= (\tau_{1}, \tau_{2}) \in \real^{2} \), define a (continuous) piecewise affine function: \( \map{\roofun_{\tau}}{ \real }{\real} \) as:
\begin{equation} 
\label{eq:l-fun}
 \roofun_{\tau}(x) := \left\{
 \begin{array}{ll}
  x - \tau_{1} + \tau_{2}  \quad &   x \le \tau_{1} \\
  - x + \tau_{1} + \tau_{2}  &    x > \tau_{1} \\
 \end{array}\right.
\end{equation}
As shown in Fig~\ref{fig:l-fun}, function \( \roofun_{\tau} \) contains two rays joining at \( \tau \), referred to as the \emph{top point} of \( \roofun_{\tau} \), and intersects the vertical axis at $(0,\tau_2-|\tau_1|)$.

\begin{figure}[htbp]
\centering
  \begin{tikzpicture} 
\draw[<-] (0, 6) node[left] {{\huge \( \chi_{\tau}(x) \)}} -- (0,0) node [below left] {{\huge 0}} -- (0, 0.3);
\draw[thick] (-2, 2) to (1, 5) to (3, 3);

%\draw [->] (-3,0) -- (4,0) node [below right] {\( z \)}; 
%\draw[dashed] (1, 5) -- (0, 5) node[left] {\( \mathbf{1}^{T}(p^{u} - p^{l}) \)}; 
%\node[left] at (0, 4) {\(- 2 \mathbf{1}^{T} p^{l}\) };
%\draw[dashed] (-2, 2) -- (-2, 0) node[below]{\( \mathbf{1}^{T} p^{l} \)};
%\draw[dashed] (3, 3) -- (3, 0) node[below]{\( \mathbf{1}^{T} p^{u} \)};
%\draw[dashed] (1, 5) -- (1, 0) node[below]{\( \mathbf{1}^{T} (p^{u} + p^{l}) \)};

%\draw [->] (-3.5,0) -- (5,0) node [below right] {\( z \)}; 
%\draw[dashed] (1, 5) -- (0, 5) node[left] {\( \beta_{v_{1}} + |\tilde{p}_{v_{1}}| \)}; 
%\node[left] at (0, 4) {\( \beta_{v_{1}} \) };
%\draw[dashed] (1, 5) -- (1, 0) node[below]{\( \tilde{p}_{v_{1}} \)};
%\draw[dashed] (-2, 2) -- (-2, 0) node[below]{\( \frac{1}{2}( \tilde{p}_{v_{1}} - \beta_{v_{1}} - | \tilde{p}_{v_{1}}| ) \)};
%\draw[dashed] (3, 3) -- (3, 0) node[below]{\(  \frac{1}{2}( \tilde{p}_{v_{1}} + \beta_{v_{1}} + | \tilde{p}_{v_{1}}| )   \)};

\draw [->] (-3,1) -- (4.5,1) node [below right] {{\huge\( x \)}}; 
%\draw[dashed] (1, 5) -- (0, 5) node[left] {\( q^{u} - q^{l} \)}; 
%\node[left] at (0, 4) {\(- 2q^{l} \) };
%\draw[dashed] (-2, 2) -- (-2, 0) node[below]{\(q^{l} \)};
%\draw[dashed] (3, 3) -- (3, 0) node[below]{\( q^{u} \)};
%\draw[dashed] (1, 5) -- (1, 0) node[below]{\( q^{l}+q^{u} \)};

%\draw[dashed] (1, 5) node[above] {\( x \) } -- (0, 5) node[left] {\( x_{2} \)}; 
%\draw[dashed] (1, 5) -- (1, 0) node[below]{\( x_{1} \)};
%\draw[dashed] (-2, 2) -- (-2,1) node[below]{{\huge\( b_{1} \) }};
%\draw[dashed] (3, 3) -- (3,1) node[below]{{\huge \( b_{2} \) }};
\draw[fill] (1, 5) circle (0.05cm) node[above] {{\huge \( \tau \) }}; 
\draw[fill] (0, 4) circle (0.05cm) node[above left] {{\huge \( \tau_{2} - |\tau_{1}| \) }};

%\draw[gray, thick, dotted] (-2, 2) to (0, 0) to (3, 3); 
\end{tikzpicture}
\caption{Illustration of \( \roofun_{\tau}(x) \) defined in \eqref{eq:l-fun}.}
\label{fig:l-fun}
\end{figure}
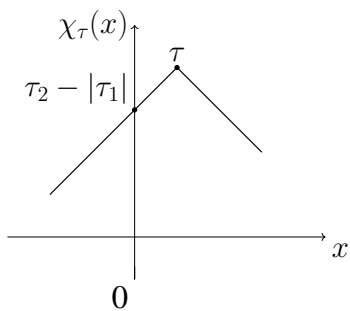

It will be useful to consider a \( \roofun \) function defined over a restricted domain. When the restricted domain is a closed interval, Lemma~\ref{lem:lx-extension} can be used to translate the top point into the domain, if not already inside, without changing function values. It is straightforward that, with bounded domain, \( \roofun \) functions include linear functions as special cases. 
% Lemma~\ref{lem:lx-extension} implies that for any \( \tau \) and \( [b_{1}, b_{2}] \), there exists \( \tau' \in [b_{1}, b_{2}] \) such that \( \roofun_{\tau'}(x) = \roofun_{\tau}(x) \) for all \( x\in [b_{1}, b_{2}] \). Therefore, when writing $\map{\roofun_{\tau}}{[b_{1}, b_{2}]}{\real}$, we implicitly assume that $\tau_1 \in [b_{1}, b_{2}]$. Moreover, if \kscomment{\( b_{1} < \tau_{1} < b_{2} \)}, then \( \roofun_{\tau} \) contains two line segments joining at \( \tau \) (see Fig~\ref{fig:l-fun}). If \( \tau= b_{1} \) or \( \tau=b_{2} \), then \( \roofun_{\tau} \) is a linear function. If \( b_{1} = b_{2} \), then by the convention \( \tau \in [b_{1}, b_{2}] \) \( \tau_{1} = b_{1} \) and hence \( \roofun_{\tau} \) contains the single point \( \tau \). 
%Extension of \eqref{eq:l-fun} when $x \notin [b_{1}, b_{2}]$ is given in the next lemma, whose proof follows from simple algebra.
\begin{lemma}
\label{lem:lx-extension}
Consider a point \( \tau = (\tau_{1}, \tau_{2})\in \real^{2} \) and a closed interval \( [b_{1}, b_{2}] \subset \real \) such that \( \tau_{1}  \not\in [b_{1}, b_{2}] \), let \( \tau':= (b_{2}, b_{2}-\tau_{1}+\tau_{2}) \) if \( \tau_{1} > b_{2} \) and \( \tau' := (b_{1}, -b_{1}+\tau_{1}+\tau_{2}) \) if \( \tau_{1} < b_{1} \). Then \( \roofun_{\tau}(x) = \roofun_{\tau'} (x) \) for all \( x \in [b_{1}, b_{2}] \). 
%
%Consider the function $\map{\roofun_{\tau}}{[b_{1}, b_{2}]}{\real}$ in \eqref{eq:l-fun} for a given $\tau=(\tau_1,\tau_2) \in \real^2$, with $\tau_1 \in [b_{1}, b_{2}]$. The restriction of $\roofun_{\tau}$ to $[b_{1}',b_{2}'] \subseteq [b_{1}, b_{2}]$ is equal to $\roofun_{\tau'}$, where $\tau'=\tau$ if $\tau_1 \in [b'_{1},b'_{2}]$, \( \tau':= (b'_{1}, b'_{2}-\tau_{1}+\tau_{2}) \) if \( \tau_{1} > b'_{2} \) and \( \tau' := (b'_{1}, -b'_{1}+\tau_{1}+\tau_{2}) \) if \( \tau_{1} < b'_{1} \). Furthermore, in each case: (i)  $\tau'$ is the top point; and (ii) if $0 \in [b'_{1},b'_{2}]$, then $\roofun_{\tau'}$ intersects the vertical axis at $(0,\tau_2'-|\tau_1'|)$.
\end{lemma}
We first consider the case when \( X_{j} \) is a single piece of closed interval for all \(j\in [n] \). As we note in Remarks~\ref{rem:feasible-equiv-flow} and \ref{rem:correspondence}, we believe that this case is common in practice. In this case, the feasible set of \eqref{opt:lp-meta-star} is convex and hence \eqref{opt:lp-meta-star} is convex if input functions \( \starobjin_{j} \) are concave, e.g., if \( \starobjin_{j} \) are \( \roofun \) functions. The next result shows that \( \roofun \) functions defined over closed intervals are invariant through operator \( \starop \) by solving \eqref{opt:lp-meta-star} explicitly.
\begin{lemma}
\label{lem:meta-star-net}
If \( \starobjin_{j} \) is a \( \roofun \) function (cf. \eqref{eq:l-fun}) with top point \( \tau_{j} = (\tau_{j}^{1}, \tau_{j}^{2}) \) and \( X_{j} = [q_{j}^{l}, q_{j}^{u}] \) (\( q_{j}^{l} \le \tau_{j}^{1} \le q_{i}^{u} \)) for all \( j\in [n] \), then the following hold true for \( \starop \) defined in \eqref{opt:lp-meta-star}:
\begin{enumerate}[(i)]
\item \( \starobjout =  \starop \{(\starobjin_{j}, X_{j})\}_{j \in [n]} \) is a \( \roofun\) function with top point \(  \left(\trans{\onebf} \tau^{1}, \trans{\onebf} \tau^{2} \right) \) and domain \( [\trans{\onebf}q^l, \trans{\onebf}q^u] \). 
\item the set of maximizers is $\setdef{x^{*} \in [q^{l}, \tau^{1}]  }{ \trans{\onebf} x^{*} = z }$ if $\trans{\onebf} q^{l} \le z < \trans{\onebf} \tau^{1} $, is \( \{ \tau^{1} \} \) if \( z=  \trans{\onebf} \tau^{1} \), and is $\setdef{x^{*} \in [\tau^{1}, q^{u}]  }{ \trans{\onebf} x^{*} = z }$ if $ \trans{\onebf} \tau^{1} < z \le \trans{\onebf} q^{u}$.
\end{enumerate} 
where \( \tau^1:= \tau_{[n]}^{1} \), \( \tau^2:= \tau_{[n]}^{2} \), \( q^{l} := q^{l}_{[n]}\) and \( q^{u}:= q^{u}_{[n]} \) are \( n \) dimensional vectors. 
\end{lemma}
\begin{remark}
\label{rem:linearity-invariance}
 Lemma~\ref{lem:lx-extension} implies that the condition \( \tau_{j}^{1} \in X_{j} \) in Lemma~\ref{lem:meta-star-net} is without loss of generality. In addition, the top point of \( \starobjout \) is inside its domain, that is, \( \trans{\onebf} \tau^{1} \in  [\trans{\onebf}q^l, \trans{\onebf}q^u] \).
%\leavevmode
%\begin{enumerate}
%\item Lemma~\ref{lem:lx-extension} implies that the condition \( \tau_{j}^{1} \in X_{j} \) in Lemma~\ref{lem:meta-star-net} is without loss of generality. In addition, the top point of \( \starobjout \) is inside its domain, that is, \( \trans{\onebf} \tau^{1} \in  [\trans{\onebf}q^l, \trans{\onebf}q^u] \).
%\item Lemma~\ref{lem:meta-star-net} implies the following. If \( \starobjin_{j} \), \( j\in [n] \), are linear functions with identical slope \( 1 \) (and \( -1 \), respectively), then \( \starobjout \) is linear with the same slope \( 1 \) ( and \( -1 \), respectively). For a set \( S\subset [n] \), if \( \starobjin_{j} \) is a linear function with slope \( 1 \) (and \( -1 \), respectively) for all \( j\in S \) (and \( j\in [n] \setminus S \), respectively), then \( \starobjout \) contains two linear pieces with slope \( 1 \) over domain \( [\trans{\onebf}q^l, \trans{\onebf} \tau^{1}]  \) and slope \( -1 \) over domain \( [\trans{\onebf} \tau^{1}, \trans{\onebf}q^u] \). Moreover, for  \( z \in [\trans{\onebf} \tau^{1}, \trans{\onebf}q^u] \): \( x_{j}^{*} = q_{j}^{u} \) for all \( j \in S \), and for  \( z \in [\trans{\onebf}q^l, \trans{\onebf} \tau^{1}]  \):  \( x_{j}^{*} = q_{j}^{l} \) for all \( j \in [n]\setminus S \). 
%\item It is straightforward to see from the proof of Lemma~\ref{lem:meta-star-net} that the same results hold even if there exists  \( j \in [n] \) such that \( X_{j} \) is unbounded, that is \( q_{j}^{l} =  -\infty  \) and/or \( q_{j}^{u} =  \infty \). 
%\end{enumerate}
\end{remark}

We now consider the case when \( X_{j} \) is the union of a finite number of disjoint intervals for all \(j\in [n] \). Assume \( X_{j} \) contains \( m_{j} \) pieces of intervals and denote each piece by \( X_{j}^{k} \), then \( X_{j} = \cup_{k=1}^{m_{j}} X_{j}^{k} \). The way to solve the nonconvex problem \eqref{opt:lp-meta-star} is to decompose it into multiple subproblems, where each subproblem is associated with a combination \( \sigma \in \comset :=  \Pi_{j=1}^{n} [m_{j}] \subset \real^{n} \) of intervals \( X_{j}^{\sigma_{j}} \) and is denoted by \(  \starop(\sigma) := \starop \{(\starobjin_{j}, X^{\sigma_{j}}_{j})\}_{j \in [n]}  \). That is, 
\begin{equation}
\label{eq:domain-decomposition}
 \starop \{(\starobjin_{j}, X_{j})\}_{j \in [n]}  = \max_{ \sigma \in \comset}  \starop(\sigma)
\end{equation}
If the restriction of \( \starobjin_{j} \) to \( X_{j}^{k} \) is a \( \roofun \) function for all \( k\in [m_{j}] \) and \( j\in [n] \), then Lemma~\ref{lem:meta-star-net} can be used to solve subproblems \(  \starop(\sigma) \) for all \( \sigma \in \comset \). This motivates the following definition of \emph{piecewise \( \roofun \) functions}. 

\begin{definition}
\label{def:piecewise-chi-fun}
\( \map{g}{X\subseteq \real}{\real} \) is called a \emph{piecewise \( \roofun \) function} if the domain \( X \) is the union of multiple disjoint intervals, and over each of these intervals, \( g \) is a restricted \( \roofun \) function. 
%A single variable real valued function \( \map{g}{X\subseteq \real}{\real} \) is called a \emph{piecewise \( \roofun \) function} if it is composed of multiple segments defined over an equal number of disjoint intervals, such that each segment is a \( \roofun \) function, as defined in \eqref{eq:l-fun}, restricted over one of those intervals. 
\end{definition}
\begin{remark}
The domain \( X \) of a piecewise \( \roofun \) function, as defined in Definition~\ref{def:piecewise-chi-fun}, is not necessarily connected. 
\end{remark}
%Example~\ref{eg:piecewise-lx} shows that the output function from \( \starop \) is a piecewise \( \roofun \) function if the input functions are all \( \roofun \) function. The next result generalizes this observation and shows that piecewise \( \roofun \) functions are invariant through operator \( \starop \). 

The next result shows that piecewise \( \roofun \) functions are invariant through operator \( \starop \), whose proof follows straightforwardly from \eqref{eq:domain-decomposition}, Lemma~\ref{lem:meta-star-net}, Remark~\ref{rem:linearity-invariance} and the fact that the point-wise maximum of multiple \( \roofun \) functions is a piecewise \( \roofun \) function.  
\begin{proposition}
\label{prop:star-net-invariance-ext}
If, for all \( j\in [n] \), \( X_{j} \subset \real \) is the union of disjoint and closed intervals (including rays) and \( \map{\starobjin_{j}}{ X_{j} }{\real} \) is a piecewise \( \roofun \) function, then \( \starop \{\starobjin_{j}, X_{j} \}_{j=1}^{n} \) is a piecewise \( \roofun \) function. 
\end{proposition}
%\begin{proof}
%The proof follows straightforwardly from \eqref{eq:domain-decomposition}, Lemma~\ref{lem:meta-star-net}, Remark~\ref{rem:linearity-invariance} and the fact the point-wise maximum of multiple \( \roofun \) functions is a piecewise \( \roofun \) function. 
%\end{proof}

The number of subproblems on the right hand side of \eqref{eq:domain-decomposition} can be large. The next result provides conditions under which the number of subproblems required to be considered in \eqref{eq:domain-decomposition} can be significantly decreased. This is illustrated in Example~\ref{prop:star-net-invariance}.
%The next result further gives the necessary and sufficient conditions for the point-wise maximum of multiple \( \roofun \) functions of disconnected domains to be a \( \roofun \) function. The proof is provided in Appendix~\ref{proof:star-net-invariance-condition}. 
\begin{proposition}
\label{prop:star-net-invariance}
For all \( j\in [n] \), let \( X_{j} \subset \real \) be the union of  disjoint and closed intervals (including rays) and let \( \map{\starobjin_{j}}{ \conv{X_{j}} }{\real} \) be a \( \roofun \) function with top point \( \tau_{j} = (\tau_{j}^{1}, \tau_{j}^{2}) \). Then \( \starop \{ (\starobjin_{j}, X_{j}) \}_{j \in [n]} = \starop \{ (\starobjin_{j}, \conv{X_{j}} ) \}_{j \in [n]}  \) if and only if (1) \( \tau_{j}^{1} \in X_{j} \) for all \( j\in [n] \); and (2) both \( \sum_{j=1}^{n} X_{j}\cap (-\infty, \tau_{j}^{1}] \) and  \( \sum_{j=1}^{n} X_{j}\cap [\tau_{j}^{1}, \infty) \) are connected sets. 
\end{proposition}

\begin{example}
\label{ex:merge-subproblem}
Consider the following quantities.
Let \( X_{1}  =  [-4, -2]\cup [-1.5, 1.5] \cup [2, 4]  \), \( X_{2} = [-4, 4] \), and \( X_{3} = [-4, -2] \cup [-1, 1] \cup [2, 4] \). Let  \( \map{g_{i}}{X_{i}}{\real} \), \( i\in [3] \), be piecewise \( \roofun \) functions. \( g_{1} \) contains four pieces: \( \roofun_{(-3, 4)} \) over \( [-4, -2] \),  \( \roofun_{(-1, 2)} \) over \( [-1.5, 0] \),  \( \roofun_{(1, 2)} \) over \( (0, 1.5] \), and \( \roofun_{(3, 4)} \) over \( [4, 2] \). \( g_{2} \) contains two pieces \( \roofun_{(-3, 4)} \) over \( [-4, 0] \) and \( \roofun_{(3, 4)} \) over \( (0, 4] \). \( g_{3} \) has the same values as \( g_{1} \) over its domain. Let \( \starobjout_{i} = \starop(\{\starobjin_{i, j}, X_{i, j}\})_{j\in [2]} \), where  \( \starobjin_{i, j} = g_{i} \) and \( X_{i, j} = X \) for \( j\in [2] \) and \( i \in [3] \). 

Application of Proposition~\ref{prop:star-net-invariance} can be illustrated using the above quantities as follows. Computing
\( \starobjout_{1} \) requires solving 16 subproblems according to \eqref{eq:domain-decomposition}. However, following Proposition~\ref{prop:star-net-invariance}, it is sufficient to consider only 4 subproblems. This is because \(  \starobjout_{1} =  \starobjout_{2} \) and computing \( \starobj_{2} \) requires solving 4 subproblems. The proof for \(  \starobjout_{1} =  \starobjout_{2} \) is as follows: first, since \( g_{3} \le g_{1} \le g_{2} \),\footnote{We use the convention that the values at undefined points are $- \infty$.} it is straightforward to see from \eqref{opt:lp-meta-star} that \( \starobjout_{3} \le \starobjout_{1} \le \starobjout_{2} \); then, Proposition~\ref{prop:star-net-invariance} implies that \( \starobjout_{3} =  \starobjout_{2} \). 
\end{example}

\section{An equivalent state aggregation approach}
\label{sec:state-aggregation-main}
The computational approach proposed in Section~\ref{sec:tree-dp} has provable guarantees only for tree reducible networks. In this section, we return to the approach outlined in Section~\ref{sec:problem-statement}, and \eqref{eq:value-iteration} in particular, for a general network topology. We recall that, in its current form, \eqref{eq:value-iteration} is not amenable to implementations because the underlying state space is uncountably infinite. 
In this section, we develop an equivalent finite abstraction of the state space through \emph{state aggregation}, and correspondingly develop an aggregated version of \eqref{eq:value-iteration}.
%a state aggregation approach to equivalently transform the continuous state space of \eqref{eq:control-formulation} to a finite set. This equivalent finite representation then facilitates implementation of search algorithms to solve \eqref{eq:control-formulation} optimally. 

\subsection{A State Aggregation Approach}
\label{sec:state-aggregation}
The key idea is to develop a finite \emph{consistent partition} of one time step reachable sets. We recall a few standard terminologies. 
% in developing finite representation of the state space is a finite \emph{consistent partitioning} of the action set to get the set of aggregated control actions. In general, 
A \emph{cover} of a set \( S \) is a collection of nonempty subsets \( \{S_{i}\}_{i \in I} \) of \( S \) such that \( S = \cup_{i\in I} S_{i} \) and a \emph{partition} is a cover with pairwise disjoint elements. We call a cover or partition finite if it contains finitely many elements. 
%Equivalently, the following conditions hold for \( \{ U_{i}\}_{i \in I} \) to be a partition of \( U \): (i). \( \cup_{i\in I} U_{i}= U \); (ii). \( U_{i} \cap U_{j} = \emptyset\) for \( i \neq j \); (iii). \( U_{i} \neq \emptyset  \) for all \( i \in I\). 
Furthermore, for a network state \( (\mathcal{E}, p) \), a partition \( \{S_{i}\}_{i \in I} \) of set \( S \subseteq \mathcal{B}_{ \mathcal{E}} \) is said to be \emph{consistent} if, \( \mathcal{F}_{ \mathcal{E}}( \mathcal{E}, u) = \mathcal{F}_{ \mathcal{E}}( \mathcal{E}, \tilde{u}) \) for all \( u, \tilde{u} \in S_{i} \), \( i\in I \). Consistency implies that it is valid to write \( \mathcal{F}_{ \mathcal{E}}( \mathcal{E}, u) \equiv \mathcal{F}_{ \mathcal{E}}( \mathcal{E}, S_{i}) \) for all \( u\in S_{i} \) and \( i\in I \). Note here that the set \( S_i \) is not necessarily the set of admissible control actions \( U( \mathcal{E}, p) \), and can be any arbitrary set of balanced supply-demand vectors. We extend the notion of the set of admissible control actions \eqref{def:control-space} as follows: for a link set \( \mathcal{E} \) and a set \( P \subseteq \mathcal{B}_{ \mathcal{E}} \), 
\begin{equation}
\label{eq:feasible-load-set}
U( \mc E, P)  := \cup_{p \in P} \; U(\mathcal{E}, p) = \mathcal{B}_{ \mathcal{E}} \cap \cube{P}
\end{equation} 
where \( \cube{P}: = \cup_{p\in P} \cube{p} \). %Note \( U( \mathcal{E}, P) \) is a polytope as long as \( P \) is a polytope. 

A finite consistent partition of the set of control actions induces a natural finite cover of the reachable state space at each stage in the cascading dynamics. At $t=0$, the state space \( \{( \mathcal{E}^{0}, p^{0})\} \) is a singleton, and therefore \(  \{( \mathcal{E}^{0}, P^{0})\} \) forms a trivial partition with \( P^{0} := \{p^{0}\} \). Let \( \{U^{0}_{i}\}_{i \in I_{1} } \) be a finite consistent partition of \( U( \mathcal{E}^{0}, P^{0}) \). Then at \( t=1 \), the reachable state space \( \setdef{( \mathcal{F}_{ \mathcal{E}}( \mathcal{E}^{0}, u), u)}{u\in U( \mathcal{E}^{0}, p^{0})} \) is covered by \( \{ (\mathcal{E}_{i}^{1}, P_{i}^{1})\}_{i\in I_{1}} = \{(\mathcal{F}_{ \mathcal{E}}( \mathcal{E}^{0}, U^{0}_{i}),U^{0}_{i})\}_{i \in I_1}\). Let \( \{U^{1}_{j}\}_{j \in I_{2}^{i} } \) be a finite consistent partition of \( U( \mathcal{E}_{i}^{1}, P_{i}^{1}) \), for all \( i\in I_{1} \). Then at \( t=2 \),  the state space reachable from $(\mc E_i^1,P_i^1)$, \( \setdef{( \mathcal{F}_{ \mathcal{E}}( \mathcal{E}^{1}_{i}, u), u)}{u\in U( \mathcal{E}_{i}^{1}, P_{i}^{1})} \), is covered by  \( \{ (\mathcal{E}_{j}^{2}, P_{j}^{2})\}_{j\in I_{2}^{i} } = \{(\mathcal{F}_{ \mathcal{E}}( \mathcal{E}_{i}^{1}, U^{1}_{j}),U_{j}^{1})\}_{j \in I_2^i} \), for all \( i\in I_{1} \). Repeated application of this procedure to all the subsequent stages then gives the desired finite representation. We employ the elements \( ( \mathcal{E}^{t}, P^{t}) \) of the cover at time \( t \) as aggregated states. A natural extension of \eqref{eq:cascade-dynamics} to dynamics over the aggregated states is as follows: 
\begin{equation}
\label{eq:cascade-dynamics-aggregated-new}
\left(\mc E^{t+1}, P^{t+1} \right) = \mc F \left(\mathcal{E}^{t} , P^{t}, U^{t} \right), \qquad U^{t} \in  \mathbb{U}(\mathcal{E}^{t},P^{t})
\end{equation}
where  \( \mc F_p(\mc E, P, U) \equiv \mc F_p(U)  := U \) and \( \mc F_{\mc E}( \mc E, P, U) \equiv \mc F_{\mc E}(\mc E, U) := \setdef{i\in \mathcal{E}}{ |f_{i}( \mathcal{E}, u) | \le c_{i}, \, \forall \, u \in U }  \).
%where
%\begin{equation}
%\label{def:feasible-link-set-aggregated}
%\begin{split}
%\mc F_{\mc E}( \mc E, P, U) \equiv \mc F_{\mc E}(\mc E, U) & := \setdef{i\in \mathcal{E}}{ -c_{i} \le f_{i}( \mathcal{E}, u)  \le c_{i} \text{ for all } u \in U }  \\
%\mc F_P(\mc E, P, U) \equiv \mc F_P(U) & := U
%\end{split}
%\end{equation}
 \( \mathbb{U}(\mathcal{E}^{t},P^{t}) \) is defined by the particular choice of consistent partition of \( U( \mathcal{E}^{t}, P^{t}) \), and serves as the set of admissible aggregated control actions at state \( (\mathcal{E}^{t},P^{t} )\). We associate \(\mathcal{E}\) with a vector \( \beta \in \{1, 0, -1\}^{ \mathcal{E}} \). \( \beta_{i} \) is used to denote whether the flow capacity constraint of link \( i\in \mathcal{E} \) is satisfied or not: the flow stays within capacities for \( \beta_{i} =0 \) and exceeds the upper and lower capacities for \( \beta_{i} = 1 \) and \( \beta_{i} = -1 \), respectively. Let \( U( \mathcal{E}, P, \beta) := \{u \in  U( \mathcal{E}, P) \, |\,   f_{i}( \mathcal{E}, u) < -c_{i} \text{ for } \beta_{i} =-1;   |f_{i}( \mathcal{E}, u)| \le c_{i} \text{ for } \beta_{i} =0;  f_{i}( \mathcal{E}, u) > c_{i} \text{ for } \beta_{i} =1; \forall \, i \in \mathcal{E} \}  \)\footnote{See Remark~\ref{rem:aggregation-is-decomposition} for a variant definition of \( U( \mathcal{E}, P, \beta) \) .}. The consistent partition used in this paper is: 
 %For simplicity, the same notation \( \mathcal{F}( \cdot) \), \( \mathcal{F}_{ \mathcal{E}}(\cdot) \) and \( \mathcal{F}_{p}( \cdot) \) are used here as that in \eqref{eq:cascade-dynamics}. 
%We now describe the partition \( \mathcal{U}( \mathcal{E}, P) \) used in this paper. 
%
%%The vector \( \beta \in \real^{ \mathcal{E} } \) hence has \( 3^{ |\mathcal{E}|} \) number of different values, denoted as \( \{ \beta^{1}, \ldots, \beta^{3^{\mathcal{E}}} \} \). 
%For a given state \( ( \mathcal{E}, P) \) and \( \beta\), let 
% \( U( \mathcal{E}, P, \beta) := \{u \in  U( \mathcal{E}, P) \, |\,   f_{i}( \mathcal{E}, u) < -c_{i} \text{ for } \beta_{i} =-1;  -c_i \le f_{i}( \mathcal{E}, u) \le c_{i} \text{ for } \beta_{i} =0;  f_{i}( \mathcal{E}, u) > c_{i} \text{ for } \beta_{i} =1; \forall \, i \in \mathcal{E} \}  \)
\begin{equation}
\label{eq:partition-def}
 \mathbb{U}( \mathcal{E}, P): = \setdef{U( \mathcal{E}, P, \beta )}{ \beta \in \{1, 0, -1\}^{ \mathcal{E}} }\setminus \emptyset 
\end{equation}
%where \( I( \mathcal{E}, P) := \setdef{j}{ U( \mathcal{E}, P, \beta^{j}) \neq \emptyset} \). 
%\begin{equation}
%\label{eq:partition-full}
%\begin{split}
%I( \mathcal{E}, P) := \setdef{j}{ U( \mathcal{E}, P, \beta^{j}) \neq \emptyset} & \\
%U( \mathcal{E}, P, \beta) := \{u \in  U( \mathcal{E}, P) \, |\, &  f_{i}( \mathcal{E}, u) < -c_{i} \text{ for } \beta_{i} =-1;  -c_i \le f_{i}( \mathcal{E}, u) \le c_{i} \text{ for } \beta_{i} =0; \\
%& f_{i}( \mathcal{E}, u) > c_{i} \text{ for } \beta_{i} =1; \forall \, i \in \mathcal{E} \} \\
%\end{split}
%\end{equation}
If \( U( \mathcal{E}, P) \) is a polytope, then each member of \( \mathbb{U}( \mathcal{E}, P) \) is also a polytope, with possibly half open boundary due to the strict inequalities in \( U( \mathcal{E}, P, \beta) \). 
An algorithmic procedure to compute and store the partition \( \mathbb{U} \)  defined in \eqref{eq:partition-def} will be presented in Section~\ref{sec:algorithm-section}.

State aggregation gives a finite tree (cf. Section~\ref{sec:problem-statement}), whose nodes are aggregated states and arcs are aggregated control actions. The set of feasible aggregated states is \( \mathbb{S} :=\{(\mathcal{E}, P) \, |\, ( \mathcal{E}, p) \in \mathcal{S}, \, \forall\, p \in P\} \) and the reward associated with a state \( ( \mathcal{E}, P) \in \mathcal{S} \) is \( \sup_{p\in P} \trans{\sdsign}  p =  \max_{p\in \cl{P}} \trans{\sdsign}  p \), where \( \cl{P} \) denotes the closure of \( P \). 
%The notion of a feasible state in \eqref{eq:S-def} is extended as:
%\begin{equation}
%\label{eq:S-def-aggregated}
%\mathbb{S} :=\setdef{(\mathcal{E}, P)}{ ( \mathcal{E}, p) \in \mathcal{S}, \, \forall\, p \in P} 
% \end{equation}
% Similar to \eqref{eq:reward-function-def}, we assign value to aggregated states as follows to incorporate additional constraints on $U^{N-1}$:
%\begin{equation}
%\label{eq:reward-function-def-aggregated}
%r( \mathcal{E}, P)= \left\{ 
%\begin{array}{ll} 
% \max_{p\in \cl{P}} \trans{\sdsign}  p  \; & ( \mathcal{E}, P) \in \mathbb{S} \\
% -\infty & ( \mathcal{E}, P) \notin \mathbb{S}
%\end{array} \right. 
%\end{equation}  
% Each state \( (\mathcal{E}, P) \) is assigned the value: \( r( \mathcal{E}, P) : =  \max_{p\in \cl{P}} \trans{\sdsign}  p \) for \(  ( \mathcal{E}, P) \in \hat{\mc S}\) and \( r( \mathcal{E}, P) : = -\infty \) else,
%\begin{equation}
%\label{eq:agg-reward}
%r( \mathcal{E}, P) := \left\{ 
%\begin{array}{ll} 
%\max_{p\in \cl{P}} \trans{\sdsign}  p  \; & ( \mathcal{E}, P) \in \hat{\mc S} \\
% -\infty & ( \mathcal{E}, P) \notin \hat{\mc S}
%\end{array} \right. 
%\end{equation}
%where \( \cl{P} \) denotes the closure of \( P \). 
The optimal search over the aggregated tree can be performed through the following calculations, which are adaptations of \eqref{eq:value-iteration}: %\( P \) may not be closed, but since it is bounded and the function \(  \trans{\sdsign}  p \) is continuous, we have \( \max_{p\in \cl{P}} \trans{\sdsign}  p = \sup_{p\in P} \trans{\sdsign}  p \). 
\begin{subequations}
\label{eq:value-iteration-aggregated}
\begin{align}
%\label{eq:value-iteration-a-aggregated}
%\mathbb{J}_{0}( \mathcal{E}, P) & = r( \mathcal{E}, P) \\
\label{opt:lp-redispatch-new-aggregated} 
\mathbb{J}_{1}( \mathcal{E}, P ) & =  \max_{u \in U(\mc E,\cl{P})} \trans{\sdsign}u \,\, \text{ s.t. } \, |f( \mathcal{E}, u)| \leq c_{ \mathcal{E}} \\
%\begin{aligned}
%J_{1}( \mathcal{E}, p ) =  \max_{u \in U(\mc E,p)} \quad & \trans{\sdsign} u \\
%\text{s.t.} \quad & -c_{ \mathcal{E}} \leq f( \mathcal{E}, u) \leq c_{ \mathcal{E}}  
%%\\
%%&  u \in U(\mathcal{E}, p)
%\end{aligned} \\
 \label{eq:value-iteration-b-aggregated} 
\mathbb{J}_{t+1}(\mathcal{E}, P) & = \max_{U \in \mathbb{U}(\mathcal{E}, P)} \mathbb{J}_{t} \left(\mathcal{F}_{ \mathcal{E}}(\mathcal{E}, U), U \right)
 \end{align}
\end{subequations} 
where \( t\in [N-1] \) and \( \mathbb{J}_{t}( \mathcal{E}, P) \) is the maximum among values of all feasible aggregated states that can be reached in at most $t$ time steps starting from $(\mc E,P)$. Similar to \eqref{opt:lp-redispatch}, the flow constraint is imposed in \eqref{opt:lp-redispatch-new-aggregated} to ensure the additional constraint on \( u^{N-1} \) (cf. Remark~\ref{rem:additional-check}). This implies that the unique (and optimal) aggregated control action associated with \( \mathbb{J}_{1}( \mathcal{E}, P) \) is \( U^{N-1, *}=U( \mathcal{E}, P, \zerobf) \). This is to be contrasted with  \eqref{eq:value-iteration-b-aggregated} for \( t\ge 2 \) that the optimal aggregated control action \( U^{t, *} \) is not trivial to obtain.
\eqref{opt:lp-redispatch-new-aggregated} maximizes a linear function over a bounded closed set and \eqref{eq:value-iteration-b-aggregated} maximizes over the finite set \( \mathbb{U}( \mathcal{E}, P) \). Therefore, the optimal value is achievable in every iteration in \eqref{eq:value-iteration-aggregated}.
%Likewise, the problem is to find an aggregated action sequence \( (U^{0}, \ldots, U^{N-1}) \) that starts from the initial aggregated state \( (\mathcal{E}^{0}, \{p^{0}\})  \) and gives the value \( J_{N}( \mathcal{E}^{0}, \{p^{0}\} ) \). 
The next result shows that the iterations in \eqref{eq:value-iteration-aggregated} give the same value as that in \eqref{eq:value-iteration}. 
\begin{theorem}
\label{thm:equivalence}
Consider a network with initial state $(\mc E^0,p^0)$, link weights $w \in \preal^{\mc E^0}$ and link capacities $c \in \preal^{\mc E^0}$. For every aggregated state $(\mc E,P)$ obtained from the consistent partition in \eqref{eq:partition-def}, the computations in \eqref{eq:value-iteration} and \eqref{eq:value-iteration-aggregated} satisfy the following for \( t\in [N] \):
$$
\mathbb{J}_{t}( \mathcal{E}, P) = \sup_{ p\in P} J_{t}( \mathcal{E}, p)  
$$
%
%initial active link set \( \mathcal{E}^{0} \) and initial supply-demand vector \( p^{0} \), let \( J_{t}( \mathcal{E}, p) \) and \( J_{t}( \mathcal{E}, P) \), for \( t\in [N] \), be the value functions determined by \eqref{eq:value-iteration} and \eqref{eq:value-iteration-aggregated}, respectively. Then \( \mathbb{J}_{t}( \mathcal{E}, P) = \sup_{ p\in P} J_{t}( \mathcal{E}, p)  \) for all \( t\in [N] \).  
\end{theorem}
\begin{proof}
We prove by induction. For \( t = 1 \): \(  \displaystyle{\sup_{p\in P}} J_{1}( \mathcal{E}, p)  = \displaystyle{ \sup_{p\in P} \max_{u \in U(\mc E, p)}} \, \{  \trans{\sdsign} u \, |\,   |f( \mathcal{E}, u)| \le c_{\mathcal{E}} \} =  \mathbb{J}_{1}( \mathcal{E}, P) \).
%\begin{equation*}
%\begin{array}{rrl} 
% \sup_{p\in P} J_{1}( \mathcal{E}, p) & = \displaystyle{ \sup_{p\in P} \max_{u \in U(\mc E, p)}} \, \,   \trans{\sdsign} u \quad \text{s.t. } -c_{ \mathcal{E}} \leq f( \mathcal{E}, u) \leq c_{ \mathcal{E}} 
%% \\
%% & = \displaystyle{\max_{u \in U(\mc E,p)}}  \sup_{p\in P}  \, \,   \trans{\sdsign} u \quad \text{s.t. } -c_{ \mathcal{E}} \leq f( \mathcal{E}, u) \leq c_{ \mathcal{E}}
%\end{array} 
%\end{equation*}
%which is equal to \( \mathbb{J}_{1}( \mathcal{E}, P) \) as shown in \eqref{opt:lp-redispatch-new-aggregated}. 
Suppose the claim is true for \(  t \in [k] \). 
\begin{align*}
& \mathbb{J}_{k+1}(\mathcal{E}, P) = \max_{U \in \mathbb{U}( \mathcal{E}, P)} \mathbb{J}_{k}( \mathcal{F}_{ \mathcal{E}}( \mathcal{E}, U), U)  \\
&=  \max_{U \in \mathbb{U}( \mathcal{E}, P)} \sup_{u \in U} J_{k}( \mathcal{F}_{ \mathcal{E}}( \mathcal{E}, u), u)  = \sup_{u \in  U( \mathcal{E}, P)} J_{k}( \mathcal{F}_{ \mathcal{E}}( \mathcal{E}, u), u)  \\
 & = \sup_{p \in P} \sup_{u \in U( \mathcal{E}, p)} J_{k}( \mathcal{F}_{ \mathcal{E}}( \mathcal{E}, u), u)  =  \sup_{p \in P} J_{k+1}( \mathcal{E}, p)
\end{align*} 
where the first equality is due to  \eqref{eq:value-iteration-b-aggregated}; the second equality is due to the induction assumption; the third and forth equalities are due to \( \cup_{U \in \mathbb{U}( \mathcal{E}, P) } U = U( \mathcal{E}, P) = \cup_{ p\in P} U( \mathcal{E}, p)  \), as implied by the definitions of \( \mathbb{U}( \mathcal{E}, P) \) and \( U( \mathcal{E}, P) \); and the last equality is due to \eqref{eq:value-iteration}. 
\end{proof} 

We make the following remarks on the state aggregation and aggregated tree search. 
\begin{remark}
\label{rem:aggregation-is-decomposition}
\leavevmode
\begin{enumerate}
\item Every path of length \( N \) in the aggregated search tree corresponds to a, possibly different, topology sequence that can occur in the cascading dynamics. The state aggregation proposed here is the minimal among all finite abstractions for exact computation of optimal control. 
%Details are provided in Section~\ref{sec:problem-complexity}. 
%For a general posed optimal control of cascading failure problem, the state aggregation based on \eqref{eq:cascade-dynamics-aggregated-new} and \eqref{eq:partition-full} gives the minimal finite presentation of networks states, that enables  exact solutions. This is because the number of possible topology sequence in the cascading dynamics can be, at most, equal to the number of paths in the aggregated search tree constructed from this presentation. Due to the non-monotonicity of the power flow over a general network \cite{Ba.Savla:TCNS16, Ba.Savla.CDC16}, all possible topology sequences need to be considered in order to compute an exact optimal load shedding control.
\item It is straightforward to see from \eqref{eq:value-iteration-aggregated} that the value of \( \mathbb{J}_{t}( \mathcal{E}, P) \) depends only on \( \cl{P} \) for all \( t \in [N]\). Therefore, without introducing error in the computation of \( \mathbb{J}_{t}( \mathcal{E}, P) \), we use \( ( \mathcal{E}, \cl{P}) \) and \( \cl{U} \) in place of \( ( \mathcal{E}, P) \) and \( U \). Correspondingly, without stating explicitly, hereafter we use the following variant of the definition in \eqref{eq:partition-def}: \( U( \mathcal{E}, P, \beta):= \{u \in  U( \mathcal{E}, P) \, |\,   f_{i}( \mathcal{E}, u) \le  -c_{i} \text{ if } \beta_{i} =-1;  |f_{i}( \mathcal{E}, u)| \le c_{i} \text{ if } \beta_{i} =0;  f_{i}( \mathcal{E}, u) \ge c_{i} \text{ if } \beta_{i} =1; \forall \, i \in \mathcal{E} \}  \). 
\end{enumerate}
\end{remark}

\subsection{Optimal Control Synthesis: From Aggregated to the Original State Space}
\label{sec:retrieve-optimal-control} 
The numerical implementation of \eqref{eq:value-iteration-aggregated} is shown in Section \ref{sec:algorithm-section}. Given such a procedure to compute \( U^{*}= (U^{0, *}, \ldots, U^{N-1, *}) \), we next present a result to derive \( u^{*} = (u^{0, *}, \ldots, u^{N-1, *}) \), i.e., control actions for the cascading failure dynamics in \eqref{eq:cascade-dynamics} over the unaggregated state space. However, since the set of feasible control action sequences $\mc D$ is not necessarily closed (cf. Remark~\ref{rem:additional-check}), finding \( u^{*} \) whose cost is \emph{exactly} the same as that of \( U^{*} \) may not be possible. It is however possible to find \( u^{*} \) whose cost is arbitrarily close to that of \( U^{*} \). 

%\ksmargin{check if the superscripts $N$ and $N-1$}
\begin{proposition}
\label{prop:close-solution}
For a network with initial state $(\mc E^0,p^0)$, link weights $w \in \preal^{\mc E^0}$, and link capacities $c \in \preal^{\mc E^0}$, consider $\mathbb{J}_N(\mc E^0,\{p^0\})$ computed by \eqref{eq:value-iteration-aggregated}. For every \( \epsilon >0 \), there exists \( \tilde{u} \in \mathcal{D} \) such that \(\mathbb{J}_{N}( \mathcal{E}^{0}, \{p^{0}\}) \ge \trans{\sdsign} \tilde{u}^{N-1} \ge  \mathbb{J}_{N}( \mathcal{E}^{0}, \{p^{0}\}) - \epsilon\). 
\end{proposition}
\begin{proof}
For brevity, in this proof, we let $\mathbb{J}_N(\mc E^0,\{p^0\}) \equiv \mathbb{J}_N(\mc E^0,p^0)$.
Theorem~\ref{thm:equivalence} implies that \( \mathbb{J}_{N}( \mathcal{E}^{0}, p^{0}) = J_N(\mc E^0,p^0) \ge \trans{\sdsign} u^{N-1} \) for all \( u \in \mathcal{D} \). Therefore, we only prove the second inequality. 

Let  \( U^{*} \) be an optimal aggregated control sequence associated with computing $\mathbb{J}_N(\mc E^0,p^0)$ in \eqref{eq:value-iteration-aggregated}, and let \( \mathcal{E}^{*} \) be the induced active link set sequence (recall \( U^{N-1, *}= U( \mathcal{E}^{N-1, *}, U^{N-2, *}, \zerobf) \)). Let \( u^{N-1, *} \) be a maximizer to \eqref{opt:lp-redispatch-new-aggregated} for \( \mathbb{J}_{1}( \mathcal{E}^{N-1,*}, U^{N-2,*}) \). Then, \( u^{N-1, *} \in \cl{U^{N-1, *}} \) and \( \mathbb{J}_{N}( \mathcal{E}^{0}, p^{0}) = \trans{\sdsign} u^{N-1, *} \). We now show that, for arbitrary \( \epsilon >0 \), there exists  \( \tilde{u} \in \mathcal{D} \) such that \( \trans{\sdsign} \tilde{u}^{N-1} \ge  \trans{\sdsign} u^{N-1, *} - \epsilon\). 

Let \( M(u, \epsilon) \) be the open ball centered at \( u \in \real^{ \mathcal{V}} \) with radius \( \epsilon \).  Since \( u^{N-1, *} \in \cl{U^{N-1, *}} \), \( U^{N-1, *} \cap M(u^{N-1, *}, \epsilon/| \mathcal{V}_{l}|) \neq \emptyset \) for every \( \epsilon >0 \). It is then possible to pick \( \tilde{u}^{N-1} \in U^{N-1, *} \cap M(u^{N-1, *}, \epsilon/| \mathcal{V}_{l}|) \) such that \( \tilde{u}^{N-1} \neq u^{N-1, *} \) and \(  \trans{\sdsign} \tilde{u}^{N-1} > \trans{\sdsign} u^{N-1, *} - \epsilon  \). It is now sufficient to show that there exist \( \tilde{u}^{0}, \tilde{u}^{1}, \ldots, \tilde{u}^{N-2} \) such that \( \tilde{u}^{t+1} \in \cube{\tilde{u}^{t}} \) and \( \tilde{u}^{t} \in U^{t, *} \) for all \( 0 \le t \le N-2 \). We provide details for \( \tilde{u}^{N-2} \); the reasoning for \( \tilde{u}^{N-3}, \ldots, \tilde{u}^{0} \) follows along the same lines. 

Since \( u^{N-1, *} \in U( \mathcal{E}^{N-1, *}, \cl{U^{N-2,*}}) \), there exists \( u^{N-2, *} \in \cl{U^{N-2,*}} \) such that \( u^{N-1, *} \in \cube{u^{N-2, *}} \). Hence, we can pick \( \tilde{u}^{N-2}  \in U^{N-2,*} \cap M(u^{N-2, *}, \| u^{N-1, *} - \tilde{u}^{N-1} \|_{2})\) so that \(  \tilde{u}^{N-2}  \neq u^{N-2, *} \) and \( \tilde{u}^{N-1} \in \cube{\tilde{u}^{N-2} }\), where the special choice of \(  \tilde{u}^{N-1} \neq  u^{N-1, *}  \) ensures that \( M(u^{N-2, *}, \| u^{N-1, *} - \tilde{u}^{N-1} \|_{2}) \) has positive radius. 
\end{proof}

The proof of Proposition~\ref{prop:close-solution} implies that, in order to find \( u^{*} \), it remains to solve for \( u^{N-1, *} \) from \eqref{opt:lp-redispatch-new-aggregated}. 
%Since the value of \( \mathbb{J}_{1}( \mathcal{E}, P)\) depends only on \( \cl{P} \) in \eqref{opt:lp-redispatch-new-aggregated}, the value of \( J_{t}( \mathcal{E}, P) \) depends only on \( \cl{U} \) and \( \cl{P} \) in \eqref{eq:value-iteration-b-aggregated} for every \( t \ge 2\). Therefore, in order to find  \( U^{*} \), for sake of numerical implementation, we use \( ( \mathcal{E}, \cl{P}) \) and \( \cl{U} \) in place of \( ( \mathcal{E}, P) \) and \( U \), without introducing error in computation of \( U^{*} \). Equivalently, we use the following variant definition \( U( \mathcal{E}, P, \beta):= \)
%% \( U( \mathcal{E}, P, z) := \{u \in  U( \mathcal{E}, P) \, |\,  f_{i}( \mathcal{E}, u) \le  -c_{i} \text{ if } \beta_{i} =-1;  -c_i \le f_{i}( \mathcal{E}, u) \le c_{i} \text{ if } \beta_{i} =0;  f_{i}( \mathcal{E}, u) \ge c_{i} \text{ if } \beta_{i} =1; \forall \, i \in \mathcal{E} \}  \).
%\begin{align}
%& \{u \in  U( \mathcal{E}, P) \, |\,   f_{i}( \mathcal{E}, u) \le  -c_{i} \text{ if } \beta_{i} =-1;  |f_{i}( \mathcal{E}, u)| \le c_{i} \nonumber \\
%& \text{ if } \beta_{i} =0;  f_{i}( \mathcal{E}, u) \ge c_{i} \text{ if } \beta_{i} =1; \forall \, i \in \mathcal{E} \} \label{eq:partition-variant}
%\end{align} 
%Therefore, without stating explicitly, hereafter we use \eqref{eq:partition-variant}. 
The next result implies that \eqref{opt:lp-redispatch-new-aggregated} is a linear program by showing that, for every \( (\mathcal{E}, P) \), the set \( U( \mathcal{E}, P) \) is a polytope. 
\begin{lemma}
\label{lem:polyhedron-partition}
Consider a network with initial state $(\mc E^0,p^0)$, link weights $w \in \preal^{\mc E^0}$ and link capacities $c \in \preal^{\mc E^0}$. For every aggregated state $(\mc E^t,P^t)$, $t \in [N] \cup \{0\}$, induced by the consistent partition in \eqref{eq:partition-def}, both \( P^t \) and \( U( \mathcal{E}^t, P^t) \) are polytopes. 
\end{lemma}
\begin{proof}
The claim is proved by induction. It is easy to see that \( P^{0}= \{p^{0}\} \) and \( U( \mathcal{E}^{0}, P^{0} ) = U( \mathcal{E}^{0}, p^{0}) \) are polytopes. Suppose  \( P^{t} \) and \( U( \mathcal{E}^{t}, P^{t}) \) are polytopes for some $t \in [N-1] \cup \{0\}$. It is sufficient to show that, for arbitrary aggregated control action \( U\in \mathbb{U}( \mathcal{E}^{t}, P^{t}) \), the resulting \( P^{t+1} \) and \( U( \mathcal{E}^{t+1}, P^{t+1}) \) are polytopes. \( P^{t+1} = U^t = U( \mathcal{E}^{t}, P^{t}, \beta) \) for some \( \beta \in \{-1,0,+1\}^{\mc E^t}\) (cf. \eqref{eq:partition-def}. Combining this with the induction assumption that \( U( \mathcal{E}^{t}, P^{t})\) is a polytope, we get that \( P^{t+1} \) is a polytope. It follows from the definition in \eqref{eq:feasible-load-set} that \( U( \mathcal{E}^{t+1}, P^{t+1}) \) is a polytope as well. 
\end{proof}

\begin{remark}
\label{rem:in-orthant}
\begin{enumerate}
\item Combined with the definition of \( \cube{P} \), the proof of Lemma~\ref{lem:polyhedron-partition} also implies that, for every \( ( \mathcal{E}, P) \), both \( P \) and \( U( \mathcal{E}, P) \) are contained in a closed orthant of \( \real^{ \mathcal{V}_{l} } \). 
\item The aggregated tree search based on \eqref{eq:value-iteration-aggregated} can be interpreted as a systematic way to decompose the nonconvex feasible set \( \mathcal{D}  \subset \real^{N \times | \mathcal{V}| } \) in \eqref{eq:control-formulation-new}-\eqref{eq:feasible-set-def} into a finite number of subsets. Each subset is a polytope and corresponds to a topology sequence and aggregated control sequence.  For example, the subset corresponding to \( (U^{0}, \ldots, U^{N-1}) \)  is $\Pi_{t=0}^{N-1} U^{t}$. Similar to \eqref{eq:domain-decomposition}, the optimal value of \eqref{eq:control-formulation-new} is equal to the maximum among the optimal values of multiple subproblems associated with the subsets. Each subproblem is a linear program of the form \( \sup_{u \in \Pi_{t=0}^{N-1} U^{t} } \trans{\sdsign} u^{N-1} \) and indeed coincides with \eqref{opt:lp-redispatch-new-aggregated}. 
\item The state aggregation approach allows to include running cost into the problem formulation. In that case, instead of a linear program, a dynamic program is to be solved for every \( (U^{0}, \ldots, U^{N-1}) \). 
\end{enumerate}
\end{remark} 

\subsection{Efficient Aggregated Tree Search} 
\label{sec:tree-search} 
%\eqref{eq:cascade-dynamics-aggregated} and \eqref{eq:aggre-value-iteration} describe the process of constructing the aggregated tree to find a feasible state (which is a leaf node) that has the maximal value. Each path corresponds to a possible topology sequence in cascading dynamics. 
With \eqref{eq:cascade-dynamics-aggregated-new} and \eqref{eq:partition-def} specifying how to expand nodes of aggregated network state and \eqref{eq:value-iteration-aggregated} directing the goal of search, one can then employ any classical tree search algorithm, e.g., the ones in \cite[Chap 3]{russell2009artificial}, to solve the problem. Regarding the implementation of tree search algorithms, a few remarks are in order. First,  the following relationship can be used for tree pruning in a standard branch and bound algorithm framework. 
%\eqref{eq:value-iteration-aggregated} provides a systematic way to explore all possible paths from $(\mc E^0,p^0)$ to feasible aggregated leaf nodes. \kscomment{While any classical tree search algorithm}, e.g., the ones in \cite[Chap 3]{russell2009artificial}, can be used to solve the problem, we wish to make the most of the following bounds for tree pruning: \ksmargin{these bounds have issues as discussed}
\begin{equation*}
\mathbb{J}_{1}( \mathcal{E}, P) \le \mathbb{J}_{t}( \mathcal{E}, P) \le \max_{p\in \cl{P}} \trans{\sdsign}  p  \qquad \forall \, ( \mathcal{E}, P), \forall\, t\in [N]
\end{equation*}

Second, iterative deepening depth-first search algorithm presents several advantages for the optimal control problem. On one hand, it achieves a good balance between time and space complexities. This is of particular importance because
%, as would shown in Section~\ref{sec:problem-complexity}, 
the number of aggregated states in the optimal control problem can be quite large. On the other hand,  the search can be stopped anytime in the process of computation while producing a feasible control action with reasonable performance. In fact, the search over the first \( t <N \) layer provides an optimal \( t \)-stage load shedding scheme. At the same time, the upper bound provides an estimate of performance gap from the optimal value when search is terminated early. 

Third, following \eqref{eq:flow-susceptance-relationship} and \eqref{eq:partition-def}, going from one node to another in the aggregated search tree involves computation of pseudo-inverse of Laplacian associated with the new active link set. Doing such a computation from scratch for every search move could be computationally expensive, particularly for large $\mc E^0$ or for large $N$. This problem can be addressed by incrementally updating pseudo-inverse of the Laplacian under link removal, e.g., see \cite{Ba.Savla.CDC16, Soltan.Loh.ea:CONES17, ba2018elements}. 

Finally, while the detailed search algorithm including pruning is standard and omitted here, its implementation with set objects, i.e., the aggregated control action \( U \in \mathbb{U}( \mathcal{E}, P) \) and the set \( U( \mathcal{E}, P) \), require additional tools. Section~\ref{sec:algorithm-section} is devoted to this particular problem.

%\begin{remark}
%\label{rem:aggregation-complexity}
%\begin{enumerate}
%\item \sout{The number of subproblems equals to the number of possible \( N \) stage topology sequences from the initial network state under cascading failure dynamics and also equals to the number of feasible state in the aggregated tree. Due to the nonmonotonicity of DC networks, in general, no ordering exists between any two topology sequences in terms of performance. Therefore, in order to obtain an optimal solution, one has to search through all the possible topology sequences. This is same as saying one has to scan each number at least once in order to find the maximum over an array of unsorted numbers. In this sense, the number of feasible state in the aggregated search tree quantifies the size or complexity of \eqref{eq:control-formulation}.} 
%\end{enumerate}
%\end{remark}

%\section{Implementation using Incidence Graph}
\section{Computing Aggregation Through Arrangement of Hyperplanes}
\label{sec:algorithm-section}
The numerical implementation of \eqref{eq:value-iteration-aggregated} relies critically on proper representation of  set $U(\mc E,P)$ and its partition $\mathbb{U}(\mc E,P)$. While Lemma~\ref{lem:polyhedron-partition} characterizes an important property of these objects, in this section, we provide an algorithmic procedure for their representation. Our machinery relies on and extends tools from the domain of \emph{arrangement of hyperplanes} e.g., see \cite{edelsbrunner1987algorithms} \cite[Chapter 24]{toth2004handbook}, and \emph{polytopes}, e.g., see \cite{ziegler2012lectures} \cite{grunbaum1967convex}.
%In order to proceed with numerical implementation of the aggregated tree search algorithm, proper numerical representation is required for both \( U \) and \( \mathcal{U}( \mathcal{E}, P) \). As Lemma~\ref{lem:polyhedron-partition} shows, \( U \) is polytope and \(  \mathcal{U}( \mathcal{E}, P) \) is an partition of polytope \( U( \mathcal{E}, P) \) given by \eqref{eq:partition-variant}. \sout{We hence use \emph{incidence graph} to represent them.}
%\footnote{The representation of polytopes used in the paper is usually referred as \emph{face lattice}, rather than \emph{incidence graph}. The latter is used to represent arrangement of hyperplanes. Since they are equivalent when treating as data structures, we refer them uniformly as incidence graph. }. 

%\ksmargin{include a brief description of 'arrangement of hyperplanes'}
\subsection{Arrangement, Polytope and Incidence Graph}
\label{sec:arrangement}
%\emph{Incidence graph} is a commonly used data structure to represent an arrangement of hyperplanes. 

\begin{figure}[htbp!]
\captionsetup{width=0.9\linewidth}
\centering
\begin{subfigure}[b]{.2\linewidth}
  \centering
  \begin{tikzpicture}%[node distance=0.6cm and 2cm]
 \node[supply node] (1) {1};
 \node[main node] (3) [below right =of 1] {3};
  \node[supply node] (2) [above right =of 3] {2};
 \node[demand node] (4) [below =of 3] {4} ;

 \path[main edge, edge label]
   (1) edge node{\( e_{1} \)} (2)
   (1) edge node[below left] {\( e_{2} \)} (3)
   (2) edge node{\( e_{3} \)} (3)
   (3) edge node{\( e_{4} \)} (4);
   
\end{tikzpicture}
  \vspace{1em}
  \caption{ }
  \label{fig:net-1s2d}
\end{subfigure}
 \begin{subfigure}[b]{.3\linewidth}
  \centering
  \begin{tikzpicture} 
% \draw[fill, black!10!white] (0, 0)-- (0, 5) --(5, 5) -- (5, 0) ; 
\draw [<->, thick] (0,6) node[above left] {{\huge \( u_{2} \)}} --(0,0) node[below left] at (0, 0) { {\huge \( o \)} } -- (6,0) node[below right]{ {\huge \( u_{1} \)} }; 

\node at (5, 0) [below right]{{\huge \(d \) }};
\node at (5, 5) [above right]{{\huge\( e \) }};
\node at  (0, 5) [above left]{{\huge \( r \) }};

\draw[dashed, thick] (5, 5.5) --(5, -0.5);
\draw[dashed, thick] (-0.5, 5) --(5.5, 5);
\draw[fill=gray] (3.43, 3.43) node[above right]{{\huge\( a \)} } -- (3.65, 2.35) node[above right]{{\huge \( b \) }} -- (2.35, 3.65) node[below left]{ {\huge \( c \) }} -- (3.43, 3.43);
\draw (5.7, 0.3) -- (0.3, 5.7) node[above right] {\huge \( f_{4} \)}; 
\draw (4.25, -0.5) -- (3, 5.5) node[above right]{\huge \( f_{2} \)}; 
\draw (-0.5, 4.25) -- (5.5, 3) node[above right]{\huge \( f_{3} \)}; 
\draw (3.2, 3.2) -- (2.5, 2.5) node[below]{{\huge\( U_{1} \)}};

\end{tikzpicture}
  \caption{ }
  \label{fig:arrangement-1}
\end{subfigure}
\begin{subfigure}[b]{.2\linewidth}
  \centering
  \begin{tikzpicture}
\node[face] (1)  at (0, 0) {};

\node[face] (od) at (-0.6, -0.5) {};
\node[face] (de)  at (-0.2, -0.5) {};
\node[face] (ef) at (0.2, -0.5) {};
\node[face] (of) at (0.6, -0.5) {};

\node[face] (o) at (-0.6, -1) {};
\node[face] (d)  at (-0.2, -1) {};
\node[face] (e) at (0.2, -1) {};
\node[face] (f) at (0.6, -1) {};
\node at (-0.6, -1.3) {{\scriptsize \( o \) }};
\node at (-0.2, -1.3) {{\scriptsize \( d \) }};
\node at (0.2, -1.3) {{\scriptsize \( e \) }};
\node at (0.6, -1.3) {{\scriptsize \( r \) }};

\draw (1) -- (od); \draw (1) -- (de); \draw (1) -- (ef); \draw (1) -- (of); 
\draw (od) -- (o); \draw (od) -- (d); 
\draw (de) -- (d); \draw (de) -- (e); 
\draw (ef) -- (e); \draw (ef) -- (f);
\draw (of) -- (o); \draw (of) -- (f);

\end{tikzpicture}
  \caption{ }
  \label{fig: incidence-graph-initial-construction}
\end{subfigure}
 \begin{subfigure}[b]{.2\linewidth}
  \centering
  \begin{tikzpicture}
\node[face] (1)  at (0, 0) {};
\node[face] (2) at (-0.5, -0.5) {};
\node[face] (3)  at (0.5, -0.5) {};
\node[face] (4) at (0, -0.5) {};

\node[face] (5) at (-0.5, -1) {};
\node[face] (6)  at (0.5, -1) {};
\node[face] (7) at (0, -1) {};
\node at (-0.5, -1.3) {{\scriptsize \( a \) }};
\node at (0, -1.3) {{\scriptsize \( b \) }};
\node at (0.5, -1.3) {{\scriptsize \( c \) }};

\draw (1) -- (2); \draw (1) -- (3); \draw (1) -- (4); 
\draw (2) -- (6); \draw (2) -- (5); 
\draw (3) -- (7); \draw (3) -- (6); 
\draw (4) -- (7); \draw (4) -- (5);

\end{tikzpicture}
  \caption{ }
  \label{fig:incidence-graph-control}
\end{subfigure}
\caption{\sf  (a) a network \( ( \mathcal{V}, \mathcal{E}^{0}) \) with \( \mathcal{V}_{+} = \{1, 2\} \), \( \mathcal{V}_{-}=\{4\} \), \( w= \onebf \), \( c = \trans{[10, 3, 3, 6]} \), \( p^{0} = \trans{[-5, -5,  0, 10]} \); (b) projections of \( U( \mathcal{E}^{0}, p^{0}) \) and \( \setdef{u\in \mathcal{B}_{ \mathcal{E}^{0}}}{f_{i}( \mathcal{E}^{0}, u) = c_{i}} \), \( i\in \{2, 3, 4\} \), on \( u_{1}-u_{2} \) plane; (c) incidence graph of \( U( \mathcal{E}^{0}, p^{0}) \); (d) incidence graph of \( U_{1}\). }
\label{fig:arrangement-1st-example}
\end{figure}
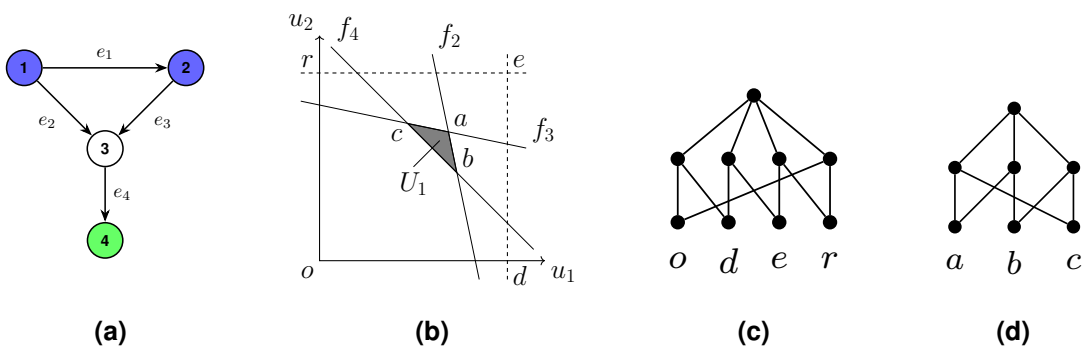

%As the later will be shown closed related to \( \mathcal{U}( \mathcal{E}, P) \), we use simple example of \( \mathcal{U}( \mathcal{E}, P) \) to illustrate the concepts relating arrangement of hyperplanes. 
%Readers new to this field can find more useful information in \cite{edelsbrunner1987algorithms} \cite[Chapter 24]{toth2004handbook}. 
We start with the simple illustrative example shown in Fig.~\ref{fig:net-1s2d}. Referring to the definition of \( \cube{p^{0}} \), there are three non-trivial components of $p^0$. Taking into account the additional constraint imposed by $\mc B_{\mc E}$, referring to \eqref{def:control-space}, the set of admissible control actions $U(\mc E^0,p^0)$ can be completely understood in terms of its two dimensional projection, say on the $u_1-u_2$ plane. In Fig.~\ref{fig:arrangement-1}, the box \( oder \) and the point $e$ correspond to the projections of $U(\mc E^0,p^0)$ and $p^0$, respectively, and the solid lines labeled by \( f_{2}, f_{3}, f_{4} \) correspond to the projections of capacity constraints associated with the links $e_2$, $e_3$ and $e_4$. The flow capacity constraint for \( e_{1} \) is ignored here because it is satisfied by all \( u\in U( \mathcal{E}^{0}, p^{0}) \) and hence irrelevant for the problem. Lines \( f_{2} \), \( f_{3} \) and \( f_{4} \) dissect the box \( oder \) into seven two dimensional pieces. Each piece, e.g., the triangle $abc$ denoted by $U_1$, is an aggregated control action, e.g., \( U_{1} = U( \mathcal{E}^{0}, p^{0}, \trans{[0, 0, 0, 1]}) \).  These seven pieces constitute the partition $\mathbb{U}(\mc E^0,\{p^0\})$. 
 
%
%Let the initial condition be $(\mc E^0,p^0)$. Fig.~\ref{fig:arrangement-1} illustrates the capacity constraints for its three links using solid lines, where the point \( o \) and the box \( oder \) are the projections of \( p^{0} \) and feasible control set \( U( \mathcal{E}^{0}, p^{0}) \), respectively, on the \( u_{1}-u_{2} \) plane. In this case, it is sufficient to focus on this two dimensional space. The lines in Fig.~\ref{fig:arrangement-1} dissect the action set \( oder \) into seven parts including \( acb \) as an example. By definition, \( \mathbb{U}( \mathcal{E}^{0}, \{p^{0}\}) \) is the indeed the collection of these parts. 

In general, a finite collection $ \mathcal{H} $ of hyperplanes in $ \real^{d} $ dissects $ \real^{d} $ into finitely many connected pieces of various dimensions. The collection of these pieces is called the \emph{arrangement}, denoted by $ \mathcal{A}( \mathcal{H} ) $,  induced by $ \mathcal{H}$ %{\cite{edelsbrunner1987algorithms} \cite[Chapter 24]{toth2004handbook}}
, and each piece is called a \emph{face}, denoted by \( \Gamma \), of the arrangement. The dimension of a face is the dimension of its affine hull; a \( k \) dimensional face is called a \( k \)-face, denoted by \( \Gamma^{k} \). For convenience, \( 0 \)-face, \( 1 \)-face,  \( (d-2) \)-face, \( (d-1) \)-face and \( d \)-face are, respectively, referred to as \emph{vertex}, \emph{edge}, \emph{ridge}, \emph{facet} and \emph{cell}\footnote{In some literature, \emph{cell} is used to refer to the connected pieces and \emph{face} is used exclusively for the \( 2 \)-face. In this paper, we adopt the terminology convention in \cite{edelsbrunner1987algorithms} and use \emph{cell} to denote \( d \)-face exclusively in \( \real^{d} \).}. We call two faces \emph{incident} if one is contained in the boundary of the other and if the difference in their dimensions is one. In a pair of incident faces, the lower (or higher) dimensional face is called the subface (or superface) of the other. In Fig.~\ref{fig:arrangement-1}, the three solid lines (i.e., hyperplanes) dissect box \( oder \) into seven cells, nine edges (or facets) and three vertices (or ridges). As indicated by this example, in the setting of this paper, for every state \( ( \mathcal{E}, P) \), the capacity constraints, balanced condition (captured by $\mc B_{\mc E}$), and load shedding requirement (captured by $\cube(P)$) form the collection of hyperplanes. We are interested in the substructure of the arrangement of these hyperplanes inside \( U( \mathcal{E}, P) \). The closure of each facet in the substructure corresponds to an aggregated control action and the collection of these facets corresponds to \( \mathbb{U}( \mathcal{E}, P) \)\footnote{Fig.~\ref{fig:arrangement-1} shows the projection of the arrangement in \( \real^{3} \) onto the \( u_{1}-u_{2} \) plane. Each cell in Fig.~\ref{fig:arrangement-1} is the projection of a facet of the arrangement.}. 

The closure of a bounded face in the arrangement is a polytope, or polytope for short. We use the same letter \( P \), as in the aggregated state, to denote a general polytope for simplicity, because every aggregated state is a polytope, as shown in Lemma~\ref{lem:polyhedron-partition}. Formally, a \emph{polytope} is a point set \( P \subset \real^{d} \) that can be presented either as a convex hull of a finite number of points in \( \real^{d} \) or the bounded intersection of finite number of closed half spaces in \( \real^{d} \) \cite{ziegler2012lectures}. The same notion of face, as in an arrangement, is used for a polytope \( P\subset \real^{d} \) \footnote{The notion of \emph{face} is slightly different for an arrangement and a polytope. While the former considers a face as an open set, the latter treats a face as a closed set. This difference does not affect the results in this section and hence ignored.} and furthermore, a face of \( P \) can be described as \( \Gamma = P \cap \setdef{x\in \real^{d}}{\trans{\pi} x = \pi_{0}} \), where the linear inequality \( \trans{\pi} x \le \pi_{0} \) must be satisfied for all \( x \in P \).
%that is, the hyperplane \(  \setdef{x\in \real^{d}}{\trans{\pi} x = \pi_{0}} \) must contain \( P \) on one of its closed sides. 
We call \( P \subset \real^{d} \) full dimensional if its dimension is \( d \). For a  full dimensional polytope, the affine hull of its facet  \( \Gamma_{i}^{d-1} \) is a hyperplane, denoted by \( H_{i} = \setdef{x\in \real^{d}}{\trans{(\pi^{i})} x = \pi^{i}_{0}} \) and referred as its \emph{defining hyperplane}. As a convention, the direction of \( \pi^{i} \) for \( H_{i} \) is chosen to point outwards from \( P \). 
%\kscomment{Every \( k \)-face \( \Gamma^{k} \) of \( P \) (\( k< \dim P \)) is the convex hull of exactly \( k+1 \) vertices.} 

The geometry of the arrangement of hyperplanes and polytopes is difficult to comprehend, especially in high dimensions. Because of this, they are represented using \emph{incidence graph} (also called the \emph{facial lattices} or \emph{face lattices}). The incidence graph of an arrangement or a polytope contains the incidence relationship between various faces. It is a layered (undirected) graph whose nodes have a one-to-one correspondence with faces of the arrangement or polytope, and an edge exists between two nodes if and only if the corresponding faces are incident. All the nodes corresponding to faces of the same dimension constitute a layer. 
% a layer contains nodes corresponding to faces of the same dimension.  %An edge exists between two nodes if the corresponding faces are incident. 
We place the layer corresponding to vertices at the bottom, and the layer corresponding to cells on the top. 
%\sout{We represent the polytope by the substructure of the incidence graph that is under the branch from the node for the face. For convenience, we call the substructure \emph{incidence graph of the polytope}} \footnote{ \sout{The incidence graph of a polytope is equivalent to its \emph{face lattice}. The later is commonly used to represent a polytope. We use the former in the paper for simplicity.}}. 
For example. Fig.~\ref{fig: incidence-graph-initial-construction} shows the incidence graph of the polytope \( oder \) (or projection of \( U( \mathcal{E}^{0}, p^{0})\)). The single node at the top layer corresponds to \( oder \) itself, the four nodes in the middle layer correspond to the four edges $or$, $re$, $ed$ and $do$, and the four nodes in the bottom layer correspond to the four vertices $o$, $r$, $e$ and $d$. The edges between these layers correspond to the incidence relation between the faces, as shown in Figure~\ref{fig:arrangement-1}. Similarly, \ref{fig:incidence-graph-control} shows the incidence graph of the polytope \( acb \) (or aggregated control action \( U_{1} \)).

Furthermore, we have the following remark on the auxiliary information required for storing the incidence graph of an arrangement or a polytope. 
\begin{remark}
When implemented, the incidence graph is usually associated with some auxiliary information that enables numerical implementation of geometric operations such as determining if a hyperplane intersects with a face, or finding the intersection of a hyperplane and an edge. Further details can be found, e.g., in \cite{de2008computational} and \cite{edelsbrunner1987algorithms}. While several choices for auxiliary information are possible, in this paper, we use the following: analytical expressions for hyperplanes, coordinate for vertices, and for every face, the mean of the coordinates of the vertices contained in it. We shall not explicitly  mention this auxiliary information in algorithms in subsequent sections. 
\end{remark}

\subsection{Constructing the Incidence Graph of $U(\mc E,P)$ and $\mathbb{U}(\mc E,P)$}
\label{subsec:arrangement} 
We recall from Section~\ref{sec:state-aggregation-main} that the implementation of \eqref{eq:value-iteration-aggregated} relies on an efficient procedure to construct the representations (i.e., incidence graphs) of $U(\mc E,P)$ and $\mathbb{U}(\mc E,P)$ from that of \( P \). Such a procedure starts with the incidence graph of \( P^{0} = \{p^{0}\} \), which is known trivially, and is to be repeatedly invoked at every iteration in \eqref{eq:value-iteration-aggregated}. The procedure consists of two steps: (I) construct the incidence graph of \( \cube{P} \) from that of \( P \), and (II) construct the incidence graph of $U(\mc E,P)$ and $\mathbb{U}(\mc E,P)$ from that of \( \cube{P} \). The key ingredient in step (II) is a sub-procedure to update incidence of graph of \( \cube{P} \) upon addition of hyperplanes corresponding to \( \mathcal{B}_{ \mathcal{E}} \) (cf. \eqref{eq:feasible-load-set}) and flow capacity constraints to get \( U(\mathcal{E}, P) \) and \(\mathbb{U}(\mc E,P)\) respectively. Implementation of such a sub-procedure exists in well-known algorithms for constructing arrangement of an arbitrary set of hyperplanes, e.g., in\cite{edelsbrunner1986constructing} and \cite[Chapter 7]{edelsbrunner1987algorithms}. On the other hand, to the best of our knowledge, a systematic algorithmic description to execute step (I) does not exist in the literature. The purpose of Section~\ref{sec:cts-control-space} is to address this deficiency.

\begin{remark} 
\label{rem:arrangement-construction}
\leavevmode
\begin{enumerate} 
\item The algorithms in \cite{edelsbrunner1986constructing} and \cite[Chapter 7]{edelsbrunner1987algorithms} have time complexity \( \Theta(| \mathcal{H}|^{d}) \) for constructing a general positioned arrangement of hyperplanes in \( \mathcal{H} \) in a \( d \) dimensional affine space. Note for a connected network with active link set \( \mathcal{E} \), \( U( \mathcal{E}, P) \) is contained in the \( | \mathcal{V}_{l}| -1\) dimensional affine space \( \mathcal{B}_{ \mathcal{E}} \cap \setdef{p \in \real^{ \mathcal{V}}}{p_{i} = 0, \forall\, i\in \mathcal{V}\setminus \mathcal{V}_{l}} \). 
\item For a network with a single supply-demand pair, the hyperplanes and aggregated controls reduce to points and contiguous intervals that can be represented by two numbers. In this case, \( \cube(P) \) can be computed in constant time: e.g., for \( P = (p^{l}, p^{u}] \subset \preal \), \( \cube(P) = [0, p^{u}]  \); and the incidence graph of \( \mathbb{U}(\mc E, P) \) at every state \( ( \mathcal{E}, P) \) can be computed in linear time with respect to the number of infeasible links. 
\item Recalling that \( \Drelone_i \) defined in \eqref{eq:Drel-def} is for a (sub-)network $\mc G_i$ between a single supply-demand pair, the procedure in 2) of this remark can be used for its construction. \( \Drelone_i \) is the set of feasible aggregated control sequences without monotonicity constraint. Therefore, when constructing the incidence graph for the set \( \mathbb{U}(\mathcal{E}, P) \) of aggregated controls at state \( ( \mathcal{E}, P) \), instead of building upon $ \cube{P}$, one needs to use $[-c^{\text{min-cut}}, c^{\text{min-cut}}]$, where $c^{\text{min-cut}}$ is the min-cut capacity of \( \mathcal{G}_{i} \). Similarly, when constructing \( \Drelone_i \) for constant control (as used in Section~\ref{sec:invariance}), one needs to use $ P$ rather than \( \cube{P} \).
\end{enumerate}
\end{remark}

\subsection{Constructing the Incidence Graph of $\cube{P}$ from the Incidence Graph of $P$} 
\label{sec:cts-control-space} 
Remark~\ref{rem:in-orthant} implies that it is sufficient to focus on \( P \subset \preal^{d} \). We first consider \( \cube{p^{0}}  \) for \( p^{0} \in \preal^{d} \). Since the dimensions corresponding to \( p_{i}^{0} =0 \) can be ignored, we assume \( p^{0} \in \spreal^{d} \) without loss of generality. In this case, \( \cube{p^{0}} \subset \preal^{d} \) is a hypercube and its incidence graph can be obtained straightforwardly by Lemma~\ref{lem:arrangement-cube}, whose proof is omitted. Fig.~\ref{fig: incidence-graph-initial-construction} shows an example in \( \real^{2} \). 
\begin{lemma}
\label{lem:arrangement-cube}
For \( p^{0}\in \spreal^{d} \), let $\map{\alpha}{\cube p^0}{\{-1, 0, 1\}^{d} }$ be defined as: $\alpha_i(x)=-1$ if $x_i=0$, $\alpha_i(x)=0$ if $x_i \in (0,p_i^0)$, and $\alpha_i(x)=1$ if $x_i=p_i^0$. 
%
%the position vector \( \alpha(x) \in \{-1, 0, 1\}^{d} \) of \( x\in \cube{p^{0}} \) be such that  \( x_{i} = 0 \) if \( \alpha_{i}(x) = -1 \), \(0< x_{i}  <p_{i}^{0}  \) if \( \alpha_{i}(x) = 0 \) and \( x_{i}(x) = p_{i}^{0} \) if \( \alpha_{i}(x) = 1 \), for all \( i\in [d] \). 
Then,
\begin{enumerate}[(i)]
\item every \( \tilde{\alpha} \in \{-1, 0, 1\}^{d} \) is associated with a \( (d-|\tilde{ \alpha}|_{1}) \)-face \( \Gamma( \tilde{\alpha} ) : = \cl{\setdef{x\in \cube{p^{0}}}{ \alpha(x) = \tilde{\alpha} }} \) of \( \cube{p^{0}} \);
\item two faces \( \Gamma( \alpha^{1}) \) and \( \Gamma( \alpha^{2}) \)  of $\cube p^0$ are incident if \( \alpha^{1} \) and \( \alpha^{2} \) are equal except for one component which equals zero in one among \( \alpha^{1} \) and \( \alpha^{2} \).
\end{enumerate}
\end{lemma}
\begin{remark}
Lemma~\ref{lem:arrangement-cube} (i) describes a procedure to enumerate all the nodes in the incidence graph of $\cube p^0$, in terms of all vectors in $\{-1,0,1\}^d$, whereas (ii) specifies how to add edges to the incidence graph. 
\end{remark}

For a general polytope $P \subset \preal^{d}$, we present a sequential procedure to construct the incidence graph of \( \cube{P} \) from that of \( P \). For this purpose, we define the \emph{projection} and \emph{sweep} of a polytope \( P\subset \preal^{d}  \) in direction \( \unit_{k} \), \( k\in \until{d} \): 
\begin{align}
\label{eq:proj-def}
\proj_{k}(P) :=& \setdef{ p -  p_{k}\unit_{k}}{ p\in P } \\
\label{eq:sweep-def}
\swp_{k}(P) :=&  \setdef{ p - \theta_{k} p_{k}\unit_{k}}{ p\in P, \theta_{k} \in [0, 1] }
\end{align}
It is straightforward that \( P \subseteq \swp_{k}(P) \) and \( \proj_{k}{P}\subseteq \swp_{k}(P) \). In fact, \( \swp_{k}(P) \) is the trace of \( P \) projecting to \( \proj_{k}{P} \) in the direction of \( \unit_{k} \) and therefore, it is also a polytope in \( \preal^{d} \). One can again apply sweep on \( \swp_{k}(P) \) along \( \unit_{i} \) for some \( i\neq k \) and get \( \swp_{i}(\swp_{k}(P)) =  \setdef{ p - \theta_{1} p_{k}\unit_{k} - \theta_{2} p_{i}\unit_{i}}{p\in P, (\theta_{1}, \theta_{2}) \in [0, 1]^2}  \). This motivates us to define sweep for an index set \( I \subseteq \until{d}\) as 
\[ \swp_{I} (P):=  \setdef{ p - \sum_{k\in I} \theta_{k} p_{k}\unit_{k}}{ p\in P, \theta_{k} \in [0, 1]\, \forall k\in I }.  \]
%Similarly, \( \proj_{I}(P) := \setdef{ p -  \sum_{k\in I} p_{k}\unit_{k}}{ p\in P }  \). 
With this definition, \( \cube{P} = \swp_{\until{d}}(P) \) can be obtained by recursively applying sweep on \( P \), e.g., 
\[ \swp_{\until{d}} (P) = \swp_{1}(\swp_{2}(\ldots \swp_{d}(P))). \]
Therefore, in order to obtain \( \cube{P} \), it is sufficient to focus on constructing \( \swp_{k}(P) \) from \( P \) for a given \( k\in \until{d} \). 

Let \( \bar{H}:= \setdef{x\in \real^{d}}{x_{k} = 0} \) and \( \bar{H}^{+} := \{x\in \real^{d}\, |\, x_{k} \ge 0 \} \).\footnote{We do not show subscript $k$ for brevity in notation.} 
For a given polytope \( P \subset \bar{H}^{+}  \) and \( k\in \until{d} \),  \( \swp_{k}(P) \) relates to projection between two affine spaces: \( \aff P \subset \real^{d} \) and \( \aff \proj_{k} P \subset \bar{H} \). We differentiate between the following two scenarios based on the difference in dimensions of these two affine spaces: (I) 
%projection occurs between two affine spaces of equal dimensions, that is, 
\(\dim P = \dim (\proj_{k}{P}) \); and (II) 
%projection happens from a higher dimensional space to a lower one, that is, 
\( \dim P = \dim(\proj_{k} P) +1 \). Scenario I occurs when \( P \)  is contained in a hyperplane that is not orthogonal to \( \bar{H} \). 
%\( P \subset H:= \setdef{x\in \real^{d}}{\trans{\pi} x = \pi_{0}} \) with $\pi_k \neq 0$, i.e., when $H$ is not orthogonal to \( \bar{H} \).
In this case, there is a one-to-one correspondence between the points in \( P \) and \( \proj_{k} P \). Scenario II occurs when either $\dim{P}=d$, or every hyperplane containing $P$ is orthogonal to \( \bar{H} \). 
%or \( \dim{P} = d \), scenario (II) happens and in this case, two different points \( x^{1}, x^{2} \in P \) exit such that \( \proj_{k} x^{1} = \proj_{k} x^{2} \). 
The two scenarios are illustrated in Figure~\ref{fig:sweep-2cases} for \( \real^{2} \) and \( k = 2 \).

\begin{figure}[htbp]
   \centering
      \begin{subfigure}{.2\linewidth}
  \centering
  \begin{tikzpicture} 
\draw [<->] (0,2) node [left] {\( x_{2} \)}-- (0,0) node [below left] {0} -- (2,0) node [below right] {\( x_{1} \)}; 
\draw[fill] (0.8, 1.3) node[above]{\scriptsize \( P \)} circle (0.05cm);
\draw[fill] (0.8, 0) node[below]{\scriptsize \( \mathrm{proj}_{2} P \)} circle (0.05cm);

%\draw[blue, thin] (2.5, 2.5) -- +(0.3, 0.3) node[above right] { \( P \)};
\draw[densely dashed] (0.8, 1.3) to (0.8, 0);
\node[right] at (0.8, 0.7) {\scriptsize \( \mathrm{sweep}_{2} P \)}; 
%\node[above] at (2.5, 0) {\scriptsize \( \mathrm{proj}_{2}(P) \)}; 
%\draw[thin] (3.5, -0.3) -- (3.9, 0.3);

\end{tikzpicture}
  \caption{}
  \label{fig:sweep-2d-case1a}
\end{subfigure}%
   \begin{subfigure}{.2\linewidth}
  \centering
  \begin{tikzpicture} 
\draw [<->] (0,2) node [left] {\( x_{2} \)}-- (0,0) node [below left] {0} -- (2,0) node [below right] {\( x_{1} \)}; 

\draw[ultra thick] (0.4, 1.7) to (1.3, 0.8);
\draw[ultra thick] (0.4, 0) to (1.3, 0);
\draw[densely dashed] (0.4, 1.7) to (0.4, 0);
\draw[densely dashed] (1.3, 0.8) -- (1.3, 0);

\node[above right] at (0.9, 1.2) {\scriptsize \( P \)};
\node[below] at (0.85, 0) {\scriptsize \( \mathrm{proj}_{2} P \)};

\node[right] at (1.3, 0.5) {\scriptsize \( \mathrm{sweep}_{2}P \)};
%\node[above] at (2.5, 0) {\scriptsize \( \mathrm{proj}_{2}(P) \)}; 
\draw[thin] (1, 0.5) -- (1.4, 0.5);

\end{tikzpicture}
  \caption{}
  \label{fig:sweep-2d-case1b}
\end{subfigure}%
\begin{subfigure}{.2\linewidth}
  \centering
  \begin{tikzpicture} 
\draw [<->] (0,2) node [left] {\( x_{2} \)}-- (0,0) node [below left] {0} -- (2,0) node [below right] {\( x_{1} \)}; 
\draw[ultra thick] (0.8, 0.9) to (0.8, 1.7) node[above] {\scriptsize \( P \)};
\draw[fill] (0.8, 0) circle (0.05cm);
\draw[densely dashed] (0.8, 0.9) to (0.8, 0) node[below] {\scriptsize \( \mathrm{proj}_{2}(P) \)}; 

\node[right] at (0.8, 0.8){\scriptsize \( \mathrm{sweep}_{2}(P) \)}; 
%\draw[thin] (3.5, -0.3) -- (3.9, 0.3);

\end{tikzpicture}
  \caption{}
  \label{fig:sweep-2d-case2a}
\end{subfigure}
\begin{subfigure}{.22\linewidth}
  \centering
  \begin{tikzpicture} 
\draw [<->] (1.2,2) node [left] {\( x_{2} \)}-- (1.2,0) node [below left] {0} -- (3.4,0) node [below right] {\( x_{1} \)}; 
\draw[fill] (1.6, 1.7) to (2, 0.8) to (2.7, 1.3) to (1.6, 1.7);
\draw[ultra thick] (1.6, 0) to (2.7, 0);
\draw[densely dashed] (1.6, 1.7) to (1.6, 0);
\draw[densely dashed] (2.7, 1.3) -- (2.7, 0);

\node[above right] at (2.1, 1.5) {\scriptsize \( P \)};
\node[below] at (2.15, 0) {\scriptsize \( \mathrm{proj}_{2} P \)};

\draw[thin] (2.4, 0.7) -- (2.8, 0.7);
\node[right] at (2.7, 0.7) {\scriptsize \( \mathrm{sweep}_{2}P \)};

\end{tikzpicture}
  \caption{}
  \label{fig:sweep-2d-case2b}
\end{subfigure}
\caption{\sf Two possible scenarios for projection: (a) and (b) show the projection onto a space of the same dimension; and (c) and (d) show projection onto a space of lower dimension.}
\label{fig:sweep-2cases}
\end{figure}
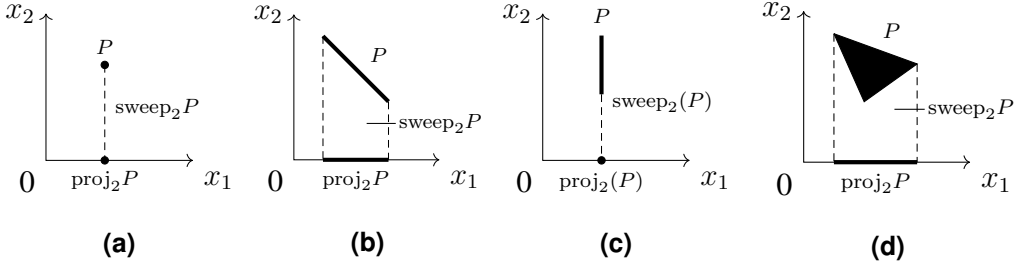

%We start with consideration of simple cases and pattern would show in this process. For \( P = \{x\} \), \( \swp_{k}(P) \) is the edge \( (x, \proj_{k}(x)) \) in line \( \elim_{k}(x) \). For \( P = \conv(x, y) \) (\( x\neq y \)),  if \( \proj_{k}(x) \neq \proj_{k}(y) \),  \( \swp_{k}(P) \) is a 2-face containing \( (x, y) \), \( (x, \proj_{k}(x)) \), \( (y, \proj_{k}(y)) \) and \( (\proj_{k}(x), \proj_{k}(y)) \) as subfaces (i.e., edges); if \( \proj_{k}(x) = \proj_{k}(y) \), \( \swp_{k}(P) \) is the edge \( (x, \proj_{k}(x)) \), where we assume, without loss of generality, \( x_{k} > y_{k} \). In general, we have the following result for a \( l \)-polytopes, \( 0\le l\le d-1 \), which lies in a hyperplane that is not orthogonal to \(\bar{h}:= \setdef{x\in \real^{d}}{x_{k} = 0} \). 

\subsection*{\underline{Scenario I}} 
In this scenario, \( \proj_{k}(P) \) is \emph{affinely isomorphic}%\footnote{Two polytopes \( P_{1} \subseteq \real^{d_{1}} \) and \( P_{2} \subseteq \real^{d_{2}} \) are \emph{affinely isomorphic} if there is an affine map \( \map{f}{\real^{d_{1}}}{\real^{d_{2}}} \) that is a bijection between \sout{the points of} the two polytopes.}  
\footnote{Two polytopes \( P_{1} \subseteq \real^{d_{1}} \) and \( P_{2} \subseteq \real^{d_{2}} \) are \emph{affinely isomorphic} to each other if there exists an affine and bijection map between them.}
to \( P \). Therefore, its incidence graph is identical to that of $P$. The following result relates the incidence graph of \( \swp_{k}(P) \) to that of \( P \) and \( \proj_{k}(P) \). 
\begin{proposition}
\label{prop:sweep-case-1}
Consider an \( n \) dimensional polytope \( P \subset \setdef{x\in \real^{d}}{x_{k} > 0 } \), \(0\le n \le d-1 \). If there exists a hyperplane containing $P$ which is not orthogonal to \( \bar{H} \), then: 
%
%For a number \( k\in [d] \) and a \( n \)-polytope \( P \subset \bar{H}^{+} = \setdef{x\in \real^{d}}{x_{k} > 0 } \) ( \(0\le n \le d-1 \)), if \( P \) is contained in a hyperplane that is not orthogonal to \( \bar{H} \),  
%% with \(0\le l \le d-1 \), \( h:= \setdef{x\in \real^{d}}{\trans{\pi} x = \pi_{0}} \) with \( \pi_{k} \neq 0 \) for some \( k\in \until{d} \), 
%then we have the following facts on \( \swp_{k}(P) \): 
(i)  \( \swp_{k}(P) \) is an \( (n+1) \) dimensional polytope; (ii) an \(l\)-face of \( \swp_{k}(P) \) is either an \(l\)-face of \( P \) or of \( \proj_{k}{P} \), or it is \( \swp_{k}{ \Gamma^{l-1}} \) for some \( (l-1) \)-face \( \Gamma^{l-1} \) of \( P \); and (iii) each face of \( P \) and of \(  \proj_{k}{P} \) is a face of \( \swp_{k}(P) \), and \( \swp_{k}{\Gamma} \) is a face of \( \swp_{k}{P} \) for every face \( \Gamma \) of \( P \). 
\end{proposition}
\begin{proof}
Let $ \hat{\gamma} := 1.2 \max_{x\in P} x_{k} $. Since \( P \) is contained in a hyperplane that is not orthogonal to \( \bar{H} \), the line segment \( [\zerobf, \hat{\gamma} \unit_{k}] \) is not parallel to \( \aff{ P} \). Let \( \hat{P}: =  P-[0, \hat{\gamma} \unit_{k}] \). 
 %is a \( (n+1) \)-prism with basis \( P \) \footnote{ For a \( d-1 \) polytope \( P \) and a segment \( [0, x] \) (\( x \neq 0 \)) not parallel to \( \aff{P} \), the vector sum \( P + [0, x] \) is a \( d \)-polytope and called the \( d -\)\emph{prism} with basis \( P \) \cite{\cite[Chap 4.4]{grunbaum1967convex}}. }. 
We have the following facts on \( \hat{P} \) \cite[Chapter 4.4]{grunbaum1967convex}: \( \hat{P} \) is an \( (n+1) \) dimensional polytope; an \(l\)-face of \( \hat{P} \) is either an \(l\)-face of \( P \) or of \( P - \{ \hat{\gamma} \unit_{k} \} \), or it is the Minkowski sum of singleton \(-\{ \hat{\gamma} \unit_{k} \}  \) and some \( (l-1) \)-face of \( P \); each face of \( P \) and of \(  P - \{ \hat{\gamma} \unit_{k} \}  \) is a face of \( \hat{P} \); and the Minkowski sum of singleton \( -\{ \hat{\gamma} \unit_{k} \} \) and any face of \( P \) is a face of \( \hat{P}  \). 

It is then sufficient to show that \( \swp_{k}{P} \) and  \( \hat{P} \) are combinatorial isomorphic\footnote{Two polytopes \( P_{1}\) and \( P_{2} \) are \emph{combinatorial isomorphic} to each other if there exists a one-to-one correspondence \( \varphi \) between the set of faces in \( P_{1} \) and the set of faces in \( P_{2} \), such that  \( \varphi \) is inclusion-preserving, i.e., for two faces \( \Gamma_{1} \) and \( \Gamma_{2} \) of \( P_{1} \), \( \Gamma_{1} \subset \Gamma_{2}  \) if and only if \( \varphi(\Gamma_{1}) \subset \varphi(\Gamma_{2})\).}.
By definition, \( \swp_{k}(P) = \hat{P} \cap \bar{H}^{+} \). The correspondence between face \( \hat{\Gamma} \)  of \( \hat{P} \)  and face \( \tilde{\Gamma} \)  of \( \swp_{k}(P) \) is as follows: \( \tilde{\Gamma} = \hat{\Gamma} \) if \( \hat{\Gamma} \subset \setdef{x\in \real^{d}}{x_{k} > 0 }  \); \( \tilde{\Gamma} = \proj_{k}(\hat{\Gamma}) \) if \( \hat{\Gamma} \subset \setdef{x\in \real^{d}}{x_{k} < 0 }  \); and \( \tilde{\Gamma} = \hat{\Gamma} \cap \bar{H}^{+} \) otherwise. 
%The facets of \( \hat{P} \) include \( P \subset \bar{H}^{+} \), \( P -  \{\hat{\gamma} \unit_{k} \}  \subset \bar{H}^{-} := \setdef{x\in \real^{d}}{x_{k} < 0 } \) and \( \Gamma^{n-1} - [0, \hat{\gamma} \unit_{k}]  \) for all ridge \( \Gamma^{n-1}  \) of \( P \). The facet \( \Gamma^{n-1} - [0, \hat{\gamma} \unit_{k}]  \) is orthogonal to and intersects with \( \bar{H} \) on \( \proj_{k}{ \Gamma^{n-1}} \) for every ridge \( \Gamma^{n-1}  \) of \( P \). It is then straightforward to see that \( \swp_{k}(P) \) and \( \hat{P} \) are combinatorial isomorphic, with facet \( P \) of \( \swp_{k}{P} \) corresponding to facet \( P \) of \( \hat{P} \), facet \( P - \{\hat{\gamma} \unit_{k}\}  \) corresponding to \( \proj_{k}(P) \), and \( \Gamma^{n-1}  - \{\hat{\gamma} \unit_{k}\} \) corresponding to \( \proj_{k}( \Gamma^{n-1} ) \) for all ridges \( \Gamma^{n-1} \) of \( P \).
\end{proof}

\begin{remark}
\label{rem:sweep-case-1}
The complement of the \( P \subset \{x\in \real^{d}\, |\, x_{k} > 0 \} \) condition in Proposition~\ref{prop:sweep-case-1} can be handled as follows. If \( P \subset \bar{H} \), then trivially \( \swp_{k}(P) = P = \proj_{k}(P) \). For the remaining scenario when $P \cap \bar{H} \neq \emptyset$, one can use a standard perturbation trick~\cite{edelsbrunner1987algorithms}. Let $\hat{P}$ be a small perturbation of $P$ such that  \( \hat{P} \subset \setdef{x\in \real^{d}}{x_{k} > 0 } \), and hence to which Proposition~\ref{prop:sweep-case-1} applies. The incidence graph of $P$ is then obtained from that of $\hat{P}$ by merging faces which are in close proximity (relative to the perturbation) and maintaining incidence relationships between the remaining faces. 
%
%For \( P \subset \bar{H}^{+} \setminus \bar{H} \), one can consider a disturbed version \( \tilde{P} \)  of \( P \) such that \( \tilde{P} \subset \setdef{x\in \real^{d}}{x_{k} > 0 } \) and employ Proposition~\ref{prop:sweep-case-1} to obtain the incidence graph of \( \tilde{P} \). The incidence graph of \( P \) can then be obtained by merging the ``close'' faces (within disturbances) of \( \tilde{P} \) while keeping their incidence relationships. 
\end{remark}

Example~\ref{eg:sweep-case-1} illustrates how to use Proposition~\ref{prop:sweep-case-1} to obtain the incidence graph of \( \swp_{k}(P) \) in \( \real^{2} \).
\begin{example}
\label{eg:sweep-case-1}
Consider the polytope \( P \) corresponding to the line segment between points $a$ and $b$ in Figure~\ref{fig:sweep-2d-polytope}. The figure also shows the corresponding $\hat{P}$, \( \proj_{2}(P) \) and \( \swp_{2}(P) \). The subgraph shown in solid black in Figure~\ref{fig:sweep-2d-face-lattice} is the incidence graph of $P$, whereas the subgraph shown in gray, which is identical to the solid black one, is the incidence graph of $\proj_{2}{P}$, where \( a' = \proj_{2}(a) \), \( b' = \proj_{2}(b) \). The incidence graph of \( \swp_{2}(P) \) is constructed using Proposition~\ref{prop:sweep-case-1} as follows:  
\begin{enumerate}
\item The vertices (\( 0 \)-faces) contain vertices both of \( P \), i.e., \( a \)  and \( b \), and of \( \proj_{2}(P) \), i.e., \( a' \)  and \( b' \);
\item The edges (\( 1 \)-faces) contain both the edge of \( P \), i.e., \( ab \), the edge of \( \proj_{2}(P) \), i.e., \( a'b' \), and the edges formed by sweep of vertices, i.e., \( aa' \) and \( bb' \); 
\item Edges \( aa' \) and \( bb' \) are incident to \( a \), \( a' \)  and \( b \), \( b' \), respectively, as they result from sweeping, contain the corresponding vertices, and their dimensions differ by \( 1 \); 
\item \( \swp_{2}(P) \) is a two dimensional polytope. Its incidence graph contains itself as a \( 2 \)-face incident to all the edges. 
\end{enumerate}
%
%The incidence graph of \( \swp_{2}(P) \) shown in Fig.~\ref{fig:sweep-2d-face-lattice} is constructed from \( P \) according to Proposition~\ref{prop:sweep-case-1}, where \( a' := \proj_{2}(a) \), \( b':= \proj_{2}(b) \) and the  substructure in black is the incidence graph of \( P \). 

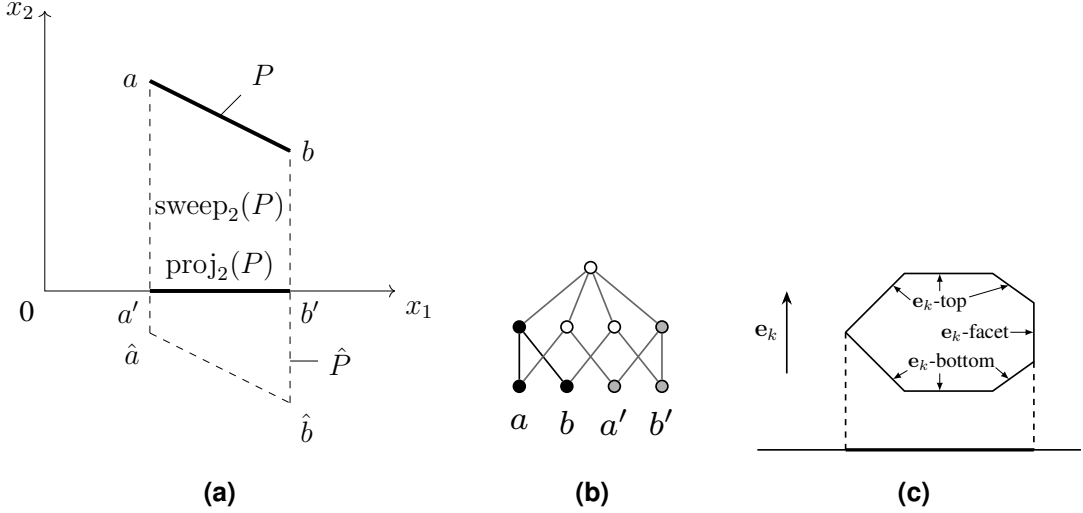
\begin{figure}[htbp]
   \centering
   \begin{subfigure}[t]{.4\linewidth}
  \centering
  \begin{tikzpicture} 
\draw [<->] (0,4) node [left] {\( x_{2} \)}-- (0,0) node [below left] {0} -- (5,0) node [below right] {\( x_{1} \)}; 
\draw[ultra thick] (1.5, 3) node[left] {\( a \)} to (3.5, 2) node[right] {\( b \)};
\draw[thin] (2.5, 2.5) -- +(0.3, 0.3) node[above right] {\( P \)};
\draw[dashed] (1.5, -0.6) node[below left] { \( \hat{a} \) } to (1.5, 0) node[below left]{\( a' \) } to (1.5, 3);
\draw[dashed] (3.5, -1.6) node [below right] {\( \hat{b} \)} to (3.5, 0) node[below right]{\( b' \) } to (3.5, 2);
\draw[ultra thick] (1.5, 0) -- (3.5, 0);
\draw[dashed] (1.5, -0.6) -- (3.5, -1.6);

\draw[thin] (3.5, -1.0) -- +(0.4, 0) node[right] {\( \hat{P} \)};

\node at (2.5, 1.2) { \( \mathrm{sweep}_{2}(P) \)}; 
\node[above] at (2.5, 0) { \( \mathrm{proj}_{2}(P) \)}; 

\end{tikzpicture}
  \caption{}
  \label{fig:sweep-2d-polytope}
\end{subfigure}%
\begin{subfigure}[t]{.2\linewidth}
  \centering
  \begin{tikzpicture}
\node[face] (1)  at (-0.6, 0) {};
\node[face] (2) at (-0.2, 0) {};
\node[face, fill=black!30] (3)  at (0.2, 0) {};
\node[face, fill=black!30] (4) at (0.6, 0) {};

\node at (-0.6, -0.3){{\scriptsize \( a \) }};
\node at (-0.2, -0.3){{\scriptsize \( b \) }};
\node at (0.6, -0.3){{\scriptsize \( b' \) }};
\node at (0.2, -0.3){{\scriptsize \( a' \) }};

\node[face] (5) at (-0.6, 0.5) {};
\node[face, fill=none] (6) at (-0.2, 0.5) {};
\node[face, fill=none] (7) at (0.2, 0.5) {};
\node[face, fill=black!30] (8) at (0.6, 0.5) {};

\node[face, fill=none] (9) at (0, 1) {};

\draw (1) -- (5); \draw (2) -- (5);
\draw[color=black!60] (1)[color=black!60] -- (6); \draw[color=black!60] (3) -- (6);
\draw[color=black!60] (2)[color=black!60] -- (7); \draw[color=black!60] (4) -- (7);
\draw[color=black!60] (3)[color=black!60] -- (8); \draw[color=black!60] (4) -- (8);

\draw[color=black!60] (5) -- (9); \draw[color=black!60] (6) -- (9); \draw[color=black!60] (7) -- (9); \draw[color=black!60] (8) -- (9);

\end{tikzpicture}
  \caption{}
  \label{fig:sweep-2d-face-lattice}
\end{subfigure}
\begin{subfigure}[t]{.3\linewidth}
  \centering
  \begin{tikzpicture} 
\draw[thick] (-0.5,0) -- (5, 0);
\draw[thick] (1, 2) -- (2, 3) -- (3.5, 3) -- (4.2, 2.5) -- (4.2, 1.5) -- (3.5, 1) -- (2, 1) -- (1, 2);
\draw[thick, dashed] (1, 2) -- (1, 0);
\draw[thick, dashed] (4.2, 1.5) -- (4.2, 0);
\draw[ultra thick] (1, 0) -- (4.2, 0);

\node at (2.6, 2.5) {\small \( \mathbf{e}_{k} \)-top};
\draw[-latex]  (2.1, 2.55) -- (1.8, 2.8);
\draw[-latex]  (2.6, 2.7) -- (2.6, 3);
\draw[-latex]  (3.1, 2.6) -- (3.78, 2.8);

\node at (2.8, 1.5) {\small \( \mathbf{e}_{k} \)-bottom};
\draw[-latex]  (2.05, 1.35) -- (1.8, 1.2);
\draw[-latex]  (2.6, 1.3) -- (2.6, 1);
\draw[-latex]  (3.55, 1.4) -- (3.78, 1.2);

\node at (3.2, 2) {\small \(\mathbf{e}_{k} \)-facet};
\draw[-latex] (3.8, 2) --(4.2, 2);

\draw[thick, -Stealth] (0, 1.3)--  (0, 2.7);
\node[left] at (0, 2) {\( \mathbf{e}_{k} \)};

\end{tikzpicture}
  \caption{}
  \label{fig:facet-classification}
\end{subfigure}
\caption{\sf (a) a sample polytope $P$ in $\real^2$, its sweep, projection, and related objects; (b) the incidence graph of $\swp_2(P)$; and (c) different facets according to the direction \( \unit_{k} \)}
\label{fig:sweep-2d-example}
\end{figure}
\end{example}

\subsection*{\underline{Scenario II}} 
For the second scenario, we first identify the faces of \( P \) that play a role in the projection and then resort to Proposition~\ref{prop:sweep-case-1} for construction of \( \swp_{k} (P) \). It is sufficient to consider \( P \subset \preal^{d} \) to be full dimensional, i.e., \( \dim P =d \). Otherwise, one can work in the affine space \( \aff P \), which is orthogonal to \( \bar{H} \), and the same results hold. Let \( \Gamma^{d-1}_{i}\) be a facet of \( P \subset \bar{H}^{+} \), recall that the direction of \( \pi^{i} \) for its defining hyperplane \( H_{i} \)  is pointed outwards from the polytope. A direction vector \( \mu \in \real^{d} \) classifies the facets of \( P \) into three types \cite{burger1996polytope} according to the value of \( \trans{\mu} \pi^{i} \): \( \mu-\)\emph{facet} for \( \trans{\mu} \pi^{i} =0 \), \( \mu- \)\emph{bottom} for \( \trans{\mu} \pi^{i} < 0 \) and \( \mu- \)\emph{top} for \( \trans{\mu} \pi^{i} > 0 \). With this definition, a facet \( \Gamma^{i} \) belongs to \( \unit_{k}- \)top if \( \pi^{i}_{k} >0 \), \( \unit_{k}- \)facet if \( \pi^{i}_{k} =0 \) and \( \unit_{k} - \)bottom if \( \pi^{i}_{k}<0 \). This is illustrated in Figure~\ref{fig:facet-classification}.

%\begin{figure}[htbp]
%   \centering
%  \includestandalone[width=.3\linewidth]{facet-classification}
%\caption{Different facets according to the direction \( \unit_{k} \) }
%\label{fig:facet-classification}
%\end{figure}

The \( \unit_{k}\)-top and  \( \unit_{k} \)-bottom facets can be described as the facets that are ``visible from the direction \( -\unit_{k} \)'' and ``visible from the direction \( \unit_{k} \)'', respectively. 
%\footnote{While these terms are inherited from \cite{burger1996polytope}, however, the meaning for the \( \mu- \)\emph{bottom} and \( \mu- \)\emph{top} are exchanged in this paper, as the author believe it would make sense to put the \emph{top} above the \emph{bottom} in the direction \( \mu \). }. 
For the projection concerned in \( \swp_{k}(P) \), only points in \( \unit_{k}- \)top play a role. This is straightforward to see in \( \real^{2} \) and \( \real^{3} \). For example, the top vertex of the vertical edge in Fig.~\ref{fig:sweep-2d-case2a} and the top edge of the triangular face in \ref{fig:sweep-2d-case2b} completely determine the sweeps. In general, Proposition~\ref{prop:top-facet-boundary-ridge} shows the same is true for \( \real^{d} \), where we say point \( x\in \real^{d} \) is \emph{shaded} by point \( \hat{x} \in \real^{d} \) in direction \( \unit_{k} \) if \( \hat{x}_{k} \ge x_{k} \) and \( \hat{x}_{i} = x_{i} \) for all \( i \in [d]\setminus \{k\} \). It is clear from the definition that a point plays no role in the projection along \( \unit_{k} \) if it is shaded by another point of \( P \) in \( \unit_{k} \). 
%Consequently, the construction of \( \swp_{k}(P) \) depends entirely on the \( \unit_{k} \)-top facets of \( P \) and all other faces that are not included in any of them can hence be removed. 
The ridges in \( \unit_{k}-\)top facet of \( P \) are of two types: one is in the intersection between two \( \unit_{k}- \)top facets, and the other is in the intersection between a \( \unit_{k}-\)top facet and either a \( \unit_{k}-\)facet or a \( \unit_{k}-\)bottom facet. Proposition~\ref{prop:top-facet-boundary-ridge} implies that the second type of ridges determine the boundaries of \( \swp_{k}(P) \) and $\proj_k(P)$. They are hence called \emph{boundary ridges} of \( P \) in direction \( \unit_{k} \). 

\begin{proposition}
\label{prop:top-facet-boundary-ridge}
For a full dimensional polytope \( P \subset \bar{H}^{+} \) and \( k\in [d] \), the following are true:
\begin{enumerate}[(i)]
\item every point in \( P \) is shaded in direction \( \unit_{k} \) by a point in a \( \unit_{k}- \)top facet of \( P \);
\item every \( \unit_{k} \)-top facet is a facet of \( \swp_{k}(P) \); and
\item for a ridge \( \Gamma^{d-2} \) of \( P \), \( \swp_{k}( \Gamma^{d-2}) \) is a facet of \( \swp_{k}(P) \) if and only if \( \Gamma^{d-2} \) is a boundary ridge and \( \Gamma^{d-2} \not\subseteq \bar{H} \). 
\end{enumerate} 
\end{proposition}

With Proposition~\ref{prop:sweep-case-1} and Proposition~\ref{prop:top-facet-boundary-ridge}, \( \cube P \)  of a full dimensional polytope \( P \subset \preal^{d} \) can be constructed as follows: set \( k = 1\); while $k \leq d$, do the following:
\begin{enumerate}[(I)]
\item find \( \unit_{k} \)-top facets  of \( P \) and remove all faces of \( P \)  that are not contained in anyone of them;
\item find the boundary ridges and construct their sweep along \( \unit_{k} \) according to Proposition~\ref{prop:sweep-case-1} and Remark~\ref{rem:sweep-case-1}; 
\item add a facet in \( \bar{H} \) that is incident to the projections of all boundary ridges along \( \unit_{k} \) and a cell, corresponding to \( \swp_{k}(P) \), that is incident to all the facets;
\item set \( P = \swp_{k}(P) \), $k=k+1$, and repeat;
\end{enumerate}
where the second step is possible because every boundary ridge belongs to a \( \unit_{k}\)-top facet that is not orthogonal to \( \bar{H} \). The above four-step procedure is illustrated in Example~\ref{eg:sweep-case-2}. 

\begin{example}
\label{eg:sweep-case-2}
Consider \( P = U_{1} \subset \preal^{2} \) shown in Fig.~\ref{fig:arrangement-1}, in this case \( \cube{U_{1}} = \swp_{1} (\swp_{2} (U_{1})) \). These two sweep operations are shown in Fig.~\ref{fig:sweep-1} and  \ref{fig:sweep-2}, respectively. In particular, Fig.~\ref{fig:sweep-incidence-graph} shows the incidence graph of \( \swp_{2}(U_{1}) \), where the solid black substructure is inherited from the incidence graph of \( U_{1} \). As can be seen, the \( \unit_{2} \)-bottom facet \( cb \) is removed in \( \swp_{2}(U_{1}) \); the \( \unit_{2} \)-top facets, edges \( ac \) and \( ab \), and the associated subfaces, vertices \( a \), \( b \), \( c \), persist; the boundary ridges, vertices \( b \) and \( c \), play critical roles in the construction of \( \swp_{2}(U_{1}) \): their projections, i.e., \( b' \) and \( c' \), and sweeps, i.e., \( bb' \) and \( cc' \), are added as ridges and facets, respectively.

\begin{figure}[htbp]
\centering
 \begin{subfigure}[b]{.3\linewidth}
  \centering
  \begin{tikzpicture} 
\draw [<-> ] (0,4.5) node[left] {\large \( u_{2} \)} --(0,0) node[below left] at (0, 0) {\large \( o \)} -- (4.5,0) node[below right]{\large \( u_{1} \)}; 
\draw[fill=gray] (3.43, 3.43)node[above right]{\large \( a \) } -- (3.65, 2.35) node[above right]{\large \( b \) }-- (2.35, 3.65) node[below left]{\large \( c \) }-- (3.43, 3.43);
\draw[densely dashed] (3.65, 2.35) -- (3.65, 0) node[above right]{\large \( b' \) };
\draw[densely dashed] (2.35, 3.65) -- (2.35, 0) node[above left]{\large \( c' \) };
\draw[ultra thick] (2.35, 0) -- (3.65, 0);
\draw (3, 0) --(3, -0.2) node[below]{\large \( \mathrm{proj}_{2}(P) \) };
\draw (2.5, 1.7) --(2.2, 1.7) node[left]{\large \( \mathrm{sweep}_{2}(P) \) };
\draw (3, 3.3) --(3, 3.7) node[above]{\large \( P \) };
%\draw[->] (4.3, 2.4) --(4.3, 1.8); 
%\node at (4.4, 2.1)[right] { \large \( e_{2} \) };

\end{tikzpicture}
  \caption{ }
  \label{fig:sweep-1}
\end{subfigure}
\begin{subfigure}[b]{.3\linewidth}
  \centering
  \begin{tikzpicture}
\node[face] (1)  at (0, 0) {};

\node[face] (ac) at (-0.5, -0.5) {};
\node[face, fill=black!30] (bb1)  at (0.5, -0.5) {};
\node[face] (ab) at (0, -0.5) {};

\node[face] (a) at (-0.5, -1) {};
\node[face] (c)  at (0.5, -1) {};
\node[face] (b) at (0, -1) {};
\node at (-0.5, -1.3){{\scriptsize \( a \) }};
\node at (0, -1.3){{\scriptsize \( b \) }};
\node at (0.5, -1.3){{\scriptsize \( c \) }};

\node[face, fill=black!30] (b1) at (1, -1) {};
\node[face, fill=black!30] (c1)  at (1.5, -1) {};
\node[face] (b) at (0, -1) {};
\node at (1, -1.3){{\scriptsize \( b' \) }};
\node at (1.5, -1.3){{\scriptsize \( c' \) }};

\node[face, fill=black!30] (cc1) at (1, -0.5) {};
\node[face, fill=black!30] (b1c1)  at (1.5, -0.5) {};

\draw (1) -- (ac); \draw[color=black!60] (1) -- (bb1); \draw (1) -- (ab); 
\draw[color=black!60] (1) --(cc1); \draw[color=black!60] (1) -- (b1c1);
\draw (ac) -- (a); \draw (ac) -- (c); 
\draw[color=black!60] (bb1) -- (b); \draw[color=black!60] (bb1) -- (b1); 
\draw (ab) -- (b); \draw (ab) -- (a);
\draw[color=black!60] (cc1) --(c); \draw[color=black!60] (cc1)--(c1);
\draw[color=black!60] (b1c1)--(b1); \draw[color=black!60] (b1c1)--(c1);

\end{tikzpicture}
  \caption{ }
  \label{fig:sweep-incidence-graph}
\end{subfigure}
\begin{subfigure}[b]{.3\linewidth}
  \centering
  \begin{tikzpicture} 
\draw [<-> ] (0,4.5) node[left] {\Large \( u_{2} \)} --(0,0) node[below left] at (0, 0) { \Large\( o \)} -- (4.5,0) node[below right]{\Large \( u_{1} \)}; 
\draw[fill=gray] (3.43, 3.43) -- (3.65, 2.35) -- (3.65, 0) -- (2.35, 0) -- (2.35, 3.65) -- (3.43, 3.43);
\draw[densely dashed] (2.35, 3.65) -- (0, 3.65);
% \draw[->] (2.0, 4.2) -- (1,  4.2); 
% \node at (1.5, 4.4) {\Large  \(e_{1}\) };
\draw[ultra thick] (0, 0) -- (0, 3.65);
\node at (1.3, 1.9) { \large \( \mathrm{cube}(P) \) };
\draw (2.5, 0.9) -- (2.2, 0.9) node[left]{\large  \( \mathrm{sweep}_{2}(P) \) };

\end{tikzpicture}
  \caption{}
  \label{fig:sweep-2}
\end{subfigure}
\caption{\sf Illustration of sweep: (a) sweep of \( U_{1} \); (b) the incidence graph of \( \swp_{2}(U_{1}) \); and (c) \( \cube{U_{1}} \) as sweep of \( \swp_{2}(U_{1}) \).}\label{fig:sweep}
\end{figure}
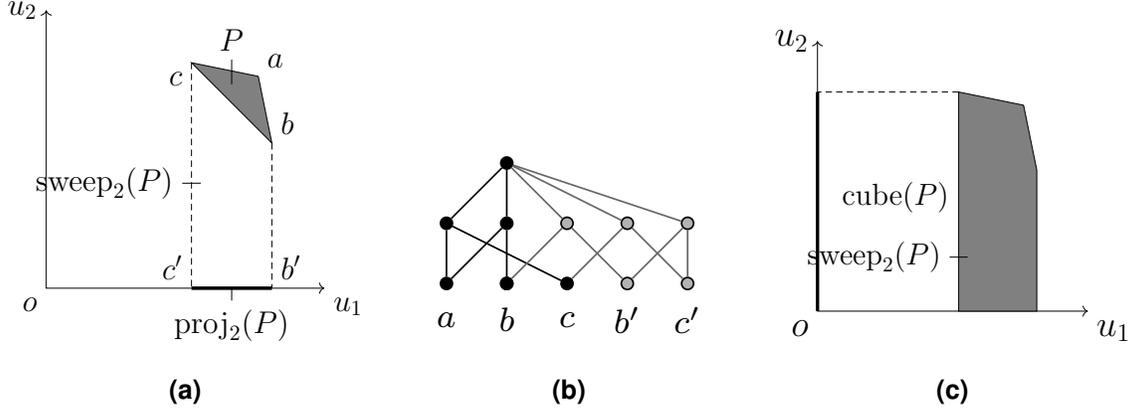
\end{example}

\begin{remark}
If $(\mc E,P)$ has multiple connected components, then it is efficient to construct incidence graphs for each component separately, using the technique presented in this section. However, one first needs to construct the incidence graph for the projection of $P$ onto the subspace associated with every component. For a component that does not contain nodes $\{i, \ldots, j\}$, the corresponding projection is obtained iteratively as $ \proj_{i}(\ldots \proj_{j} (P))$. The required operations on the incidence graph of $P$ to execute these projections are already contained in Propositions~\ref{prop:sweep-case-1} and \ref{prop:top-facet-boundary-ridge}.
%
%Finally, we consider the case when the network gets disconnected. In this case, each connected network component is associated with a subspace of balanced supply-demand vectors that are supported on the network component. One first needs to compute for each of the connected network components the projection of \( P \) onto the associated subspace, and then perform the sweep operation within the subspace. The additional step of computing the projection of \( P \) can be achieved recursively in the same fashion as that for sweep of $ P $, that is, $ \proj_{\{i, j\}} P = \proj_{i}(\proj_{j} P) $. The implementation of the projection is straightforward using Proposition~\ref{prop:sweep-case-1} and  \ref{prop:top-facet-boundary-ridge}. 
\end{remark}

%%% Local Variables:
%%% mode: latex
%%% TeX-master: "dp-load-shedding-main"
%%% End:

%\input{complexity}

%%% Local Variables:
%%% mode: latex
%%% TeX-master: "powernetwork-cdc16.tex"
%%% End:

\section{Approximation Algorithm and Simulations}
\label{sec:approx-algorithm-simulations}

\subsection{Approximation Algorithm via Projection}
\label{sec:approximation-method} 
The exponential dependence of the time complexity on $d=|\mc V_l|-1$ (see Remark~\ref{rem:arrangement-construction})
%, as discussed in Section~\ref{sec:complexity-analysis}, 
can be prohibitive for networks that contain large number of non-transmission nodes. We now outline a strategy to project the admissible control actions onto a lower dimensional space.  Aggregation and search in the lower dimensional space then gives an approximation. 

%The approximating methods are based on incomplete tree search. Conceptually, it is similar to generate only a subset of the aggregated search tree by ignoring certain paths. However, we realize approximation by restricting the searching within a low dimensional subspace of the action set, that is obtained from projection transformation. This is motivated by the fact that the number of nodes in the aggregated search increases exponentially with the dimension \(d:= |\mathcal{V}_{l}| \). 

Consider a network \( \mathcal{G} =( \mathcal{V}, \mathcal{E}) \) with supply-demand vector \( p \). For the sake of presentation in this section, assume, without loss of generality that $\mc V=\mc V_l$, i.e., every node is a non-transmission node; if this is not the case, then one can focus only on the subspace of control actions corresponding to the non-transmission nodes. 
% let \( \hat{p}\in \real^{d} \) be the collection of nonzero components of \( p \), and \( u \in U'( \mathcal{E}, \hat{p}) \subset \real^{d} \) be the associated action set. It is equivalent to work with \( \hat{p} \) and \( U'(\mathcal{E}, \hat{p}) \). 
 Let \( \Phi = [\Phi_{1}, \ldots, \Phi_{|\mc V_l|}] \in \real^{|\mc V_l| \times |\mc V_l|} \) be an orthonormal (transformation) matrix, and let \( B \subset [|\mc V_l|] \) be an index set. The approximation strategy, which is parametrized by $(\Phi,B)$, considers aggregated set of admissible control actions in the subspace  \( U( \mathcal{E}, P) \cap \range(\Phi_{B}) = \setdef{u\in U( \mathcal{E}, P)}{ \trans{\Phi_{i}} u = 0, \, i \not \in B } \)\footnote{\( \range(\Phi_{B}) \) denotes the range of matrix \( \Phi_{B}\). }, i.e., in the subspace of control actions which can be expressed as a linear combination of \( \Phi_{i} \), \( i\in B \). Remark~\ref{rem:arrangement-construction} implies that this reduces the dimension, and hence correspondingly the time complexity, from $|\mc V_l| -1 $ to $|B|$. Moreover, since the constraints $\trans{\Phi_{i}} u = 0$, $i \not \in B$, are hyperplanes, they can be easily integrated into the construction of arrangements to get \( \mathbb{U}( \mathcal{E}, P) \cap \range(\Phi_{B})\). In fact, by setting \( \Phi_{i} = {\onebf}/{\sqrt{|\mc V_l|}} \) for some \( i\not \in B \), the constraint \( \trans{\Phi_{i}} u = 0 \) is the balance constraint for a connected network (cf. \eqref{eq:load-balance}). 
% then one should only consider the actions in the new action set \( U'( \mathcal{E}, \hat{p}) \cap \range(\Phi_{B}) = \setdef{u\in U'( \mathcal{E}, \hat{p})}{ \trans{\Phi_{i}} u = 0, \forall \, i\in B^c } \), which is the intersection with \( |B^c| \) hyperplanes. The dimension of action set is decreased by \( |B^c| \). The implementation of this approximation is straightforward. We add these \( |B^c| \) additional hyperplanes into the arrangement and only keep the substructure that is inside their intersection. 
%\sout{The procedure can be extended to $\Phi$ being \kscomment{arbitrary} transformations.} % \ksmargin{what does arbitrary mean here?}

%Example~\ref{eg:proportional-control} is a simple application of the approximation method described.
%\ksmargin{talk about generalization in the example}
%\ksmargin{is the last sentence needed ?}
\begin{example}
\label{eg:proportional-control}
Consider a network with initial supply-demand vector \( p^{0} \). By choosing \( B = \{1\}\) and \( \Phi_{1} = p^{0}/\|p^{0}\|_{2} \), we get \emph{proportional control policies}~\cite[Section 6.1.1]{Bienstock:16}, i.e., a class of control policies whose action set at state \( ( \mathcal{E}, p) \) is $\setdef{ \lambda p}{ 0\le \lambda \le 1}$. 
%The feasible set for control actions is \( \setdef{( \lambda_{1}, \ldots, \lambda_{N-1} ) \times p^{0} }{0 \le \lambda_{N-1} \le \ldots \le \lambda_{1} \le 1  } \).    
\end{example}

The one-dimensional search space resulting from proportional control policy, as shown in Example~\ref{eg:proportional-control}, is favorable for computational purposes. However, the projection-based approximation strategy implies that one could possibly find better control actions, using  comparable computational budget, by using different projections. This is illustrated using simulations on IEEE 39 benchmark system in Section~\ref{sec:simulations}.

\subsection{Simulations}
\label{sec:simulations}
\begin{figure}[htb!]
\begin{subfigure}[b]{.49\linewidth}
  \centering
\centering
   \includegraphics[width=\linewidth]{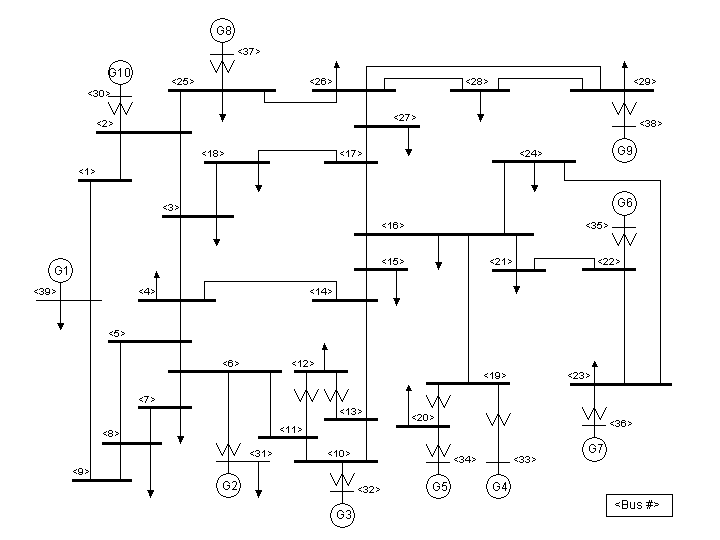}
     \caption{ }
   \label{fig:ieee39-pic}
\end{subfigure}
\begin{subfigure}[b]{.49\linewidth}
  \centering
  \begin{tikzpicture}                   
 \node[supply node] (39) at (0, -1) {39};   \node[main node] (1) at (0, 1) {1};
 \node[main node] (2) at (1.5, 2.5) {2};  \node[main node] (9) at (1.5, -2.5) {9};
  \node[main node] (3) at (3.5, 2.5) {3};   \node[main node] (8) at (3.5, -2.5) {8}; 
  \node[demand node] (4) at (5, 1) {4};   \node[main node] (5) at (5, -1) {5}; 
  \node[main node] (6) at (6, -2.5) {6};     \node[main node] (7) at (4.5, -4) {7};    \node[main node] (31) at (6.1, -4) {31}; 
  \node[main node] (14) at (7, 1) {14};      \node[main node] (13) at (8.5, -0.2) {13}; 
   \node[main node] (11) at (7.5, -3.5) {11}; \node[main node] (32) at (10.5, -3.5){32};
 \node[main node] (12) at (7.5, -1.8) {12};  \node[main node] (10) at (9, -2.3) {10};  
 \node[main node] (18) at (5.3, 3.3) {18};  \node[main node] (17) at (7, 4) {17}; 
   \node[main node] (15) at (9, 1.6) {15};     \node[demand node] (16) at (9, 3.4) {16};  
   
    \node[main node] (25) at (3, 4) {25};  
    \node[main node] (26) at (4.2, 5.5) {26};   \node[main node] (27) at (5.8, 5.5) {27};    
         \node[main node] (30) at (0, 4) {30};      \node[main node] (37) at (1.5, 5.5) {37}; 
     \node[main node] (28) at (3.2, 6.8) {28};      \node[main node] (29) at (5.2, 6.8) {29};  
     \node[main node] (38) at (7, 6.8) {38}; 
     
     \node[main node] (24) at (10, 5) {24};       \node[main node] (23) at (12, 5) {23};    \node[main node] (36) at (11, 6.5) {36};
     \node[main node] (21) at (11, 3.4) {21};   \node[main node] (22) at (13, 3.4) {22};    \node[main node] (35) at (13, 1.7) {35};
     \node[main node] (19) at (10.5, 2) {19};        \node[main node] (33) at (11.8, 1) {33};
      \node[main node] (20) at (10.5, 0.3) {20};     \node[main node] (34) at (10.5, -1.5) {34};

%\draw[blue, <-, thick] (39) -- ++ (-1cm, -1cm) node[below] {\( p \)};
%\draw[red, <-, thick] (36) -- ++ (1.5cm, 0.3cm) node[below] {5};
%\draw[blue, ->, thick] (4) -- ++ (-1cm, -1cm) node[below] {10};

\path[main edge, font=\sffamily\large]
	(1) edge node{\( e_{1} \)} (2)        % e1 
	(1) edge node{\( e_{2} \)}  (39)     % e2
	(2) edge node{\( e_{3} \)} (3)	     % e3
	(2) edge node{\( e_{5} \)} (30)	     % e5
	(2) edge node{\( e_{4} \)} (25)	     % e4
	(3) edge node{\( e_{6} \)} (4)	     % e6
	(3) edge node{\( e_{7} \)} (18)	     % e7
	(4) edge node{\( e_{8} \)} (5)  	% e8
	(4) edge node{\( e_{9} \)} (14) 	 % e9
	(5) edge node{\( e_{10} \)} (6)  	 % e10
	(5) edge node{\( e_{11} \)} (8)  	% e11
	(6) edge node[above left]{\( e_{12} \)} (7)   	% e12
	(6) edge node{\( e_{13} \)} (11)  	 % e13
	(6) edge node[left]{\( e_{14} \)} (31)   	% e14
	(7) edge node{\( e_{15} \)} (8)   	% e15
	(8) edge node{\( e_{16} \)}  (9)      % e16
	(9) edge node{\( e_{17} \)} (39)      % e17
	(10) edge node{\( e_{18} \)} (11)   % e18
	(10) edge node[right]{\( e_{19} \)} (13)   % e19
	(10) edge node{\( e_{20} \)} (32) % e 20
	(11) edge node[right]{\( e_{21} \)} (12)   % e21
	(12) edge node{\( e_{22} \)} (13)   % e22
	(13) edge node{\( e_{23} \)} (14)   % e23
	(14) edge node{\( e_{24} \)} (15)   % e24
	(15) edge node{\( e_{25} \)}  (16) 		% e25
	(16) edge node{\( e_{27} \)}  (19) 		% e27
	(16) edge node{\( e_{28} \)}  (21) 		% e28
	(16) edge node{\( e_{29} \)}  (24) 		% e29
	(17) edge node{\( e_{26} \)}  (16) 		% e26
	(17) edge node{\( e_{30} \)}  (18) 		% e30
	(17) edge node{\( e_{31} \)} (27)	     % e31
	(19) edge node{\( e_{32} \)}   (20) 		% e32
	(19) edge node{\( e_{33} \)}   (33) 		% e33
	(20) edge node{\( e_{34} \)}   (34) 		% e34
	(21) edge node{\( e_{35} \)}   (22) 		% e35
	(22) edge node[above right]{\( e_{36} \)}   (23) 		% e36
	(22) edge node{\( e_{37} \)}   (35) 		% e37
	(23) edge node[above right]{\( e_{39} \)}   (36) 		% e39
	(23) edge node{\( e_{38} \)}   (24) 		% e38
	(25) edge node{\( e_{40} \)}  (26)	     % e40
	(25) edge node{\( e_{41} \)}  (37)	     % e41
	(26) edge node{\( e_{42} \)}  (27)	     % e42
	(26) edge node{\( e_{43} \)}  (28)	     % e43
	(26) edge node{\( e_{44} \)}  (29)	     % e44
	(28) edge node{\( e_{45} \)}  (29)	     % e45
	(29) edge node{\( e_{46} \)} (38)	     % e46 
	;     	
\end{tikzpicture} 
  \caption{ }
  \label{fig:ieee39-diagram}
\end{subfigure}
   \caption{\sf (a) the IEEE 39 bus network from \cite{zimmerman2011matpower}; (b) simplified visualization of the network in (a), with the specific choice of supply and demand nodes used for simulations}
   \label{fig:ieee39}
\end{figure}
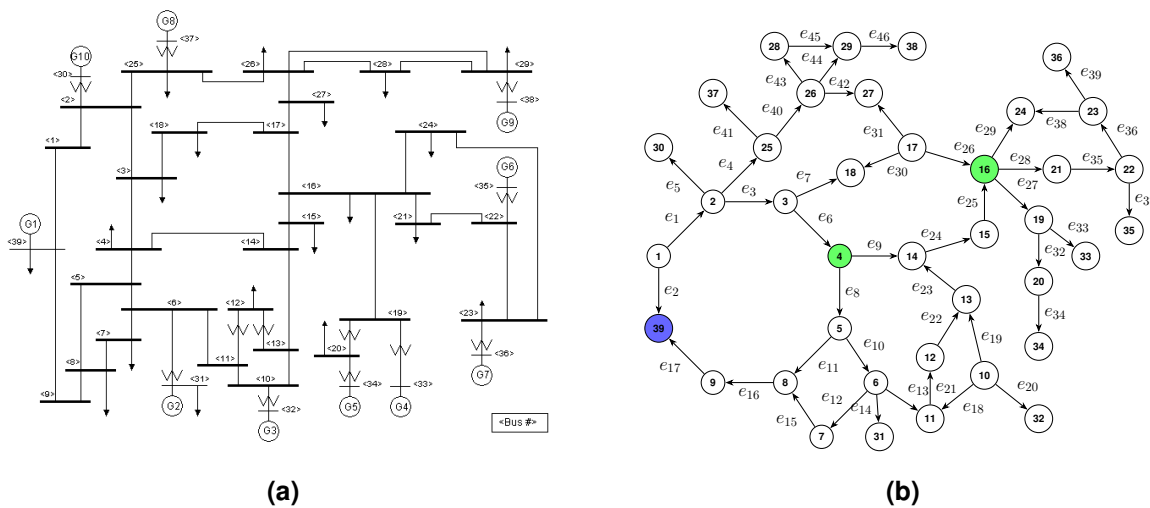

We conducted numerical experiments on the IEEE 39 bus system illustrated in Figure~\ref{fig:ieee39}. Node 39 is selected to be the only supply; nodes 4 and 16 are selected to be loads; all the other nodes are transmission nodes. This particular choice of supply-demand nodes is consistent with the fact that the actual supply and demand on these nodes, as reported in \cite{zimmerman2011matpower}, have relatively large values. $p^0$ was chosen to be proportional to the actual values reported in \cite{zimmerman2011matpower} for nodes $4$, $16$ and $39$: \( p_{4}^{0} = p_{16}^{0} = -5 \), \( p_{39}^{0} = 10\) and \( p_{i}^{0}=0 \) for all \( i \notin \{4, 16, 39\} \). Link susceptances \( w \) are from \cite{zimmerman2011matpower}. The link capacities were selected as follows: \( c_{8} = 0.5 \), \( c_{9} = 1 \), \( c_{i} = 2.5 \) for \( i\in \{13, 21, 22, 23\} \); \( c_{i} = 3.0 \) for \( i \in \{3, 28, 29, 35, 36, 38 \} \); \(  c_{i} = 3.5 \) for \( i\in \{16, 17\} \); \( c_{i} = 4.0 \) for \( i\in \{7, 26, 30\} \); \( c_{i} = 4.5 \) for \( i\in \{1, 2, 4, 24, 25, 31, 39, 40, 42\} \) and \( c_{i} = 2.0 \) for other links.

\begin{table*}[hbt]
\caption{Optimal residual load under \eqref{eq:control-formulation-new} and under the projection-based approximations in \eqref{eq:approx-control-sims}}
\label{tab:performance}
\begin{center}
\begin{tabular}{c|ccccccccccc|c}
\noalign{\hrule height 1pt}
%\multicolumn{2}{c}{\( N \)}   &	1	&	2	&	3	&	4	&	5	\\
%\midrule
%\parbox[t]{0.1mm}{\multirow{11}{*}{\rotatebox[origin=c]{90}{1-D approximation  \( \eta = \) }}} 
%&	0.0	&	3.716	&	9.860	&	10.000	&	10.000	&	10.000	\\
%&	0.1	&	3.502	&	9.806	&	11.112	&	11.112	&	11.112	\\
%&	0.2	&	3.310	&	9.750	&	11.090	&	11.090	&	11.090	\\
%&	0.3	&	3.140	&	9.696	&	11.028	&	11.028	&	11.028	\\
%&	0.4	&	2.984	&	8.974	&	10.000	&	10.000	&	10.000	\\
%&	0.5	&	2.844	&	9.000	&	9.000	&	9.000	&	9.000	\\
%&	0.6	&	2.718	&	7.334	&	7.334	&	7.334	&	7.334	\\
%&	0.7	&	2.600	&	6.742	&	6.742	&	6.742	&	6.742	\\
%&	0.8	&	2.494	&	4.578	&	5.000	&	5.000	&	5.000	\\
%&	0.9	&	2.396	&	4.078	&	4.444	&	4.444	&	4.444	\\
%&	1.0	&	2.304	&	4.000	&	4.000	&	4.000	&	4.000	\\
%\hline
%\multicolumn{2}{c}{Optimal control} &	3.716	&	9.860	&	11.150	&	11.150	&	11.150	\\ 
 \backslashbox{\( N \)}{\( \eta \)} &	 0	&	0.1	&	0.2	&	0.3	&	0.4	&	0.5	&	0.6	&	0.7	&	0.8	&	0.9	&	1	& Optimal \\
\noalign{\hrule height 1pt}
1	&	3.716	&	3.502	&	3.310	&	3.140	&	2.984	&	2.844	&	2.718	&	2.600	&	2.494	&	2.396	&	2.304	&	3.716	\\
2	&	9.860	&	9.806	&	9.750	&	9.696	&	8.974	&	9.000	&	7.334	&	6.742	&	4.578	&	4.078	&	4.000	&	9.860	\\
3	&	10.000	&	11.112	&	11.090	&	11.028	&	10.000	&	9.000	&	7.334	&	6.742	&	5.000	&	4.444	&	4.000	&	11.150	\\
4	&	10.000	&	11.112	&	11.090	&	11.028	&	10.000	&	9.000	&	7.334	&	6.742	&	5.000	&	4.444	&	4.000	&	11.150	\\
5	&	10.000	&	11.112	&	11.090	&	11.028	&	10.000	&	9.000	&	7.334	&	6.742	&	5.000	&	4.444	&	4.000	&	11.150	\\
\noalign{\hrule height 1pt}
\end{tabular}
\end{center}
\end{table*}%

For the above network parameters, $(\mc E^0,p^0)$ is infeasible. Furthermore, under no load shedding control action, i.e., $u^t \equiv p^0$ for all $t$, the only supply node 39 would get disconnected at $t=1$ from the load nodes 4 and 16. However, using the control formulation of this paper, such a scenario can be prevented while minimizing the amount of load to be shed. Table~\ref{tab:performance} (last column) shows the values of residual load, i.e., the optimal solution to \eqref{eq:control-formulation-new}, computed by the techniques in Section~\ref{sec:algorithm-section}, for different control horizons $N$.  The residual load is expectedly nondecreasing with \( N \). This confirms that multi-round control does lead to increase in the residual load, or equivalently decrease in cumulative load shed, in comparison to the single round (\(N=1 \)) control underlying power re-dispatch. However, there are no gains in residual load beyond \( N\ge 3 \). This is because the network in Figure~\ref{fig:ieee39} contains a very few cycles; and once the network becomes a tree, the optimal load shedding action is to ensure feasible of all the links in this case.

%The network is infeasible for the given parameters and according to the cascading dynamics \eqref{eq:cascade-dynamics} the only generator node 39 would get disconnected immediately from both loads (nodes 4 and 16) if no load shedding action was taken. As shown in the paper, by carefully shedding load at appropriate time, the cascading failure could be terminated while a considerable portion of load could be maintained at the expense of link failures. 

%Table \ref{tab:performance} shows the residual load (i.e., terminal reward) under different load shedding policies. These values are reported for different number of allowed stages $N$, with each column corresponding to each number of stages. The second row is obtained using the algorithm proposed in Section~\ref{sec:tree-search}. It shows the the residual load from an optimal controls over all possible load shedding policies. 
Table~\ref{tab:performance} also shows optimal residual load within the class of control policies obtained by projection onto a one-dimensional space, as described in Section~\ref{sec:approximation-method}. Specifically, we chose 
\begin{equation}
\label{eq:approx-control-sims}
B= \{1\}, \qquad \Phi_{1} \propto \eta \bar{p}^{1} + (1-\eta) \bar{p}^{2}, \, \, \eta \in [0, 1]
\end{equation}
where \( \bar{p}^{1} \in \real^{39} \) has 1 and -1 on node 39 and 4, respectively, and 0 elsewhere; \( \bar{p}^{2} \in \real^{39} \) has 1 and -1 on node 39 and 16, respectively, and 0 elsewhere. 
Recalling Example~\ref{eg:proportional-control}, it is easy to see that the set of proportional control policies~\cite[Section 6.1.1]{Bienstock:16} corresponds to \( \eta = 0.5 \). 
Table~\ref{tab:performance} contains values for optimal residual load under such an approximation for different values of $\eta$ and $N$. These values show that, similar to the optimal control actions, for every $\eta$, the optimal residual load is nondecreasing in $N$ and stays constant for $N \geq 3$.  
While there is no general monotone relationship in $\eta$ (uniformly for all $N$), the best control actions for $N \geq 3$ correspond to $\eta=0.1$. The control actions corresponding to $\eta=0.1$ perform uniformly better than the proportional control policy $(\eta=0.5)$ which requires comparable computational cost, and give fairly similar performance as the optimal control actions which are obtained under considerable computational costs (Section~\ref{sec:algorithm-section}).  
%The other eleven rows show the maximal residual load obtained from the one dimensional (1-D) approximation algorithm proposed in Section~\ref{sec:approximation-method}.  Each row corresponds to a projection onto a different subspace. As explained in Section~\ref{sec:approximation-method}, the control space from the 1-D approximation is determined by \( \Phi_{B} \). In the simulation, we chose \( B= \{1\} \) and further parameterize \( \Phi_{1} \) as follows: 
%\begin{equation*}
%\Phi_{1} = \eta \bar{p}^{1} + (1-\eta) \bar{p}^{2} \qquad \eta \in [0, 1]
%\end{equation*}
%where \( \bar{p}^{1} \in \real^{39} \) has 1 and -1 on node 39 and 4, respectively, and 0 elsewhere; \( \bar{p}^{2} \in \real^{39} \) has 1 and -1 on node 39 and 16, respectively, and 0 elsewhere. Recall Example~\ref{eg:proportional-control}, it is easy to see that the 1-D approximation with \( \eta = 0.5 \) gives proportional control policies. 
%
%Furthermore, Table~\ref{tab:performance} shows that the performance for 1-D approximation varies considerably over different projections. The best from them, which is not necessary proportional control, can perform closely to the optimal load shedding, as demonstrated by the 1-D approximation for \( \eta = 0.1 \). Because the aggregated tree search algorithm based on hyperplane arrangement can be slow, the simulation results support the usability of the more efficient low dimensional approximation algorithms.  

\section{Conclusions and Future Work}
\label{sec:conclusions}
Cascading failure in physical networks has attracted great interest, and yet formal approaches for its control are relatively very few. This paper builds upon an existing formulation for optimal control of cascading failure in power networks under DC approximation, and provides approaches for computing optimal control in this setting. The decomposition paradigm and connections to computational geometry underlying our approaches suggest several avenues for future work.

%In this paper, we provided novel computational approaches for optimal control of a specific discrete time dynamics corresponding to cascading failure in DC power networks under controlled load shedding. These approaches make connections with and build upon tools from network optimization and combinatorial geometry. While we provide specific results on optimal and approximate control synthesis within these frameworks, the extensive literature in these other disciplines opens up several avenues for future work. 

%Beyond the specific results on optimal and approximate control synthesis that we provide in this paper, these 

%considered a dynamic programming formulation for optimal load shedding control during cascading failures in power networks, and proposed a branch and bound algorithm for its solution. 
%The incremental approach to compute flow redistribution upon link failure and exploiting monotonicity properties to reduce the set of feasible control set to be explored helps to reduce computational cost of the implementation of the algorithm. 

%The two computational frameworks proposed in this paper open up several avenues for future work. 

As an initial step towards generalization, we plan to adapt the decomposition approach to be able to compute a feasible control action for non-tree reducible networks. We also plan to explore general optimal control formulations, beyond cascading failure and networked settings, to which the approaches developed in this paper are applicable. 
%We believe that this should be possible in cascading failure of networked systems, where the network dynamics evolves at a much faster time scale in comparison to cascading failure dynamics. 
%  
%
%\ksmargin{do we have specific conjectures here?}
%\qbcomment{The cascading dynamic model can be readily extended to other large scale dynamical systems whose constituent components are allowed to fail and underlying physical dynamics are on a considerable faster time scale than the cascades of failure.}
%We plan to extend the network decomposition approach beyond the specific settings of this paper. \kscomment{Specifically, we are interested in conditions for non tree-reducible networks under which there exists a decomposition into coupled local problems which can be solved to optimality asymptotically through distributed optimization techniques.}   
%
%optimality of the network decomposition approach to \emph{loopy} networks. A good starting point is a simple extension of the algorithm that yields, at the minimum, a feasible control action for an arbitrary network. We plan to study if the semi-analytic solution methodology via an equivalent transformation into a finite system of easily-solvable equations extends beyond the single time horizon case. 
%
%With regards to the search-based approach, we plan to explore approximation techniques beyond the projection approach in this paper. \kscomment{In particular, we are interested in techniques for pruning the search tree while giving provable approximation guarantees. 
In particular, we plan to investigate connections between our approach of computing optimal control using (equivalent) finite representations and the recent work on symbolic optimal control, e.g., see \cite{Rungger.Reissig.CDC17}.
%
%paradigm in this paper is reminiscent of the literature on finite abstractions for hybrid systems, which has been done primarily for stability analysis.  We plan to explore connections between the setting of this paper, possibly extended to continuous time domain, and recent work on symbolic optimal control, e.g., see \cite{Rungger.Reissig.CDC17}.

% reference
\bibliographystyle{ieeetr}
\bibliography{bib/ksmain,bib/savla,bib/ref}

\begin{thebibliography}{10}

\bibitem{Cohen.Erez.ea:00}
R.~Cohen, K.~Erez, D.~Ben-Avraham, and S.~Havlin, ``Resilience of the internet
  to random breakdowns,'' {\em Physical review letters}, vol.~85, no.~21,
  p.~4626, 2000.

\bibitem{Watts:01}
D.~J. Watts, ``A simple model of global cascades on random networks,'' {\em
  PNAS}, vol.~99, no.~9, pp.~5766--5771, 2002.

\bibitem{Motter:2002fk}
A.~Motter, ``Cascade-based attacks on complex networks,'' {\em Phys. Rev. E;
  Physical Review E}, vol.~66, no.~6, 2002.

\bibitem{Crucitti.Latora.ea:04}
P.~Crucitti, V.~Latora, and M.~Marchiori, ``Model for cascading failures in
  complex networks,'' {\em Physical Review E}, vol.~69, no.~4, 2004.

\bibitem{Barrat.Barthelemy.ea:08}
A.~Barrat, M.~Barthelemy, and A.~Vespignani, {\em Dynamical Processes on
  Complex Networks}.
\newblock Cambridge University Press, 2008.

\bibitem{Bienstock:11}
D.~Bienstock, ``Optimal control of cascading power grid failures,'' in {\em
  Decision and Control and European Control Conference (CDC-ECC), 2011 50th
  IEEE Conference on}, pp.~2166--2173, IEEE, 2011.

\bibitem{Savla.Como.ea.TNSse14}
K.~Savla, G.~Como, and M.~A. Dahleh, ``Robust network routing under cascading
  failures,'' {\em IEEE Transactions on Network Science and Engineering},
  vol.~1, no.~1, pp.~53--66, 2014.

\bibitem{bienstock2010nk}
D.~Bienstock and A.~Verma, ``The n-k problem in power grids: New models,
  formulations, and numerical experiments,'' {\em SIAM Journal on
  Optimization}, vol.~20, no.~5, pp.~2352--2380, 2010.

\bibitem{bernstein2014power}
A.~Bernstein, D.~Bienstock, D.~Hay, M.~Uzunoglu, and G.~Zussman, ``Power grid
  vulnerability to geographically correlated failures: Analysis and control
  implications,'' in {\em INFOCOM, 2014 Proceedings IEEE}, pp.~2634--2642,
  IEEE, 2014.

\bibitem{Bienstock:16}
D.~Bienstock, {\em Electrical Transmission System Cascades and Vulnerability:
  An Operations Research Viewpoint}, vol.~22.
\newblock SIAM, 2016.

\bibitem{braess1968paradoxon}
D.~Braess, ``{\"U}ber ein paradoxon aus der verkehrsplanung,'' {\em
  Unternehmensforschung}, vol.~12, no.~1, pp.~258--268, 1968.

\bibitem{Ba.Savla.CDC16}
Q.~Ba and K.~Savla, ``A dynamic programming approach to optimal load shedding
  control of cascading failure in {DC} power networks,'' in {\em IEEE
  Conference on Decision and Control}, (Las Vegas, NV), 2016.

\bibitem{Lai.Low:Allerton13}
C.~Lai and S.~H. Low, ``The redistribution of power flow in cascading
  failures,'' in {\em 51st Annual Allerton Conference on Communication,
  Control, and Computing}, pp.~1037--1044, 2013.

\bibitem{guo2017monotonicity}
L.~Guo, C.~Liang, and S.~H. Low, ``Monotonicity properties and spectral
  characterization of power redistribution in cascading failures,'' {\em ACM
  SIGMETRICS Performance Evaluation Review}, vol.~45, no.~2, pp.~103--106,
  2017.

\bibitem{salmeron2004analysis}
J.~Salmeron, K.~Wood, and R.~Baldick, ``Analysis of electric grid security
  under terrorist threat,'' {\em IEEE Transactions on power systems}, vol.~19,
  no.~2, pp.~905--912, 2004.

\bibitem{pinar2010optimization}
A.~Pinar, J.~Meza, V.~Donde, and B.~Lesieutre, ``Optimization strategies for
  the vulnerability analysis of the electric power grid,'' {\em SIAM Journal on
  Optimization}, vol.~20, no.~4, pp.~1786--1810, 2010.

\bibitem{edelsbrunner1987algorithms}
H.~Edelsbrunner, ``Algorithms in combinatorial geometry, volume 10 of {EATCS}
  monographs on theoretical computer science,'' 1987.

\bibitem{toth2004handbook}
C.~D. Toth, J.~O'Rourke, and J.~E. Goodman, {\em Handbook of discrete and
  computational geometry}.
\newblock CRC press, 2004.

\bibitem{ziegler2012lectures}
G.~M. Ziegler, {\em Lectures on polytopes}, vol.~152.
\newblock Springer Science \& Business Media, 2012.

\bibitem{grunbaum1967convex}
B.~Gr{\"u}nbaum, V.~Klee, M.~A. Perles, and G.~C. Shephard, {\em Convex
  polytopes}, vol.~16.
\newblock Springer, 1967.

\bibitem{halperin1995arrangements}
D.~Halperin and M.~Sharir, ``Arrangements and their applications in robotics:
  recent developments,'' in {\em Proceedings of the Workshop on Algorithmic
  Foundations of Robotics}, pp.~495--511, AK Peters, Ltd., 1995.

\bibitem{agarwal2013resilience}
P.~K. Agarwal, A.~Efrat, S.~K. Ganjugunte, D.~Hay, S.~Sankararaman, and
  G.~Zussman, ``The resilience of {WDM} networks to probabilistic geographical
  failures,'' {\em IEEE/ACM Transactions on Networking}, vol.~21, no.~5,
  pp.~1525--1538, 2013.

\bibitem{Ba.Savla:TCNS16}
Q.~Ba and K.~Savla, ``Robustness of {DC} networks under controllable link
  weights,'' {\em IEEE Transactions on Control of Network Systems}, 2017.
\newblock In Press. Available at
  \href{http://arxiv.org/abs/1609.02179}{\nolinkurl
  http://arxiv.org/abs/1609.02179}.

\bibitem{russell2009artificial}
S.~J. Russell and P.~Norvig, ``Artificial intelligence: a modern approach (3rd
  edition),'' 2009.

\bibitem{chen2005cascading}
J.~Chen, J.~S. Thorp, and I.~Dobson, ``Cascading dynamics and mitigation
  assessment in power system disturbances via a hidden failure model,'' {\em
  International Journal of Electrical Power and Energy Systems}, vol.~27,
  no.~4, pp.~318--326, 2005.

\bibitem{Boyd.Parikh.ea:FTML11}
S.~Boyd, N.~Parikh, E.~Chu, B.~Peleato, and J.~Eckstein, ``Distributed
  optimization and statistical learning via the alternating direction method of
  multipliers,'' {\em Foundations and Trends in Machine Learning}, vol.~3,
  no.~1, pp.~1--122, 2011.

\bibitem{Soltan.Loh.ea:CONES17}
S.~Soltan, A.~Loh, and G.~Zussman, ``Analyzing and quantifying the effect of
  $k$-line failures in power grids,'' {\em IEEE Transactions on Control of
  Network Systems}, 2017.
\newblock In press.

\bibitem{ba2018elements}
Q.~Ba, {\em Elements of Robustness and Optimal Control for Infrastructure
  Networks}.
\newblock PhD thesis, UNIVERSITY OF SOUTHERN CALIFORNIA, 2018.
\newblock Available at
  \href{http://www-bcf.usc.edu/~ksavla/papers/Qin-thesis-main.pdf}{\nolinkurl
  http://www-bcf.usc.edu/$\sim$ksavla/papers/Qin-thesis-main.pdf}.

\bibitem{de2008computational}
M.~De~Berg, O.~Cheong, M.~Van~Kreveld, and M.~Overmars, {\em Computational
  Geometry: Introduction}.
\newblock Springer, 2008.

\bibitem{edelsbrunner1986constructing}
H.~Edelsbrunner, J.~O'Rourke, and R.~Seidel, ``Constructing arrangements of
  lines and hyperplanes with applications,'' {\em SIAM Journal on Computing},
  vol.~15, no.~2, pp.~341--363, 1986.

\bibitem{burger1996polytope}
T.~Burger, P.~Gritzmann, and V.~Klee, ``Polytope projection and projection
  polytopes,'' {\em The American mathematical monthly}, vol.~103, no.~9,
  pp.~742--755, 1996.

\bibitem{zimmerman2011matpower}
R.~D. Zimmerman, C.~E. Murillo-S{\'a}nchez, and R.~J. Thomas, ``Matpower:
  Steady-state operations, planning, and analysis tools for power systems
  research and education,'' {\em Power Systems, IEEE Transactions on}, vol.~26,
  no.~1, pp.~12--19, 2011.

\bibitem{Rungger.Reissig.CDC17}
M.~Rungger and G.~Reissig, ``Arbitrarily precise abstractions for optimal
  controller synthesis,'' in {\em IEEE Conf. on Decision and Control},
  (Melbourne, Australia), 2017.

\bibitem{sanyal2010construction}
R.~Sanyal and G.~M. Ziegler, ``Construction and analysis of projected deformed
  products,'' {\em Discrete \& Computational Geometry}, vol.~43, no.~2,
  pp.~412--435, 2010.

\bibitem{Boyd.Vandenberghe:04}
S.~Boyd and L.~Vandenberghe, {\em Convex optimization}.
\newblock Cambridge university press, 2004.

\end{thebibliography}

\appendix[Proofs]

\subsection{Proof of Proposition~\ref{prop:para-net-opt-control}}
\label{proof:para-net-opt-control}
We consider the following cases: 
\begin{enumerate}[1)]
\item \( p_{1}^{0} \le \pnetval_{| \mathcal{E}|} \). Remark~\ref{rem:value-V} implies that \( ( \mathcal{E}, p^{0}) \in \mathcal{S} \). The optimal control for every \( N \ge 1 \) would be shedding no load. In this case,  by definition \( N_{j}(p^{0}) = 1 \) for all \( j\in [\cmax] \). Every \( N \ge 1 \) satisfies \(  N_{1}(p^{0})  \le  N < N_{0}(p^{0}) \). Hence, \( u^{*, t} = \min\{\pnetval_{\clink_{1}}, p_{1}^{0}\} [ 1 \; - 1] \) for \( t \ge 0 \). Since \( p^{0}_{1} \le R_{| \mathcal{E}|} \le R_{o_{1}} \), \( u^{*} \) is optimal. 
\item  \( \pnetval_{\clink_{k+1}} < p_{1}^{0} \le \pnetval_{ \clink_{k}} \) for some \( 0\le k \le \cmax -1 \). By definition, \( N_{1}(p^{0}) = N_{2}(p^{0}) = \ldots = N_{k}(p^{0}) \) and \( N_{k+1}(p^{0}) \ge \ldots \ge N_{\cmax-1}(p^{0}) \ge 2 > N_{\cmax}(p^{0}) =1 \). We have the following cases. 
\begin{enumerate}[ a)]
\item \( N=1 \). Since \( N_{\cmax}(p^{0}) \le 1 \le N_{\cmax-1}(p^{0}) \), \( u^{*} = \min\{\pnetval_{\clink_{\cmax}}, p_{1}^{0}\} [1 \; -1] =  \pnetval_{| \mathcal{E}|}[1 \; -1]  \) is optimal, where the second equality follows from \( p^{0}_{1} > \pnetval_{\clink_{k}+1} \ge \pnetval_{\clink_{\cmax}} = \pnetval_{| \mathcal{E}|}  \). 
\item \( N \ge N_{k} \). In this case, \( N_{1}(p^{0}) \le N < N_{0}(p^{0}) \). Since \( p^{0}_{1} \le \pnetval_{\clink_{k} } \le \pnetval_{\clink_{1}} \),  \( u^{t, *} = p^{0} \) for all \( t \). \( u^{*} \) is optimal if feasible. The latter is a straightforward result from the definition of \( N_{k}(p^{0}) \).
\item \( 2 \le N < N_{k} \). In this case, \( N_{j}(p^{0}) \le N < N_{j-1}(p^{0}) \) for some \( k +1 \le j \le \cmax-1 \). Therefore, \( p^{0}_{1} > \pnetval_{\clink_{k+1} } \ge \pnetval_{\clink_{j}} \). \( u^{*, t} = p^{0} \) for \( 0\le t < N_{j}(p^{0}) -2  \) and \( u^{t, *} = \pnetval_{\clink_{j}} [1\; -1] \) for \( N_{j}(p^{0})-2 \le t \le N-1 \). 
We first show that \( u^{*} \in \mathcal{D}( \mathcal{E}, p^{0}, N) \). Let \( ( \mathcal{E}^{0}, \ldots, \mathcal{E}^{N}) \) be the topology sequence under \( u^{*} \). It is straightforward that \( ( \mathcal{E}^{0}, \ldots, \mathcal{E}^{N_{j}(p^{0}) -2}) = ( \mathcal{E}_{\mathrm{un}}^{0}, \ldots, \mathcal{E}_{\mathrm{un}}^{N_{j}(p^{0}) -2}) \). By definition of \( N_{j}(p^{0}) \), \( [\clink_{j}] \subset \mathcal{E}_{\mathrm{un}}^{N_{j}(p^{0}) -2} = \mathcal{E}^{N_{j}(p^{0}) -2} \). In addition, Remark~\ref{rem:value-V} implies \( ( [\clink_{j}], \pnetval_{\clink_{j}} [1\; -1]) \in \mathcal{S} \) and, plus the definition of \( \clink_{j} \) futher implies that \( ([l], p^{0}) \not\in \mathcal{S} \) for all \( l> \clink_{j} \). Therefore, \( \mathcal{E}^{t} = [\clink_{j}] \) for all \( t \ge N_{j}(p^{0}) -1 \) and \( u^{*} \in \mathcal{D}( \mathcal{E}, p^{0}, N) \). 
We then show optimality of \( u^{*} \) through contradiction. Suppose there exists a control \( \tilde{u} \in \mathcal{D}( \mathcal{E}, p^{0}, N) \) such that \( \pnetval_{\clink_{j}} < \tilde{u}^{N-1}_{1} \le p^{0}_{1} \). Let \( (\tilde{ \mathcal{E}}^{0}, \ldots, \tilde{ \mathcal{E}}^{N-1}) \) be the topology sequence under control \( \tilde{u} \). Remark~\ref{rem:monotonicity} implies that \( \mathcal{E}^{N-1}_{\mathrm{un}} \subseteq \tilde{ \mathcal{E}}^{N-1} \). At the same time, since \( N < N_{j-1}(p^{0}) \), \( \mathcal{E}^{N-1}_{\mathrm{un}} \supseteq \mathcal{E}_{\mathrm{un}}^{N_{j-1}(p^{0}) -2} \supset [\clink_{j-1}]\). Note that the last inclusion is strict. Therefore, \( [\clink_{j-1}] \subset \tilde{ \mathcal{E}}^{N-1} \). Remark~\ref{rem:value-V}, combined with the definition of \( \clink_{j} \) and the assumption that \( \tilde{u}^{N-1}_{1} > \pnetval_{\clink_{j}} \), implies that \( (\tilde{ \mathcal{E}}^{N-1}, \tilde{u}^{N-1}) \not \in \mathcal{S} \). This contradicts with \( \tilde{u} \) being feasible. 
\end{enumerate}
\end{enumerate}
\subsection{Proof of Lemma~\ref{lem:meta-star-net}}
\label{proof:meta-star-net}
First of all, it is clear that \eqref{opt:lp-meta-star} is feasible only for \( z \in [\trans{\onebf}q^l, \trans{\onebf}q^u]  \). Secondly, since \( \starobjin_{j} \) is concave and \( X_{j} \) is convex, \eqref{opt:lp-meta-star} is convex and strong duality holds. Therefore, it is sufficient to consider the dual problem in order to solve \eqref{opt:lp-meta-star}. Let $\mu \in \real$ be the Lagrange multiplier associated with the constraint \( z = \trans{\onebf} x \). The dual function is then given by: 
\begin{align*}
\phi( \mu) &= - \mu z + \max_{q^{l} \le x \le q^{u}}  \sum_{j=1}^{n} \left( \roofun_{\tau_{j}}(x_{j}) + \mu x_{j} \right) \\
&=  - \mu z + \sum_{j= 1}^{n} \;\,  \max_{ q^{l}_{j} \le x_{j} \le q^{u}_{j}} \left( \roofun_{\tau_{j}}(x_{j}) + \mu x_{j} \right)\\
&= \left( \sum_{j=1}^{n} x^{*}_{j}( \mu) -z \right) \mu  + \sum_{j=1}^{n} \roofun_{j} ( x^{*}_{j}( \mu))  
\end{align*}
where \( x^{*}_{j}( \mu) \in \mc X_{j}^{*}(\mu) := \argmax_{ q^{l}_{v} \le x_{j} \le q^{u}_{v}} \left( \roofun_{\tau_{j}}(x_{j}) + \mu x_{j} \right) \), for all $j \in [n] $.

For \( \mu \in [-1,1] \), \( \roofun_{\tau_{j}}(x_{j}) + \mu x_{j} \) is piecewise affine: nondecreasing with slope \( (1+ \mu) \) over \( (-\infty, \tau_{j}^{1} ] \) and nonincreasing with slope \( (\mu-1) \) over \( (\tau_{j}^{1}, +\infty)  \). Therefore, \( \tau_{j}^{1} \in \mc X_{j}^{*}(\mu) \) for all \( \mu \in [-1, 1] \) and \( j\in [n] \). This implies that \( \phi( \mu) \) is affine over \( [-1, 1] \), for every $z$. In particular, \( \mathcal{X}_{j}^{*}(-1) = [q^{l}_{j}, \tau_{j}^{1}] \) and \( \mathcal{X}_{j}^{*}(1) = [\tau_{j}^{1}, q^{u}_{j}] \) for all \( j\in [n] \).
For \( \mu > 1 \), \( \roofun_{\tau_{j}}(x_{j}) + \mu x_{j} \) is strictly increasing, and hence \( \mc X_{j}^{*}(\mu) = \{q^{u}_{j}\} \). Since \( z   \le \trans{\onebf} q^{u} \), $\phi(\mu)$ is affine and non-decreasing over $(1,+\infty)$. Similarly, by considering  $\mu<-1$, we have \(  \mc X_{j}^{*}(\mu)  =\{ q^{l}_{j} \} \) and $\phi(\mu)$ is affine and non-increasing over $(-\infty,-1)$. Collecting these facts gives that, for every $z \in [\trans{\onebf}q^l, \trans{\onebf}q^u] $, the dual function \( \phi( \mu) \) is convex and piecewise affine with possible break points at \( \mu=-1 \) and \( \mu = 1 \). Therefore, \(\starobjout(z)  = \min_{\mu \in \real} \phi(\mu) = \min\{ \phi(-1), \phi(1) \} \). 

As \( \tau_{j}^{1} \in \mathcal{X}_{j}^{*}(\mu) \) for \( \mu \in [-1, 1] \), we have \( \phi(-1) = z - \trans{\onebf} \tau^{1} +  \trans{\onebf} \tau^{2}  \) and \(\phi(1) =  - z + \trans{\onebf} \tau^{1} +  \trans{\onebf} \tau^{2} \).  \( \phi(-1) \le \phi(1) \) for \( z \in [ \trans{\onebf} q^{l}, \trans{\onebf} \tau^{1}  ] \), and \( \phi(1) \le \phi(-1) \) for \( z \in [ \trans{\onebf} \tau^{1},  \trans{\onebf} q^{u} ] \). Therefore, 
\begin{equation*}
\starobjout(z) = \left\{
\begin{array}{ll}
z - \trans{\onebf} \tau^{2} +  \trans{\onebf} \tau^{2}  \quad & \trans{\onebf} q^{l} \le z \le \trans{\onebf} \tau^{1} \\
- z + \trans{\onebf} \tau^{1} +  \trans{\onebf} \tau^{2}  &  \trans{\onebf} \tau^{1} \le z \le \trans{\onebf} q^{u} \\
\end{array}\right.
\end{equation*}
Comparing with \eqref{eq:l-fun} establishes (i). (ii) follows from the fact that, for a given optimal dual solution \( \mu^{*} \), \( x_{j}^{*}(\mu^{*}) \in \mc X_{j}^{*}(\mu^{*}) \), \( j\in [n]\), is an optimal primal solution if and only if the constraint \( z = \trans{\onebf} x^{*} \) is satisfied.

\subsection{Proof of Proposition~\ref{prop:star-net-invariance}}
\label{proof:star-net-invariance-condition}
We first show that the conditions are sufficient. We start by showing that the output functions from \( \starop \{ (\starobjin_{j}, X_{j}) \}_{j \in [n]}  \) and \( \starop \{ (\starobjin_{j}, \conv{X_{j}} ) \}_{j \in [n]} \) have the same domain under the given conditions. Since \( \tau_{j}^{1} \in X_{j} \) for all \( j\in [n] \) and \( \sum_{j=1}^{n} X_{j} \cap (-\infty, \tau_{j}^{1}] \) is connected, \( \sum_{j=1}^{n} X_{j} \cap (-\infty, \tau_{j}^{1}] = [\sum_{j=1}^{n} \min X_{j}, \sum_{j=1}^{n} \tau_{j}^{1} ] \). Similarly, \( \sum_{j=1}^{n} X_{j} \cap [\tau_{j}^{1}, \infty) = [\sum_{j=1}^{n} \tau_{j}^{1}, \sum_{j=1}^{n} \max X_{j}] \). Therefore, \(  \sum_{j=1}^{n} X_{j} = [\sum_{j=1}^{n} \min X_{j}, \sum_{j=1}^{n} \max X_{j}] = \sum_{j=1}^{n} \conv{X_{j}} \). 

It is straightforward that \( \starop \{ (\starobjin_{j}, X_{j}) \}_{j \in [n]} \le \starop \{ (\starobjin_{j}, \conv{X_{j}} ) \}_{j \in [n]}  \). Hence it is sufficient to prove the other direction. By Lemma~\ref{lem:meta-star-net}, \( \starop \{ (\starobjin_{j}, \conv{X_{j}} ) \}_{j \in [n]} \) is a \( \roofun \) function with top point \( \tilde{\tau} := (\sum_{j=1}^{n} \tau_{j}^{1}, \sum_{j=1}^{n} \tau_{j}^{2}) \). We need to show that \( \left( \starop \{ (\starobjin_{j}, X_{j}) \}_{j \in [n]}\right) (z) \ge \roofun_{\tilde{\tau}}(z) \) for all \( z \in [\sum_{j=1}^{n} \min X_{j}, \sum_{j=1}^{n} \max X_{j}]  \). 

We first consider \( z\in  [ \sum_{j=1}^{n} \min X_{j}, \sum_{j=1}^{n} \tau_{j}^{1}]  \). Let \( X_{j} \) contain \( m_{j} \) pieces of intervals. Plus that \( \tau_{j}^{1} \in X_{j} \) separates an interval of \( X_{j} \) into two pieces, with possibly one of the two containing the single point \( \tau_{j} \), there are \( m_{j}+1 \) pieces in total. Without loss of generality, label the \( m_{j}+1 \) intervals \( X_{j}^{k} \) in increasing order, that is, \( X_{j}^{k} \) such that \( \max X_{j}^{k-1} \le \min X_{j}^{k} \) for all \( k\in [m_{j}+1] \). In particular, let \( l_{j} \in [m_{j}+1] \) be such that \(  \tau_{j}^{1} = \max X_{j}^{ l_{j} } = \min X_{j}^{l_{j}+1} \) for all \( j\in [n] \). Furthermore, let \( \tilde{\comset}:= \Pi_{j=1}^{n} [m_{j}+1] \) and \( \tilde{\comset}_{\le \sigma'}: = \setdef{\sigma \in \tilde{\comset}}{ \sigma \le \sigma' } \) for all \( \sigma' \in \tilde{\comset} \). The notation \(  \tilde{\comset}_{> \sigma'} \) has similar meaning. With these notations, \( X_{j}\cap (-\infty, \tau_{j}^{1}]  = \cup_{\sigma_{j} \le l_{j}} X_{j}^{\sigma_{j}} \). Then \(  \cup_{ \sigma \in \tilde{\comset}_{\le l} } \sum_{j=1}^{n} X_{j}^{\sigma_{j}} = \sum_{j=1}^{n} \cup_{\sigma_{j}\le l_{j}} X_{j}^{\sigma_{j}} = \sum_{j=1}^{n} X_{j}\cap (-\infty, \tau_{j}^{1}] =  [ \sum_{j=1}^{n} \min X_{j}, \sum_{j=1}^{n} \tau_{j}^{1}] \). It is then sufficient to prove that for all \( \sigma \in \tilde{\comset}_{\le l} \), \( \starop( \sigma )  = \roofun_{\tilde{\tau}} \) over the domain  \( \sum_{j=1}^{n} X_{j}^{\sigma_{j}} \). 

Pick arbitrary \( \sigma \in \tilde{ \comset}_{\le l} \), without loss of generality, let \([q_{j}^{l}, q_{j}^{u}] :=  X_{j}^{\sigma_{j}} \). Restricted in this domain, \( \starobjin_{j} \) is an affine function with slope \( 1 \). Lemma~\ref{lem:lx-extension} implies that \( ( q_{j}^{u}, q_{j}^{u} - \tau_{j}^{1} +\tau_{j}^{2} ) \) can be treated as the top point of \( \starobjin_{j} \) over domain \( [q_{j}^{l}, q_{j}^{u}] \). Lemma~\ref{lem:meta-star-net} then implies that \( \starop (\sigma)  \) is an affine function with top point \( ( \sum_{j=1}^{n} q_{j}^{u},  \sum_{j=1}^{n} q_{j}^{u} - \sum_{j=1}^{n} \tau_{j}^{1} +  \sum_{j=1}^{n} \tau_{j}^{2} ) \). It is straightforward to check that this top point lies on \( \roofun_{\tilde{\tau}} \). As a result, \( \starop (\sigma) = \roofun_{\tilde{\tau}} \) when evaluated in the domain \( \sum_{j=1}^{n} X_{j}^{\sigma_{j}} \). By symmetry, the same result can be shown for \( z \in  [ \sum_{j=1}^{n} \tau_{j}^{1}, \sum_{j=1}^{n} \max X_{j}] \), by considering \( \sigma \in \tilde{\comset}_{>l} \).

We now show that the conditions are necessary. Lemma~\ref{lem:meta-star-net} implies that the solution to \eqref{opt:lp-meta-star} corresponding to \( \left( \starop \{ (\starobjin_{j}, \conv{X_{j}} ) \}_{j \in [n]} \right)(z) \) with \( z = \sum_{j=1}^{n} \tau_{j}^{1} \) is unique and \( x_{j}^{*} = \tau_{j}^{1} \) for all \( j\in [n] \). In order for \( \starop \{ (\starobjin_{j}, X_{j}) \}_{j \in [n]} = \roofun_{\tilde{\tau}} \) to be true, it must be that \( \tau_{j}^{1} \in X_{j} \) for all \( j\in [n] \). 

We next prove connectedness of  \( \sum_{j=1}^{n} X_{j} \cap (-\infty, \tau_{j}^{1}] \) and \( \sum_{j=1}^{n} X_{j} \cap [\tau_{j}^{1}, \infty) \), given that \( \starop \{ (\starobjin_{j}, X_{j}) \}_{j \in [n]} = \roofun_{\tilde{\tau}} \). If we could show \( \max_{ \sigma \in \tilde{\comset}_{\le l} \cup \tilde{\comset}_{>l} } \starop(\sigma) =\roofun_{\tilde{\tau}} \), then the connectedness would follow from the fact that two equal functions must have identical domains. In order to show this equality, it is sufficient to show that for all \( \sigma \in \tilde{\comset}\setminus (\tilde{\comset}_{\le l} \cup \tilde{\comset}_{>l}) \), either  \( \starop(\sigma) <  \roofun_{\tilde{\tau}} \) or \(  \starop (\sigma)= \starop (\sigma')  \) for some \(  \sigma' \in \tilde{\comset}_{\le l} \cup \tilde{\comset}_{>l} \). Pick arbitrary \( \sigma \in \tilde{\comset}\setminus (\tilde{\comset}_{\le l} \cup \tilde{\comset}_{>l}) \), then both \( \{j: \sigma_{j} \le l_{j} \} \) and \( \{j: \sigma_{j} \ge l_{j}+1 \} \) are nonempty. Similar to the arguments in previous paragraphs for proving the sufficient condition,  let \([q_{j}^{l}, q_{j}^{u}] :=  X_{j}^{\sigma_{j}} \), then Lemma~\ref{lem:lx-extension} and Lemma~\ref{lem:meta-star-net} imply that \( \starop(\sigma) \) is a \( \roofun \) function with top point \( (\hat{\tau}_{1}, \hat{\tau}_{2}) \), where \(\hat{\tau}_{1} := \sum_{\{j: \sigma_{j} \le l_{j} \}} q_{j}^{u} + \sum_{\{j: \sigma_{j} \ge l_{j}+1 \}} q_{j}^{l}\) and \( \hat{\tau}_{2}:= \sum_{\{j: \sigma_{j} \le l_{j} \}} (q_{j}^{u} - \tau_{j}^{1}) + \sum_{\{j: \sigma_{j} \ge l_{j}+1 \}} ( \tau_{j}^{1} - q_{j}^{l} ) + \sum_{j=1}^{n} \tau_{j}^{2} \) and domain \( [\sum_{j=1}^{n} q_{j}^{l}, \sum_{j=1}^{n} q_{j}^{u}] \). Simple algebra gives: 
\begin{align*}
& \roofun_{\tilde{\tau}} (\hat{\tau}_{1})-\hat{\tau}_{2} =  \\
& \, \, \, 2 \min \left\{ \sum_{\{j: \sigma_{j} \le l_{j} \}} ( \tau_{j}^{1} -  q_{j}^{u}),  \sum_{\{j: \sigma_{j} \ge l_{j}+1 \}} ( q_{j}^{l} -\tau_{j}^{1} ) \right\} \ge 0 
\end{align*}
If both  \( \{j: \sigma_{j} \le l_{j}-1 \} \) and \( \{j: \sigma_{j} \ge l_{j}+2 \} \) are nonempty, then the above inequality is strict, implying \( (\hat{\tau}_{1}, \hat{\tau}_{2}) \) is below \( \roofun_{\tilde{\tau}} \) and hence \( \starop(\sigma) <  \roofun_{\tilde{\tau}} \). Otherwise, if either \( \{j: \sigma_{j} \le l_{j} \} =  \{j: \sigma_{j} = l_{j} \}  \) only, or \( \{j: \sigma_{j} \ge l_{j}+1 \} =  \{j: \sigma_{j} = l_{j} + 1 \}  \) only, or both, then the above inequality becomes equality and \( (\hat{\tau}_{1}, \hat{\tau}_{2}) \) lies on \( \roofun_{\tilde{\tau}} \). We consider the first case and show that \( \left( \starop(\sigma)\right) (z) < \roofun_{\tilde{\tau}} (z)\) for \( z\in [\sum_{j=1}^{n} q_{j}^{l}, \hat{\tau}_{1}) \) and \( \left( \starop(\sigma)\right) (z) =   \starop(\sigma')(z) \) for all \( z\in [\hat{\tau}_{1}, \sum_{j=1}^{n} q_{j}^{u}] \) and for some \( \sigma' \in  \tilde{\comset}_{\le l} \cup \tilde{\comset}_{>l} \). Similar results are true for the other two cases. 

In the case considered, \( q_{k}^{u} = \tau_{k}^{1} \) for all \( k \in \{j: \sigma_{j} = l_{j} \} \), \( q_{k}^{l} \ge \tau_{k}^{1} \) for all \( k \in \{j: \sigma_{j} \ge l_{j}+1 \} \), and the later inequality is strict for at least one \( k \). Then \( \hat{\tau}_{1} > \sum_{j=1}^{n} \tau_{j}^{1} = \tilde{\tau}_{1} \). Combining with Lemma~\ref{lem:meta-star-net}, and recalling that $\hat{\tau}$ lies on $\roofun_{\tilde{\tau}}$, we get that \( \left( \starop(\sigma)\right) (z) < \roofun_{\tilde{\tau}} (z) \) for \( z \in  [\sum_{j=1}^{n} q_{j}^{l}, \hat{\tau}_{1}) \) and \( \left( \starop(\sigma)\right) (z) = \roofun_{\tilde{\tau}} (z) \) for \( z\in [\hat{\tau}_{1}, \sum_{j=1}^{n} q_{j}^{u}] \). It remains to show the existence of \( \sigma' \in  \tilde{\comset}_{\le l} \cup \tilde{\comset}_{>l} \) such that \( \left( \starop(\sigma)\right) (z) = \left(\starop(\sigma')\right)(z) \) for \( z\in [\hat{\tau}_{1}, \sum_{j=1}^{n} q_{j}^{u}] \). Such a \( \sigma' \)  is constructed as follows: \( \sigma'_{k} = l_{j}+1 \) for all \( k\in   \{j: \sigma_{j} = l_{j} \} \) and \( \sigma'_{k} = \sigma_{k} \) for all \( k \in  \{j: \sigma_{j} \ge l_{j}+1 \} \). Such a \( \sigma' \in \tilde{\comset}_{>l} \) ensures that \( \starop(\sigma') \) and \( \starop(\sigma) \) have identical top points.

\subsection{Proof of Proposition~\ref{prop:top-facet-boundary-ridge}}
\label{proof:top-facet-boundary-ridge}
The following lemma, which is also referred to as the geometric version of Farkas Lemma \cite{sanyal2010construction}, is used below in the proof of Proposition~\ref{prop:top-facet-boundary-ridge}.

\begin{lemma}[{\cite[Sect. 1.4]{ziegler2012lectures}}]
\label{lem:farkas-lemma}
Consider a full dimensional polytope \( P \subset \real^{d} \) with facets \( \Gamma^{d-1}_{i} \) whose defining hyperplanes are \( H_{i}:=\setdef{x\in \real^{d}}{\trans{(\pi^{i})}x = \pi_{0}^{i}} \), and let \( \Gamma \subset \cap_{i \in S } \Gamma^{d-1}_{i} \) be a nonempty face of \( P \) for some index set \( S \). Then, \( \Gamma = \argmax_{x\in P} \trans{\mu} x \) for some \( \mu \in \real^{d} \) if and only if there exists \( \theta \in \spreal^{S} \) such that \( \mu  = \sum_{i \in S} \theta_{i} \pi^{i} \). 
\end{lemma} 

Let \( \tau(x):= \argmax_{\hat{x} \in P} \setdef{\hat{x}_{k}}{\hat{x}_{j} = x_{j}, \forall j\neq k}\) for all \( x \in P \), and let \( P^{t}: = \cup_{x\in P} \{ \tau(x) \} \). It follows from the definition that every \( x \in P \) is shaded by \( \tau(x) \in P^{t} \) in direction \( \unit_{k} \). We now show that \( P^{t} \) is included in \( \unit_{k}- \)top facets of \( P \). For each point \( \tilde{x} \in P^{t} \), \( P \cap (\{ \tilde{x} \} + (0, +\infty) \unit_{k} ) = \emptyset\), where the set \( \{\tilde{x}\} + (0, +\infty) \unit_{k} \) is the half open ray starting from \( \tilde{x} \) and pointing in the \( \unit_{k} \) direction. The separating hyperplane theorem, e.g., see \cite{Boyd.Vandenberghe:04}, then implies that there exist \( \mu \in \real^{d} \) and \( \mu_{0} \in \real \) such that  \( \trans{\mu} x \le \mu_{0} \ \) for all \( x \in P \) and \( \trans{\mu} x > \mu_{0} \ \) for all \( x \in \{ \tilde{x} \} + (0, +\infty)\unit_{k} \). The latter implies \( \mu_{k} = \trans{\mu} \unit_{k} > 0  \). In order to show that \( \tilde{x} \) is included in a \( \unit_{k}-\)top facet of \( P \), we consider two cases. First, if \( \tilde{x} \) is in the interior of some facet \( \Gamma^{d-1}_{i} \), then, since \( H_{i} \) is the unique separating hyperplane of \( \Gamma^{d-1}_{i} \), we get \( \pi^{i} = \mu \) and \( \pi^{i}_{k} > 0  \). Hence \( \Gamma^{d-1} \) is a \( \unit_{k}-\)top facet. Second, if \( \tilde{x} \) is in a lower dimensional face, then consider a non-empty set \( \mc J \) such that \( \tilde{x} \in \cap_{i\in \mc J} \Gamma_{i} \). By contradiction, from Lemma~\ref{lem:farkas-lemma}, there exists a \( j\in \mc J \) such that \( \pi^{j}_{k} >0 \). This implies that \( \tilde{x} \) belongs to the facet \( \Gamma_{i}^{d-1} \), a \( \unit_{k}-\)top facet. This establishes (i). 

For (ii), pick an arbitrary point \( \hat{x} \) from an arbitrary \( \unit_{k} \)-top facet \( \Gamma^{d-1}_{i} \). By definition of a facet, \( \hat{x} \in \argmax_{x\in P} \trans{(\pi^{i})} x \). This implies that \( \tau(\hat{x}) = \hat{x} \), i.e., \( \hat{x} \in P^{t} \). It is then straightforward to see that \( \Gamma^{d-1}_{i} \) remains to be a facet in \( \swp_{k}(P) \), since the hyperplane \( H_{i} \) contains \( \swp_{k}(P) \) on one side. 

We now prove (iii). Let \( \Gamma^{d-2} \) be an arbitrary ridge of \( P \) and \( \Gamma^{d-1}_{i} \) and \( \Gamma^{d-1}_{j} \) be the two incident facets of \( \Gamma^{d-2} \). We first prove the conditions to be necessary by considering the following cases: 
\begin{enumerate}[(a)]
\item if \(\Gamma^{d-2} \subset \bar{H} \), then \( \swp_{k}(\Gamma^{d-2}) = \Gamma^{d-2} \) is of dimension \( (d-2) \) and it is trivial that \(  \swp_{k}(\Gamma^{d-2})  \) is not a facet of \( \swp_{k}(P) \);
\item if neither  \( \Gamma^{d-1}_{i} \) nor \( \Gamma^{d-1}_{j} \) is a \( \unit_{k} \)-top facet, then all the points in \(  \Gamma^{d-2} \) are shaded by some other points in \( P \) and hence \( \swp_{k}(\Gamma^{d-2} ) \) is not a facet of \( \swp_{k}(P) \); 
\item if both \( \Gamma^{d-1}_{i} \) and \( \Gamma^{d-1}_{j} \) are \( \unit_{k} \)-top facets, then (ii) of the proposition implies that both \( \Gamma^{d-1}_{i} \) and \( \Gamma^{d-1}_{j} \) remains to be facets of \( \swp_{k}(P) \). Since a ridge is contained in exactly two facets\cite[Chapter 3]{grunbaum1967convex},  \( \swp_{k}( \Gamma^{d-2} ) \) can not be a facet of \( \swp_{k}(P) \). 
\end{enumerate}

We now prove the sufficient condition. Let \( \Gamma^{d-2} \not \in \bar{H} \) be a boundary ridge and, without loss of generality, let the incident facets be such that \( \Gamma^{d-1}_{i} \) is a \( \unit_{k} \)-top facet and \( \Gamma^{d-1}_{j} \) is either a \( \unit_{k} \)-facet or a \( \unit_{k} \)-bottom facet, i.e., \( \pi^{i}_{k} >0 \) and \( \pi^{j}_{k} \le 0 \). Proposition~\ref{prop:sweep-case-1} implies that \( \swp_{k}( \Gamma^{d-2}) \) is of dimension \( (d-1) \). Therefore, \(  \swp_{k}( \Gamma^{d-2})  \) is a facet if it is a face of \(  \swp_{k}(P)  \). We now construct the defining hyperplane \( H \) of \(  \swp_{k}( \Gamma^{d-2}) \). Let \( \theta := \pi_{k}^{i} /(\pi^{i}_{k} - \pi_{k}^{j}) \in (0, 1] \),  \( \pi := (1-\theta) \pi^{i} + \theta \pi^{j} \) and \( \pi_{0} := (1-\theta) \pi_{0}^{i} + \theta \pi_{0}^{j} \). It is straightforward that \( \pi_{k} = 0 \). Define \( H:=\setdef{x\in \real^{d}}{\trans{\pi} x = \pi_{0}} \). It is sufficient to show that  \( \swp_{k}(\Gamma^{d-2})  \subset H \) and \( H \) contains \( \swp_{k}(P) \) on one side. Since \( \pi_{k} = 0 \), \( \swp_{k}(\Gamma^{d-2})  \subset H \) if and only if \( \Gamma^{d-2}  \subset H \). The latter is straightforward from the definition of \( H \). In order to show that \( H \) contains \( \swp_{k}(P) \) on one side, we separately consider the two possibilities for $\pi^j_k$. If \( \pi^{j}_{k} < 0 \), and hence \( \theta \in (0,1) \), then Lemma~\ref{lem:farkas-lemma} implies the claim. If \(\pi^{j}_{k} = 0  \), and hence \( \theta = 1 \), then \( H = H_{j}\), the defining hyperplane of \( \Gamma^{d-1}_{j} \). This implies the claim trivially. 

\end{document}